\documentclass[a4paper,reqno]{amsart}
\usepackage{amsfonts}
\usepackage{tikz}
\usepackage{tkz-graph}
\usepackage{accents}
\usepackage{array,calc} 
\usepackage{mathrsfs}
\usepackage{stmaryrd} 
\usepackage{pgf}
\usetikzlibrary{shapes,cd,arrows,automata,decorations.pathreplacing,angles,quotes,calc}
\usepackage{tkz-euclide}
\usepackage{amsmath,calc,graphicx}
\usepackage{url}
\usepackage[hidelinks]{hyperref}
\usepackage{amssymb}
\usepackage{amsthm}
\usepackage{multirow}
\usepackage{float}
\usepackage{enumerate}
\usepackage{enumitem}
\usepackage{amsrefs}
\numberwithin{equation}{section}
\numberwithin{figure}{section}
\theoremstyle{plain}

\newtheorem{theorem}{Theorem}[section]
\newtheorem{lemma}[theorem]{Lemma}
\newtheorem{corollary}[theorem]{Corollary}
\newtheorem{proposition}[theorem]{Proposition}
\newtheorem{definition}[theorem]{Definition}

\theoremstyle{remark}
\newtheorem{remark}[theorem]{Remark}

\newtheorem*{lem*}{\textsc{Lemma}}
\newtheorem*{cor*}{\textsc{Corollary}}
\newtheorem*{exer*}{\textsc{Exercise}}
\newtheorem*{con*}{\textsc{Conjecture}}
\newtheorem*{thm*}{\textsc{Theorem}}

\newcommand{\beq}{\begin{equation}}
\newcommand{\eeq}{\end{equation}}

\newcommand\Psix{\textrm{P}_{\textrm{VI}}}

\newcommand\B[1]{\ensuremath{\scalebox{0.5}{$($}{#1}\scalebox{0.5}{$)$}}}

\newcommand{\orcidauthorA}{0000-0001-7504-4444}
\newcommand{\orcidauthorB}{0000-0002-0461-7580}

\title[On the monodromy manifold of $q\Psix$]{On the monodromy manifold of $q$-Painlev\'e VI and its Riemann-Hilbert problem}
\date{}
\author[Nalini Joshi]{Nalini Joshi}
\thanks{NJ's ORCID ID is \orcidauthorA. NJ's research was supported by an
  Australian Research Council Discovery Project \#DP210100129.
}
\address{School of Mathematics and Statistics F07, The University of Sydney, NSW 2006, Australia}
\author{Pieter Roffelsen}
\thanks{PR's ORCID ID is \orcidauthorB.}
\email{nalini.joshi@sydney.edu.au}
\email{pieter.roffelsen@sydney.edu.au}

\begin{document}
\begin{abstract}
We study the $q$-difference sixth Painlev\'e equation ($q\Psix$) through its associated Riemann-Hilbert problem (RHP) and show that the RHP is always solvable for irreducible monodromy data. This enables us to identify the solution space of $q\Psix$ with a monodromy manifold for generic parameter values. We deduce this manifold explicitly and show it is a smooth and affine algebraic surface when it does not contain reducible monodromy. Furthermore, we describe the RHP for reducible monodromy data and show that, when solvable, its solution is given explicitly in terms of certain orthogonal polynomials yielding special function solutions of $q\Psix$.
\end{abstract}
\maketitle

\tableofcontents

\section{Introduction}\label{sec:intro}

Despite widespread knowledge of how a Riemann-Hilbert formulation allow us to describe the solutions of the Painlev\'e equations, the corresponding description remains incomplete for discrete Painlev\'e equations. In this paper, we provide such a formulation for an important equation known as the $q$-difference sixth Painlev\'e equation and show that (under certain conditions) the corresponding Riemann-Hilbert problem is solvable, the resulting monodromy mapping is bijective, and the monodromy manifold is an algebraic surface given by an explicit equation.

Assuming $q\in\mathbb{C}$, $0<|q|<1$, and given nonzero parameters $\kappa=(\kappa_0,\kappa_t,\kappa_1,\kappa_\infty)\in\mathbb C^4$, the system known as the $q$-difference sixth Painlev\'e equation is
\begin{equation}\label{eq:qpvi}
q\Psix:\ \begin{cases}\  f\overline{f}&=\dfrac{(\overline{g}-\kappa_0\,t)(\overline{g}-\kappa_0^{-1}t)}{(\overline{g}-\kappa_\infty)(\overline{g}-q^{-1}\kappa_\infty^{-1})},\\
  \   g\overline{g}&=\dfrac{(f-\kappa_t\,t)(f-\kappa_t^{-1}t)}{q(f-\kappa_1)(f-\kappa_1^{-1})},
    \end{cases} 
\end{equation}
where $f,g:T\rightarrow \mathbb{CP}^1$ are complex functions defined on a domain $T$ invariant under multiplication by $q$ and we have used the abbreviated notation $f=f(t)$, $g=g(t)$, $\overline{f}=f(qt)$, $\overline{g}=g(qt)$, for $t\in T$. We will refer to Equation \eqref{eq:qpvi} by the abbreviation $q\Psix$.

 $q\Psix$ was first derived by Jimbo and Sakai \cite{jimbosakai} as the compatibility condition of a pair of linear $q$-difference systems. They showed that this formulation could be interpreted as a $q$-difference version of isomonodromic deformation, in close parallel to the role played by the classical sixth Painlev\'e equation as the isomonodromic condition for a rank-two Fuchsian system with four regular singular points at $0$, $1$, $\infty$, $t$, where $t$ is allowed to move in $\mathbb C\setminus \{0,1\}$ \cites{fuchs,jimbomiwauenoI}. 

The sixth Painlev\'e equation ($\Psix$) plays an important role in many settings in mathematics and physics. We mention the construction of self-dual Einstein metrics in general relativity \cite{tod}, classification of 2D-topological field theories \cite{dub}, mirror symmetry, and quantum cohomology \cite{manin} as noteworthy examples. 

Letting $q\rightarrow 1$ in $q\Psix$, with $\kappa_j=q^{k_j}$ for $j=0,t,1,\infty$, under the assumption that $f\rightarrow u$ and $g\rightarrow (u-t)/(u-1)$, the system reduces to $\Psix$:
\begin{align*}
u_{tt}=&\left(\frac{1}{u}+\frac{1}{u-1}+\frac{1}{u-t}\right)\frac{u_t^2}{2}-\left(\frac{1}{t}+\frac{1}{t-1}+\frac{1}{u-t}\right)u_t\\
&+\frac{u(u-1)(u-t)}{t^2(t-1)^2}\left(\alpha+\frac{\beta t}{u^2}+\frac{\gamma(t-1)}{(u-1)^2}+\frac{\delta t(t-1)}{(u-t)^2}\right),
\end{align*}
where
\begin{equation*}
    \alpha=\frac{(2 k_\infty+1)^2}{2},\quad \beta=-2 k_0^2,\quad \gamma =2k_1^2,\quad \delta=\frac{1-4 k_t^2}{2}.
\end{equation*}

Due to its relation to $\Psix$, the $q$-difference equation $q\Psix$ has drawn increasing interest in recent times. Mano \cite{manoqpvi} derived the generic leading order asymptotics of solutions near $t=0$ and $t=\infty$ and gave an implicit solution to the corresponding nonlinear connection problem.  Jimbo et al. \cite{jimbonagoyasakai} extended Mano's asymptotic result near $t=0$ to an explicit asymptotic expansion beyond all orders for the generic solution. They obtained this asymptotic representation through an interesting connection of $q\Psix$ with conformal field theory, analogous to the one for $\Psix$ established by Gamayun et al. \cite{lisovyy2012}.

In this paper, we study $q\Psix$ via the Jimbo-Sakai linear problem \cite{jimbosakai}. Using Birkhoff's theory \cite{birkhoffgeneralized1913}, we define an associated Riemann-Hilbert problem (RHP), which captures the general solution of $q\Psix$. The jump matrices of this RHP across a single closed contour form a corresponding monodromy manifold that is a focal point of this paper.

Recently, this monodromy manifold was the object of an extensive study by Ohyama et al. \cite{ohyamaramissualoy}, who showed that such a manifold forms an algebraic surface. Furthermore, they conjectured, see \cite{ohyamaramissualoy}[Conjecture 7.10], that the algebraic surface is smooth, under additional conditions.
In this paper, we prove a stronger version of this conjecture,
see Theorem \ref{thm:main_smooth} and Remark \ref{remark:ohyamaetal}.

Consider the general class of solutions $(f,g)$ of $q\Psix$ defined on a domain $T$ given by a discrete $q$-spiral, i.e., $T=q^\mathbb{Z}t_0$, for some $t_0\in\mathbb{C}^*$.  The deformation of the Jimbo-Sakai linear problem (see \S \ref{subsec:isomonodromic_deformation}) yields an auxiliary equation associated with $q\Psix$
\begin{equation}\label{eq:auxiliary}
    \frac{\overline{w}}{w}=\kappa_\infty\frac{q \kappa_\infty\overline{g}-1}{\overline{g}-\kappa_\infty}.
\end{equation}
We refer to $(f,g)$ as a solution of $q\Psix(\kappa,t_0)$ and call the triplet $(f,g,w)$ a solution of $q\Psix^{\text{aux}}(\kappa,t_0)$.

Starting with an initial value of $(f,g)$ in $\mathbb{C}^*\times \mathbb{C}^*$, and iterating in $t$,  $q\Psix$ can become apparently singular when $(f,g)$ takes the value of one of the following eight base-points,
\begin{equation}
    \begin{aligned}
    &b_1=(0,q^{-1}\kappa_0^{+1}t), & &b_3=(\kappa_t^{+1}t,0), & &b_5=(\kappa_1^{+1},\infty), & &b_7=(\infty,\kappa_\infty^{+ 1}),\\
    &b_2=(0,q^{-1}\kappa_0^{-1}t), & &b_4=(\kappa_t^{-1}t,0), & &b_6=(\kappa_1^{-1},\infty), & &b_8=(\infty,\kappa_\infty^{-1}q^{-1}).
\end{aligned}\label{eq:basepoints}
\end{equation}
Each of these can be resolved through a blow up, so that the iteration is once again well-defined \cite{s:01}. There are, however, formal solutions of equations \eqref{eq:qpvi}, which never take a value in $\mathbb{C}^*\times \mathbb{C}^*$. We exclude such solutions from our consideration.

\subsection{Main results}
The main results of this paper are given by Theorems \ref{thm:main_solvability}, \ref{thm:main_moduli}, \ref{thm:main_smooth} and \ref{thm:main_affine} in Section \ref{sec:results}. Throughout the paper, it is assumed that the parameters $\kappa$ and $t_0$ satisfy the {\em non-resonance} conditions,
\begin{equation}\label{eq:non_res}
\kappa_0^2,\kappa_t^2,\kappa_1^2,\kappa_\infty^2\notin q^\mathbb{Z},\qquad (\kappa_t\kappa_1)^{\pm 1},
(\kappa_t/\kappa_1)^{\pm 1}\notin t_0q^\mathbb{Z}.
\end{equation}
As in Ohyama et al. \cite{ohyamaramissualoy}, the {\em non-splitting} conditions
\begin{subequations}\label{eq:intro_irreducibleparameter}
\begin{align}
&\kappa_0^{\epsilon_0}\kappa_t^{\epsilon_t} \kappa_1^{\epsilon_1} \kappa_\infty^{\epsilon_\infty}\notin q^\mathbb{Z},\label{eq:intro_irreducibleparameter1}\\
&\kappa_0^{\epsilon_0} \kappa_\infty^{\epsilon_\infty}\notin t_0q^\mathbb{Z},\label{eq:intro_irreducibleparameter2}
\end{align}
\end{subequations}
where $\epsilon_j\in\{\pm 1\}$, $j=0,t,1,\infty$, also play an important role. The monodromy manifold contains reducible monodromy when one or more of these conditions are violated -- see Lemma \ref{lem:irreducible}.

The RHP corresponding to $q\Psix$ is given by Definition \ref{def:RHPmain}. Our first main result, Theorem \ref{thm:main_solvability}, shows that the RHP with irreducible monodromy is always solvable. This has important ramifications for the mapping that sends solutions of $q\Psix$ to points on the monodromy manifold, which we will refer to as the {\em monodromy mapping}. In particular, Corollary \ref{cor:monodromy_mapping} shows that the monodromy mapping is bijective when the non-splitting conditions are satisfied. 

The RHP may be solvable in some cases of reducible monodromy. In Section \ref{subsection:reducible}, we show that in such cases, the RHP is solved explicitly in terms of certain orthogonal polynomials yielding special function solutions of $q\Psix$.

Our second main result, Theorem \ref{thm:main_moduli}, constructs an embedding of the monodromy manifold into $(\mathbb{CP}^1)^4/\mathbb{C}^*$,
 where the quotient is taken with respect to scalar multiplication. The image of this embedding is described as the zero set of a polynomial, given explicitly in Definition \ref{def:modulispace}, minus a curve.

 This embedding allows us to study algebro-geometric properties of the monodromy manifold. Our third main result,  Theorem \ref{thm:main_smooth}, focuses on the singularities of the monodromy manifold and proves that it is smooth if and only if it excludes reducible monodromy, i.e., if and only if the non-splitting conditions hold true.

Finally, our fourth main result, Theorem \ref{thm:main_affine}, identifies the monodromy manifold with an explicit affine algebraic surface when the non-splitting conditions are satisfied.

\subsection{Notation}\label{sec:not}
Here, we briefly describe the notation used in this paper. The symbol $\sigma_3$ is the well-known Pauli matrix $\sigma_3=\operatorname{diag}(1,-1)$.
The $q$-Pochhammer symbol is the (convergent) product
\begin{equation*}
(z;q)_\infty=\prod_{k=0}^{\infty}{(1-q^kz)}\qquad (z\in\mathbb{C}).
\end{equation*}
Note that the entire function $(z;q)_\infty$ satisfies
\begin{equation*}
(qz;q)_\infty=\frac{1}{1-z}(z;q)_\infty,
\end{equation*}
with $(0;q)_\infty=1$ and, moreover, possesses simple zeros at $q^{-\mathbb{N}}$. The $\mathit{q}$-theta function 
\begin{equation}\label{eq:thetasym}
\theta_q(z)=(z;q)_\infty(q/z;q)_\infty\quad (z\in \mathbb{C}^*),\quad \mathbb{C}^*:=\mathbb{C}\setminus\{0\},
\end{equation}
is analytic on $\mathbb{C}^*$, with essential singularities at $z=0, \infty$, and has simple zeros on the $q$-spiral $q^\mathbb{Z}$. It satisfies
\begin{equation}\label{eq:qtheta_identities}
\theta_q(qz)=-\frac{1}{z}\theta_q(z)=\theta_q(1/z).
\end{equation}
For $n\in\mathbb{N}^*$, we use the common abbreviation for repeated products of these functions 
\begin{align*}
\theta_q(z_1,\ldots,z_n)&=\theta_q(z_1)\cdot \ldots\cdot \theta_q(z_n),\\
(z_1,\ldots,z_n;q)_\infty&=(z_1;q)_\infty\cdot\ldots\cdot (z_n;q)_\infty.
\end{align*}

We will refer to the complex projective space $\mathbb C\mathbb P^1$ as $\mathbb P^1$ and, for positive integer $k$, denote the $k$-fold direct product $\mathbb P^1\times \ldots \times\mathbb P^1$ by $(\mathbb P^1)^k$. (We remind the reader that $\mathbb P^1\times \mathbb P^1$ is not the same space as $\mathbb P^2$.)

\subsection{Outline of the paper}
In Section \ref{sec:results}, we give the precise statements of the main results of the paper.
Section \ref{sec:linear_problem} is devoted to the Jimbo-Sakai linear system. Here, we renormalize the linear system of \cite{jimbosakai}
and describe the outcomes of Birkhoff's classical theory \cite{birkhoffgeneralized1913} for this system.
In Section \ref{sec:solvability}, we study the solvability of RHP I, defined in Definition \ref{def:RHPmain}, and prove Theorem \ref{thm:main_solvability}. Section \ref{sec:monodromy_surface} concerns the monodromy manifold and proofs of Theorems \ref{thm:main_moduli}, \ref{thm:main_smooth} and \ref{thm:main_affine} are given there.
 We conclude the paper with a conclusion in Section \ref{sec:conc}.

\subsection{Acknowledgments}
The authors thank Peter Forrester, Marta Mazzocco, Yousuke Ohyama and Andrea Ricolfi for stimulating discussions about topics related to the work presented in this paper.


\section{Detailed Statement of Results}\label{sec:results}
In order to state our main results, we recall the Jimbo-Sakai linear problem for $q\Psix$ and define the corresponding monodromy manifold and mapping in Section \ref{sec:linear_system}. In Section \ref{sec:rhp}, we formulate the associated RHP via Birkhoff's theory. In Section \ref{subsec:solvabilitymain} we state our first main result, Theorem \ref{thm:main_solvability}. Then, in Section \ref{sec:results_monodromy_surface}, we state our main results on the monodromy manifold, that is, Theorems \ref{thm:main_moduli}, \ref{thm:main_smooth} and \ref{thm:main_affine}.

\subsection{The Jimbo-Sakai linear system}\label{sec:linear_system}
Suppose $\kappa=(\kappa_0,\kappa_t,\kappa_1,\kappa_\infty)\in\mathbb C^4$, all nonzero, are given and $t\in T$ lies on a discrete $q$-spiral $T=q^\mathbb{Z}t_0$. Consider the linear system
\begin{align}
    Y(qz)&=A(z,t)Y(z),\label{eq:linear_problem}\\
    A(z,t)&=A_0(t)+z A_1(t)+z^2 A_2,\label{eq:matrix_polynomial}
\end{align}
where $A(z,t)$ is a $2\times 2$ matrix polynomial with determinant given by
\begin{equation}\label{eq:detA}
|A(z,t)|=(z-\kappa_t^{+1}t)(z-\kappa_t^{-1}t)(z-\kappa_1^{+1})(z-\kappa_1^{-1}),
\end{equation}
and assume that
\begin{equation}
    A_0(t)=H(t)\begin{pmatrix}
    \kappa_0^{+1}t & 0\\
    0 & \kappa_0^{-1}t \end{pmatrix}H(t)^{-1},\quad  A_2=\begin{pmatrix}
    \kappa_\infty^{+1} & 0\\
    0 & \kappa_\infty^{-1} \end{pmatrix}.\label{eq:diagonal}
\end{equation}
for an $H=H(t)\in GL_2(\mathbb{C})$. This is the Jimbo-Sakai linear problem \cite{jimbosakai}, which we have scaled to remove redundant parameters (see Section \ref{subsec:normalising_linear} for details).
Throughout this paper we assume that the parameters $\kappa$ and $t_0$ satisfy the non-resonance conditions \eqref{eq:non_res}, which ensure that the linear problem is fully non-resonant (see \cite{joshiroffelseniv}[Definition 1.1]).

By Carmichael \cite{carmichael1912},  the linear system \eqref{eq:linear_problem} has solutions $Y_0(z,t)$ and $Y_\infty(z,t)$ respectively given by convergent series expansions around $z=0$ and $z=\infty$ of the following form,
\begin{subequations}\label{eq:linear_sys_solutions}
\begin{align}
Y_0(z,t)&=z^{\log_q(t)}\Psi_0(z,t)z^{k_0\sigma_3},& \Psi_0(z,t)&=H(t)+\sum_{n=1}^\infty z^n M_n(t),\\
Y_\infty(z,t)&=z^{\log_q(z)-1}\Psi_\infty(z,t) z^{k_\infty \sigma_3},& \Psi_\infty(z,t)&=I+\sum_{n=1}^\infty z^{-n} N_n(t),
\end{align}
\end{subequations}
where $q^{k_j}=\kappa_j$ for $j=0,\infty$. The matrix functions $\Psi_\infty(z,t)$ and $\Psi_0(z,t)^{-1}$ extend to single-valued analytic functions in $z$ on $\mathbb{P}^1\setminus\{0\}$ and $\mathbb{C}$ respectively. Furthermore, their determinants are explicitly given by
\begin{subequations}\label{eq:solutions_det}
\begin{align}
            &|\Psi_\infty(z,t)|=\left(\kappa_t^{+1}\frac{qt}{z},\kappa_t^{-1}\frac{qt}{z},\kappa_1^{+1}\frac{q}{z},\kappa_1^{-1}\frac{q}{z};q\right)_\infty,\label{eq:solutions_det1}\\
    &|\Psi_0(z,t)|^{-1}=|H|^{-1}\left(\kappa_t^{+1}\frac{z}{t},\kappa_t^{-1}\frac{z}{t},\kappa_1^{+1}z,\kappa_1^{-1}z;q\right)_\infty.\label{eq:solutions_det2}
\end{align}
\end{subequations}

A central object of study in this paper is the {\em connection matrix}
\begin{equation*}
    C(z,t):=\Psi_0(z,t)^{-1}\Psi_\infty(z,t).
\end{equation*}
This matrix is single-valued in $z$ on $\mathbb{C}^*$ and is related to Birkhoff's connection matrix
\begin{equation*}
    P(z,t):=Y_0(z,t)^{-1}Y_\infty(z,t),
\end{equation*}
by
\begin{equation*}
    P(z,t)=z^{\log_q(z/qt)} z^{-k_0\sigma_3} C(z,t)z^{k_\infty\sigma_3}.
\end{equation*}
For our purposes, it is more convenient to work with $C(z,t)$, rather than $P(z,t)$, due to its single-valuedness.
 We will also refer to the connection matrix $C(z,t)$ as the {\em monodromy} of the linear system \eqref{eq:linear_problem}.

 For any fixed $t$, $C(z,t)$ has the following analytic characterisation in $z$.
\begin{enumerate}
    \item It is a single-valued analytic function in $z\in\mathbb{C}^*$.
    \item It satisfies the $q$-difference equation
    \begin{equation*}
        C(qz,t)=\frac{t}{z^2}\kappa_0^{\sigma_3}C(z,t)\kappa_\infty^{-\sigma_3}.
    \end{equation*}
    \item Its determinant is given by
    \begin{equation*}
    |C(z,t)|=c\, \theta_q\left(\kappa_t^{+1}\frac{z}{t},\kappa_t^{-1}\frac{z}{t},\kappa_1^{+1}z,\kappa_1^{-1}z\right),
    \end{equation*}
    for some $c\in\mathbb{C}^*$.
\end{enumerate}
We correspondingly make the following definition.
\begin{definition}\label{def:connection_matrix_space}
We denote by $\mathfrak{C}(\kappa,t)$, for any fixed $t\in\mathbb{C}^*$, the set of all $2\times 2$ matrix functions satisfying properties (1)-(3) above.
\end{definition}

Next, we consider deformations of the linear system \eqref{eq:linear_problem}, as $t\rightarrow qt$, which leave the matrix function $P(z,t)$ invariant, i.e. such that $P(z,qt)=P(z,t)$, which is equivalent to
\begin{equation*}
    C(z,qt)=z \,C(z,t). 
\end{equation*}
We call such a deformation isomonodromic.

Jimbo and Sakai \cite{jimbonagoyasakai} showed that, upon introducing the following coordinates\footnote{See equations \eqref{eq:param1} for a full parametrisation of A with respect to $\{f,g,w\}$.} $(f,g,w)$ on $A$,
\begin{subequations}\label{eq:coordinates_linear}
\begin{align}
    A_{12}(z,t)&=\kappa_\infty^{-1} w(z-f),\\
    A_{22}(f,t)&=q(f-\kappa_1)(f-\kappa_1^{-1})g,
\end{align}
\end{subequations}
isomonodromic deformation of the linear system \eqref{eq:linear_problem} is locally equivalent to $(f,g,w)$ satisfying $q\Psix^{aux}(\kappa,t_0)$. Building on this, we prove the following lemma in Section \ref{subsec:isomonodromic_deformation}.
\begin{lemma}\label{lem:isomonodromic}
Let $(f,g,w)$ be any solution of $q\Psix^{aux}(\kappa,t_0)$ and denote
\begin{equation}
\mathfrak{M}=\left\{m\in\mathbb{Z}:(f(q^mt_0),g(q^mt_0))\neq (\infty,\kappa_\infty)\right\}.
\end{equation}
Then, the linear system $A(z,t)$ is regular in $t$ on $q^\mathfrak{M} t_0$ and the corresponding connection matrix is given by
\begin{equation}\label{eq:C0definition}
    C(z,t)=z^m D(t)C_0(z),\quad (t=q^mt_0,m\in\mathfrak{M}),
\end{equation}
for a matrix $C_0(z)\in \mathfrak{C}(\kappa,t_0)$, unique up the left-multiplication by diagonal matrices. Here $D(t)$ is a diagonal matrix which may be eliminated from equation \eqref{eq:C0definition} by rescaling $H(t)\mapsto H(t)D(t)$ in equation \eqref{eq:diagonal}.
\end{lemma}
In Lemma \ref{lem:isomonodromic}, we have the freedom of rescaling the auxiliary variable $w$ by $w\mapsto \widetilde{w}=dw$, $d\in\mathbb{C}^*$, which is equivalent to gauging the linear system by a constant diagonal matrix,
\begin{equation*}
    A(z,t)\rightarrow D^{-1}A(z,t)D,\quad D=\begin{pmatrix}
    1 & 0\\
    0 & d\\
    \end{pmatrix},
\end{equation*}
and thus rescaling the matrix $C_0(z)\in \mathfrak{C}(\kappa,t_0)$ as
\begin{equation*}
    C_0(z)\rightarrow C_0(z)D.
\end{equation*}

Hence, Lemma \ref{lem:isomonodromic} provides us with a mapping
\begin{equation}\label{eq:monodromy_mapping}
    (f,g)\rightarrow [C_0(z)],
\end{equation}
which associated to any solution $(f,g)$ of $q\Psix(\kappa,t_0)$ the equivalence class of $C_0(z)$ in $\mathfrak{C}(\kappa,t_0)$ quotiented by arbitrary left and right-multiplication by invertible diagonal matrices.
This warrants the following definition.

\begin{definition}\label{def:monodromy_mapping}
We define $\mathcal{M}(\kappa,t_0)$ to be the space of connection matrices $\mathfrak{C}(\kappa,t_0)$ quotiented by arbitrary left and right-multiplication by invertible diagonal matrices.
We refer to $\mathcal{M}(\kappa,t_0)$ as the monodromy manifold of $q\Psix(\kappa,t_0)$.

Correspondingly, we call the mapping \eqref{eq:monodromy_mapping}, which associates with any solution $(f,g)$ of $q\Psix(\kappa,t_0)$, a point on the monodromy manifold, the monodromy mapping.
\end{definition}
\begin{remark}
The space $\mathcal{M}(\kappa,t_0)$ was first introduced and studied in Ohyama et al. \cite{ohyamaramissualoy}[\textsection 4.1.1], where it is denoted as $\mathcal{F}$. Ohyama et al. \cite{ohyamaramissualoy} showed how this space can naturally be endowed with the structure of a complex algebraic variety, under certain assumptions of genericity including the non-resonance \eqref{eq:non_res} and non-splitting conditions \eqref{eq:intro_irreducibleparameter}. Compatible with this structure, we endow $\mathcal{M}(\kappa,t_0)$ with the structures of a complex manifold and algebraic variety, in Theorems \ref{thm:main_smooth} and \ref{thm:main_affine} respectively. The proof that these structures are compatible with those in \cite{ohyamaramissualoy} is postponed to the end of the paper, see Remark \ref{remark:compatible}.
\end{remark}

In Section \ref{sec:isorhp}, we prove the following lemma concerning injectivity of the monodromy mapping.
\begin{lemma}\label{lem:monodromy_mapping_injective}
The monodromy mapping, defined in Definition \ref{def:monodromy_mapping}, is injective.
\end{lemma}

\subsection{The main Riemann-Hilbert problem}\label{sec:rhp}
In this paper, we analyse the monodromy mapping through the, via Birkhoff's theory \cite{birkhoffgeneralized1913}, corresponding Riemann-Hilbert problem (RHP).

To introduce this RHP, we return to the single-valued matrix functions $\Psi_0(z,t)$ and $\Psi_\infty(z,t)$, defined in equations \eqref{eq:linear_sys_solutions}.
Let us denote $t_m=q^mt_0$ for $m\in\mathbb{Z}$. By Lemma \ref{lem:isomonodromic}, we may choose $H$ such that
\begin{equation}\label{eq:pre_jump_condition}
\Psi_\infty(z,t_m)=\Psi_0(z,t_m)\,z^mC_0(z),
\end{equation}
for $m\in\mathfrak{M}$.

Next, we need to choose Jordan curves $\gamma^{\B{m}}$, $m\in\mathbb{Z}$, which separate the points in the complex plane where $\Psi_\infty(z,t_m)$ and $\Psi_0(z,t_m)$ are respectively non-invertible and singular. These points are precisely the zeros of the determinants \eqref{eq:solutions_det1} and \eqref{eq:solutions_det2} respectively. We thus make the following definition.
\begin{definition}\label{def:contours}
Consider a family $(\gamma^{\B{m}})_{m\in\mathbb{Z}}$ of positively oriented Jordan curves in $\mathbb{C}^*$ and denote by $D_+^{\B{m}}$ and $D_-^{\B{m}}$ the inside and outside of $\gamma^{\B{m}}$ respectively, for $m\in\mathbb{Z}$.
Then we call this family of curves admissable if, for $m\in \mathbb{Z}$,
\begin{align*}
q^{\mathbb{Z}_{>0}}\cdot \{\kappa_t t_m,\kappa_t^{-1} t_m,\kappa_1,\kappa_1^{-1}\}&\subseteq D_-^{\B{m}},\\
q^{\mathbb{Z}_{\leq 0}}\cdot \{\kappa_t t_m,\kappa_t^{-1} t_m,\kappa_1,\kappa_1^{-1}\}&\subseteq D_+^{\B{m}}, 
\end{align*}
where we use the notation $U\cdot V=\{uv:u\in U,v\in V\}$ for compatible sets $U$ and $V$,
and
    \begin{equation*}
     D_-^{\B{m+1}}\subseteq D_-^{\B{m}},
    \end{equation*}
 see Figure \ref{fig:contours}.
\end{definition}

\begin{figure}[ht]
	\centering
	\begin{tikzpicture}[scale=0.8]
	\draw[->] (-6,0)--(6,0) node[right]{$\Re{z}$};
	\draw[->] (0,-5)--(0,5) node[above]{$\Im{z}$};
	\tikzstyle{star}  = [circle, minimum width=3.5pt, fill, inner sep=0pt];
	\tikzstyle{starsmall}  = [circle, minimum width=3.5pt, fill, inner sep=0pt];

	\draw[domain=-1.1:6,smooth,variable=\x,red] plot ({exp(-\x*ln(2))*2*5/3*cos((5*pi/4+\x*pi/16) r)},{exp(-\x*ln(2))*2*5/3*sin((5*pi/4+\x*pi/16) r)});	
    \draw[domain=-1.1:6,smooth,variable=\x,red] plot ({exp(-\x*ln(2))*2*5/3*cos((-3*pi/8+5*pi/4+\x*pi/16) r)},{exp(-\x*ln(2))*2*5/3*sin((-3*pi/8+5*pi/4+\x*pi/16) r)});

	\draw[domain=-1.3:6,smooth,variable=\x,red] plot ({exp(-\x*ln(2))*2*4/3*cos((pi/4+\x*pi/16) r)},{exp(-\x*ln(2))*2*4/3*sin((pi/4+\x*pi/16) r)});	
	\draw[domain=-1.3:6,smooth,variable=\x,red] plot ({exp(-\x*ln(2))*2*4/3*cos((-pi/4+pi/8+\x*pi/16) r)},{exp(-\x*ln(2))*2*4/3*sin((-pi/4+pi/8+\x*pi/16) r)});
	
	\draw [blue,thick,decoration={markings, mark=at position 0.6 with {\arrow{<}}},
	postaction={decorate}] plot [smooth cycle,tension=0.6] coordinates {(4.4,0) (3.5,-3) (0,-3.9) (-3.9,-3.5) (-5,0) (-4.3,3.2) (0,3.9) (3.5,3) };
	
	\draw [blue,dashed,thick,decoration={markings, mark=at position 0.53 with {\arrow{<}}},
	postaction={decorate}] plot [smooth cycle,tension=0.6] coordinates {(3.9,0) (3.2,-2.7) (0,-3.5) (-2.4,-2.4) (-3.1,0) (-2.7,1.7) (-1.3,2.9) (0.5,3.15) (3.1,2.7) };
	
	\node[starsmall]     (or) at ({0},{0} ) {};
	\node     at ($(or)+(-0.3,0.3)$) {$0$};
	
	\node[starsmall]     (qk1) at ({sqrt(sqrt(2))*2*4/3*cos((pi/4-1/4*pi/16) r)},{sqrt(sqrt(2))*2*4/3*sin((pi/4-1/4*pi/16) r)} ) {};
	\node     at ($(qk1)+(-0.3,0.3)$) {$q\kappa_1$};
	\node[star]     (k1) at ({2*2*4/3*cos((pi/4-pi/16) r)},{2*2*4/3*sin((pi/4-pi/16) r)} ) {};
	\node     at ($(k1)+(-0.3,0.3)$) {$\kappa_1$};
	
	\node[starsmall]     (qkm1) at ({sqrt(sqrt(2))*2*4/3*cos((-pi/4+pi/8-1/4*pi/16) r)},{sqrt(sqrt(2))*2*4/3*sin((-pi/4+pi/8-1/4*pi/16) r)} ) {};
	\node     at ($(qkm1)+(0.35,0.35)$) {$q\kappa_1^{-1}$};
	\node[star]     (km1) at ({2*2*4/3*cos((pi/4-pi/16) r)},{-2*2*4/3*sin((pi/4-pi/16) r)} ) {};
	\node     at ($(km1)+(0.35,0.35)$) {$\kappa_1^{-1}$};
	
	\node[star]     (qqkt) at ({sqrt(2)*5/3*cos((-3*pi/8+5*pi/4+2/4*pi/16) r)},{sqrt(2)*5/3*sin((-3*pi/8+5*pi/4+2/4*pi/16) r)} ) {};
	\node     at ($(qqkt)+(0.5,0.35)$) {$q^2\kappa_t t$};
	\node[star]     (qkt) at ({sqrt(sqrt(2))*2*5/3*cos((-3*pi/8+5*pi/4-1/4*pi/16) r)},{sqrt(sqrt(2))*2*5/3*sin((-3*pi/8+5*pi/4-1/4*pi/16) r)} ) {};
	\node     at ($(qkt)+(0.4,0.35)$) {$q\kappa_t t$};	
	\node[star]     (kt) at ({2*2*5/3*cos((-3*pi/8+5*pi/4-pi/16) r)},{2*2*5/3*sin((-3*pi/8+5*pi/4-pi/16) r)} ) {};
	\node     at ($(kt)+(0.35,0.35)$) {$\kappa_t t$};

	\node[star]     (qqktm) at ({sqrt(2)*5/3*cos((5*pi/4+2/4*pi/16) r)},{sqrt(2)*5/3*sin((5*pi/4+2/4*pi/16) r)} ) {};
	\node     at ($(qqktm)+(+0.5,-0.4)$) {$q^2\kappa_t^{-1}t$};
	\node[star]     (qktm) at ({sqrt(sqrt(2))*2*5/3*cos((5*pi/4-1/4*pi/16) r)},{sqrt(sqrt(2))*2*5/3*sin((5*pi/4-1/4*pi/16) r)} ) {};
	\node     at ($(qktm)+(0.35,-0.35)$) {$q\kappa_t^{-1}t$};
	\node[star]     (ktm) at ({2*2*5/3*cos((5*pi/4-pi/16) r)},{2*2*5/3*sin((5*pi/4-pi/16) r)} ) {};
	\node     at ($(ktm)+(0.4,-0.35)$) {$\kappa_t^{-1}t$};

	\node[blue]     at (-1.5,4.3) {$\boldsymbol{\gamma^{\B{m}}}$};
	
	\node[blue]     at (-1.4,3.15) {$\boldsymbol{\gamma^{\B{m+1}}}$};
	
	
	
	\end{tikzpicture}
	\caption{An example of two contours $\gamma^{\B{m}}$ and $\gamma^{\B{m+1}}$ satisfying the conditions in Definition \ref{def:contours}, where $t=q^mt_0$ and the red lines denote the four spirals $q^{\mathbb{R}}\cdot x$, $x\in\{\kappa_t^{\pm 1}t_0,\kappa_1^{\pm 1}\}$.}
	\label{fig:contours}
\end{figure}
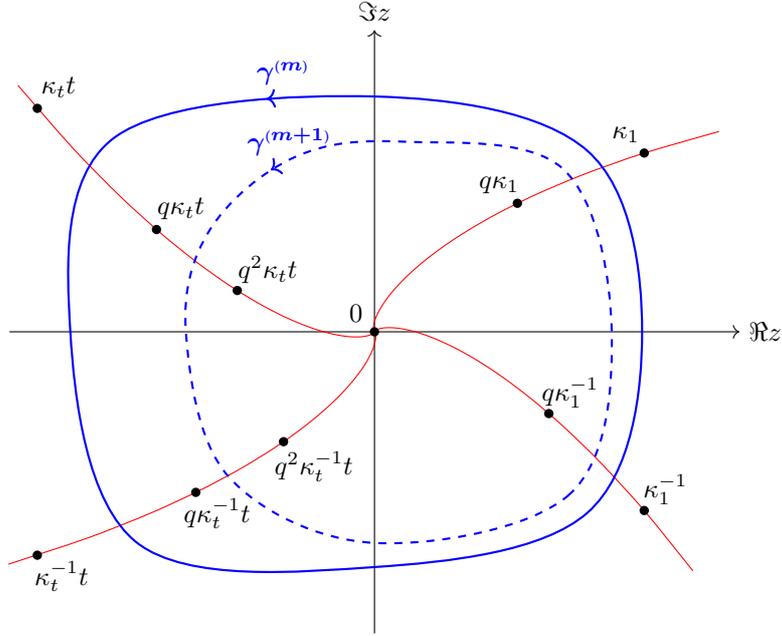

We can always construct an admissible family of curves and it follows that
\begin{equation}\label{eq:rhpsolution}
    \Psi^{\B{m}}(z)=\begin{cases}
    \Psi_\infty(z,t_m) & z\in D_+^{\B{m}},\\
    \Psi_0(z,t_m) & z\in D_-^{\B{m}},
    \end{cases}
\end{equation}
defines a solution of the following RHP, with $C(z)=C_0(z)$, for $m\in\mathfrak{M}$.
\begin{definition}[RHP I]\label{def:RHPmain}
Given a connection matrix $C\in \mathfrak{C}(\kappa,t_0)$ and a family of admissable curves $(\gamma^{\B{m}})_{m\in\mathbb{Z}}$,
 for $m\in\mathbb{Z}$, find a matrix function $\Psi^{\B{m}}(z)$ which satisfies the following conditions.
  \begin{enumerate}[label={{\rm (\roman *)}}]
  \item $\Psi^{\B{m}}(z)$ is analytic on $\mathbb{C}\setminus\gamma^{\B{m}}$.
    \item $\Psi^{\B{m}}(z')$ has continuous boundary values $\Psi_+^{\B{m}}(z)$ and $\Psi_-^{\B{m}}(z)$ as $z'$ approaches $z\in \gamma^{\B{m}}$ from $D_+^{\B{m}}$ and $D_-^{\B{m}}$ respectively, related by
		\begin{equation*}
		\Psi_+^{\B{m}}(z)=\Psi_-^{\B{m}}(z)z^mC(z),\quad z\in \gamma^{\B{m}}.
              \end{equation*}
            \item $\Psi^{\B{m}}(z)$ satisfies
              \begin{equation*}
		\Psi^{\B{m}}(z)=I+\mathcal{O}\left(z^{-1}\right)\quad z\rightarrow \infty.
              \end{equation*}
              \end{enumerate}
 \end{definition}
The matrix function $\Psi^{\B{m}}(z)$, defined in equation \eqref{eq:rhpsolution}, is uniquely characterised as the solution of RHP I. Indeed, we have the following lemma, which we prove in Section \ref{sec:isorhp}.
 \begin{lemma}\label{lem:uniqueness}
For any fixed $m\in\mathbb{Z}$, if RHP I in Definition \ref{def:RHPmain} has a solution $\Psi^{\B{m}}(z)$, then this solution is globally invertible on the complex plane and unique.
\end{lemma}
From here on we say that $\Psi^{\B{m}}(z)$ exists if and only if RHP I has a solution for that particular value of $m$, as justified by the uniqueness in the above lemma.

If RHP I is solvable, then we can construct a corresponding isomonodromic linear system, by setting
\begin{equation}\label{eq:rhptolinear}
    A(z,q^mt_0):=\begin{cases}
    z^2 \Psi^{\B{m}}(qz)\kappa_\infty^{\sigma_3} \Psi^{\B{m}}(z)^{-1} & \text{if } z\in q^{-1}(D_+^{\B{m}}\cup \gamma^{\B{m}}),\\
    q^m t_0 \Psi^{\B{m}}(qz)\kappa_0^{\sigma_3}C(z) \Psi^{\B{m}}(z)^{-1}
    & \text{if } z\in D_+^{\B{m}}\cap q^{-1}D_-^{\B{m}},\\
    q^m t_0 \Psi^{\B{m}}(qz)\kappa_0^{\sigma_3} \Psi^{\B{m}}(z)^{-1}
    & \text{if } z\in D_-^{\B{m}}\cup \gamma^{\B{m}}.
    \end{cases}
\end{equation}
This defines a matrix polynomial of the form \eqref{eq:matrix_polynomial} and the values of $(f,g,w)$ may be read directly from the solution of the RHP as follows (details are given in Section \ref{sec:isorhp}). Let
\begin{align}
\Psi^{\B{m}}(z)&=H(t_m)+\mathcal{O}(z) & &(z\rightarrow 0), \label{eq:defiH}\\
    \Psi^{\B{m}}(z)&=I+z^{-1}U(t_m)+\mathcal{O}(z^{-2}) & &(z\rightarrow \infty),\label{eq:defiU}
\end{align}
and denote $H=(h_{ij})$ and $U=(U_{ij})$, then
\begin{subequations}\label{eq:rhptofgw}
\begin{align}
    w&=(q^{-1}-\kappa_\infty^2)u_{12},\label{eq:rhptofgw:w}\\
    f&=t_m\kappa_\infty\left(\kappa_0-\kappa_0^{-1}\right)\frac{h_{11}h_{12}}{w|H|},\label{eq:rhptofgw:f}\\
    g&=q^{-1}\kappa_\infty^{-1}(f-\kappa_t t_m)(f-\kappa_t^{-1} t_m)g_1^{-1},\label{eq:rhptofgw:g}\\
    g_1&=f^2+f\left((q^{-1}-1)u_{11}+\frac{h_{21}w}{h_{11}\kappa_\infty^2}\right)+\frac{\kappa_0 t}{\kappa_\infty}.\label{eq:rhptofgw:g1}
\end{align}
\end{subequations}

\subsection{Solvability of the main RHP}\label{subsec:solvabilitymain}
The notion of reducible monodromy, given in the following definition, plays an important role in our main results.
\begin{definition}\label{def:reducible}
We call a connection matrix $C(z)\in \mathfrak{C}(\kappa,t_0)$ irreducible when none of its entries are identically zero, otherwise we call it reducible. Similarly, we call monodromy  $[C(z)]\in\mathcal{M}(\kappa,t_0)$ irreducible when $C(z)$ is irreducible and reducible otherwise.
\end{definition}
\begin{lemma}\label{lem:irreducible}
The monodromy manifold $\mathcal{M}(\kappa,t_0)$ does not contain reducible monodromy if and only if the non-splitting conditions \eqref{eq:intro_irreducibleparameter} hold true.
\end{lemma}
\begin{remark}
This lemma can be inferred from Ohyama et al. \cite{ohyamaramissualoy}[Theorem 4.3]. We give a proof in Section \ref{subsection:solvability}.
\end{remark}

We are now in a position to state our first main result, which we prove in Section \ref{subsection:solvability}.
\begin{theorem}\label{thm:main_solvability}
Consider RHP I defined in Definition \ref{def:RHPmain}. If the connection matrix $C(z)\in \mathfrak{C}(\kappa,t_0)$ is irreducible, see Definition \ref{def:reducible}, then this RHP is solvable.
More precisely, for any $m\in\mathbb{Z}$, at least one of the solutions $\Psi^{\B{m}}(z)$ and $\Psi^{\B{m+1}}(z)$ of RHP I exists.

Let $(f,g,w)$ be the unique corresponding solution of $q\Psix^{aux}(\kappa,t_0)$ via equations \eqref{eq:rhptofgw}. Then, for $m\in\mathbb{Z}$, $\Psi^{\B{m}}(z)$ fails to exist if and only if
 $(f(t_m),g(t_m))=(\infty,\kappa_\infty)$.
\end{theorem}
\begin{corollary}\label{cor:monodromy_mapping}
If the non-splitting conditions \eqref{eq:intro_irreducibleparameter} hold true, then the monodromy mapping is bijective.
\end{corollary}
\begin{proof}
Due to Lemma \ref{lem:monodromy_mapping_injective}, the monodromy mapping is injective. Take any monodromy in the monodromy manifold. Then, by Lemma \ref{lem:irreducible}, it must be irreducible. Theorem \ref{thm:main_solvability} thus shows that there exists a solution of $q\Psix$ with that monodromy. So the monodromy mapping is also surjective and the corollary follows.
\end{proof}

For reducible monodromy, solvability of RHP I is more subtle than in the irreducible case handled in Theorem \ref{thm:main_solvability}. We discuss this in Section \ref{subsection:reducible}, where we show that the RHP with reducible monodromy can be transformed into the standard Fokas-Its-Kitaev RHP \cites{fokasitskitaev1,fokasitskitaev2} for certain orthogonal polynomials. We further show that the corresponding solutions of $q\Psix$ can be expressed in terms of determinants containing Heine's basic hypergeometric functions. We thus see that special function solutions occur when the monodromy of the linear problem is reducible, a phenomenon well-known for the classical sixth Painlev\'e equation \cite{mazzoccorational}.

\subsection{Results on the monodromy manifold}\label{sec:results_monodromy_surface}
Our second main result is the identification of the monodromy manifold with an explicit surface.
To state this result, we define a set of coordinates on the monodromy manifold $\mathcal{M}(\kappa,t_0)$, using a construction introduced in our previous paper \cite{joshiroffelseniv}. 

Firstly, we require the following notation: for any $2\times2$ matrix $R$ of rank one, let $R_1$ and $R_2$ be respectively its first and second column, then we define $\pi(R)\in\mathbb{P}^1$ by
\begin{equation*}
R_1=\pi(R)R_2,
\end{equation*}
with $\pi(R)=0$ if and only if $R_1=(0,0)^T$ and $\pi(R)=\infty$ if and only if $R_2=(0,0)^T$.

 Take a connection matrix $C(z)\in \mathfrak{C}(\kappa,t_0)$ and denote
 \begin{equation}\label{eq:intro_xnotation}
    (x_1,x_2,x_3,x_4)=(\kappa_t t_0,\kappa_t^{-1} t_0,\kappa_1 ,\kappa_1^{-1}).
\end{equation}
Let $1\leq k\leq 4$, then $|C(z)|$ has a simple zero at $z=x_k$ and thus $C(x_k)$, while nonzero, is not invertible. We define the coordinates
\begin{equation*}
\rho_k=\pi(C(x_k)),\quad (1\leq k\leq 4).
\end{equation*}
Note that $(\rho_1,\rho_2,\rho_3,\rho_4)$ are invariant under left multiplication of $C(z)$ by diagonal matrices. However, multiplication by diagonal matrices from the right has the effect of scaling
\begin{equation}\label{eq:scaling}
    (\rho_1,\rho_2,\rho_3,\rho_4)\rightarrow (c\rho_1,c\rho_2,c\rho_3,c\rho_4),
\end{equation}
for some $c\in\mathbb{C}^*$.

Therefore, the coordinates $\rho$ naturally lie in $(\mathbb{P}^1)^4/\mathbb{C}^*$ and we obtain a mapping
\begin{equation}\label{eq:coordinate_mapping}
    \mathcal{P}:\mathcal{M}(\kappa,t_0)\rightarrow (\mathbb{P}^1)^4/\mathbb{C}^*,[ C(z)]\mapsto [\rho],
\end{equation}
which is easily seen to be an embedding (see Lemma \ref{lem:embedding}).

We proceed in giving an explicit description of the image of the monodromy manifold under $\mathcal{P}$. To this end, we make the following definition.

\begin{definition}\label{def:modulispace}
Define the quadratic polynomial
\begin{equation*}
T(\rho:\kappa,t_0)=T_{12}\rho_1\rho_2+T_{13}\rho_1\rho_3+T_{14}\rho_1\rho_4+T_{23}\rho_2\rho_3+T_{24}\rho_2\rho_4+T_{34}\rho_3\rho_4,
\end{equation*}
with coefficients given by
\begin{align*}
T_{12}&= \theta_q\left(\kappa_t^2,\kappa_1^2\right)\theta_q\left(\kappa_0\kappa_\infty^{-1}t_0,\kappa_0^{-1}\kappa_\infty^{-1}t_0\right)\kappa_\infty^2,\\
T_{34}&= \theta_q\left(\kappa_t^2,\kappa_1^2\right)\theta_q\left(\kappa_0\kappa_\infty t_0,\kappa_0^{-1}\kappa_\infty t_0\right),\\
T_{13}&=- \theta_q\left(\kappa_t\kappa_1^{-1}t_0,\kappa_t^{-1}\kappa_1t_0\right)\theta_q\left(\kappa_t\kappa_1\kappa_0^{-1}\kappa_\infty^{-1},\kappa_0\kappa_t\kappa_1\kappa_\infty^{-1}\right)\kappa_\infty^2,\\
T_{24}&=-\theta_q\left(\kappa_t\kappa_1^{-1}t_0,\kappa_t^{-1}\kappa_1t_0\right) \theta_q\left(\kappa_0\kappa_t\kappa_1\kappa_\infty,\kappa_t\kappa_1\kappa_\infty\kappa_0^{-1}\right),\\
T_{14}&= \theta_q\left(\kappa_t\kappa_1t_0,\kappa_t^{-1}\kappa_1^{-1}t_0\right)\theta_q\left(\kappa_1\kappa_\infty\kappa_0^{-1}\kappa_t^{-1},\kappa_0\kappa_1\kappa_\infty\kappa_t^{-1}\right)\kappa_t^2,\\
T_{23}&= \theta_q\left(\kappa_t\kappa_1t_0,\kappa_t^{-1}\kappa_1^{-1}t_0\right)\theta_q\left(\kappa_t\kappa_\infty\kappa_0^{-1}\kappa_1^{-1},\kappa_0\kappa_t\kappa_\infty\kappa_1^{-1}\right)\kappa_1^2.
\end{align*}
Note that $T$ is homogeneous and multilinear in the variables $\rho=(\rho_1,\rho_2,\rho_3,\rho_4)$. Therefore, if we denote its
homogeneous form by
\begin{equation}\label{eq:defthom}
    T_{hom}(\rho_1^x,\rho_1^y,\rho_2^x,\rho_2^y,\rho_3^x,\rho_3^y,\rho_4^x,\rho_4^y)=\rho_1^y\rho_2^y\rho_3^y\rho_4^yT\left(\frac{\rho_1^x}{\rho_1^y},\frac{\rho_2^x}{\rho_2^y},\frac{\rho_3^x}{\rho_3^y},\frac{\rho_4^x}{\rho_4^y}\right),
  \end{equation}
  then, using homogeneous coordinates $\rho_k=[\rho_k^x: \rho_k^y]\in \mathbb P^1$, $1\le k\le 4$, the equation
  \begin{equation}\label{eq:Thom}
      T_{hom}(\rho_1^x,\rho_1^y,\rho_2^x,\rho_2^y,\rho_3^x,\rho_3^y,\rho_4^x,\rho_4^y)=0,
  \end{equation}
  defines a surface in $(\mathbb{P}^1)^4/\mathbb{C}^*$. We denote this surface by
  \begin{equation*}
  \mathcal{S}(\kappa,t_0)=\{[\rho]\in (\mathbb{P}^1)^4/\mathbb{C}^*:T(\rho:\kappa,t_0)=0\} .
  \end{equation*}
\end{definition}

Our second main result is given by the following theorem, which is proven in Section \ref{subsec:moduli}.
\begin{theorem}\label{thm:main_moduli}
Denote by $\widehat{\kappa}$ the tuple of complex parameters $\kappa$ after replacing $\kappa_0\mapsto 1$.
Then the image of the monodromy manifold $\mathcal{M}(\kappa,t_0)$ under the mapping $\mathcal{P}$, defined in equation \eqref{eq:coordinate_mapping}, is given by the surface $\mathcal{S}(\kappa,t_0)$, minus the curve
\begin{equation}\label{eq:subvariety}
\mathcal{X}(\kappa,t_0):=\mathcal{S}(\kappa,t_0)\cap \mathcal{S}(\widehat{\kappa},t_0).
\end{equation}
Let us denote
\begin{equation}\label{eq:defsstar}
    \mathcal{S}^*(\kappa,t_0)=\mathcal{S}(\kappa,t_0)\setminus \mathcal{X}(\kappa,t_0),
\end{equation}
then, the mapping
\begin{equation*}
 \mathcal{M}(\kappa,t_0)\rightarrow \mathcal{S}^*(\kappa,t_0),\ \textrm{where}\ [C(z)] \mapsto \mathcal{P}([C(z)]),
\end{equation*}
is a bijection.
\end{theorem}
The curve $\mathcal{X}(\kappa,t_0)$ in the above theorem has a geometric interpretation, which is described in the following remark.
\begin{remark}\label{remark:curveX}
The curve $\mathcal{X}=\mathcal{X}(\kappa,t_0)$ does not depend on $\kappa_0$ and can be written as the intersection
\begin{equation}\label{eq:subvariety_infty}
\mathcal{X}=
\bigcap_{\lambda_0\in\mathbb{C}^*}\mathcal{S}(\lambda_0,\kappa_t,\kappa_1,\kappa_\infty,t_0).
\end{equation}
Informally, one can think of points on the curve $\mathcal{X}$ in $\mathcal{S}(\kappa,t_0)$ as corresponding to connection matrices $C(z)\in \mathfrak{C}(\kappa,t_0)$ whose determinant is identically zero, i.e. they satisfy properties (1) and (2) of Definition \ref{def:connection_matrix_space}, but property (3) with $c=0$. Therefore, these coordinate values do not lie in the image of $\mathcal{P}$. In the proof of Theorem \ref{thm:main_moduli}, we obtain an explicit parametrisation of $\mathcal{X}$, see equation \eqref{eq:parametrisationXcurve}.
\end{remark}
We note that, any point $[\rho]\in \mathcal{S}(\kappa,t_0)$ with  more than two coordinates zero or more than two coordinates infinite, necessarily lies on the closed curve $\mathcal{X}$, defined in equation \eqref{eq:subvariety}, and is thus not a point on the surface $\mathcal{S}^*(\kappa,t_0)$.

However, when one of the non-splitting conditions \eqref{eq:intro_irreducibleparameter} is violated, one of the coefficients of the polynomial $T(\rho)$, in Definition \ref{def:modulispace}, vanishes. In that case, there exist points $[\rho]\in \mathcal{S}(\kappa,t_0)$ with precisely
 two coordinates zero or two coordinates infinite. Such points cannot lie on the closed curve $\mathcal{X}$ (as this would imply that $\kappa_0\in q^\mathbb{Z}$), and are in one to one correspondence with reducible monodromy on the monodromy manifold.

For example,
\begin{equation*}
    \left\{[\rho]\in \mathcal{S}^*:\rho_1=\rho_2=0\right\}=\begin{cases}
    \{[(0,0,\rho_3,\rho_4)]:\rho_3,\rho_4\in\mathbb{P}^1\setminus\{0\}\} & \text{ if }T_{34}=0,\\
    \emptyset & \text{ otherwise,}
    \end{cases}
\end{equation*}
and
\begin{equation*}
    \left\{[\rho]\in \mathcal{S}^*:\rho_3=\rho_4=\infty\right\}=\begin{cases}
    \{[(\rho_1,\rho_2,\infty,\infty)]:\rho_{1},\rho_2\in\mathbb{C}\} & \text{ if }T_{34}=0,\\
    \emptyset & \text{ otherwise.}
    \end{cases}
\end{equation*}

If $\kappa_0=\kappa_\infty t_0$, so that $T_{34}=0$, then these two subspaces correspond respectively to the equivalence classes of the collection of upper-triangular connection matrices
\begin{equation*}
    C(z)=\begin{pmatrix}
    \theta_q\left(\frac{z}{\kappa_t t_0},\frac{z\kappa_t}{t_0}\right) & c\,\theta_q\left(\frac{z}{\nu t_0},\frac{z\nu}{\kappa_0\kappa_\infty}\right)\\
    0 & \theta_q\left(\frac{z}{\kappa_1},z\kappa_1\right)
    \end{pmatrix}\qquad (c\in\mathbb{C},\nu \in \mathbb{C}^*),
\end{equation*}
and the equivalence classes of the collection of lower-triangular connection matrices
\begin{equation*}
    C(z)=\begin{pmatrix}
    \theta_q\left(\frac{z}{\kappa_t t_0},\frac{z\kappa_t}{t_0}\right) & 0\\
    c\,\theta_q\left(\frac{z}{\nu t_0},z\nu \kappa_0\kappa_\infty\right) & \theta_q\left(\frac{z}{\kappa_1},z\kappa_1\right)
    \end{pmatrix}\qquad (c\in\mathbb{C},\nu \in \mathbb{C}^*),
\end{equation*}
in the monodromy manifold.

Furthermore, these two subspaces intersect at the single point $[(0,0,\infty,\infty)]\in \mathcal{S}^*(\kappa,t_0)$, which corresponds to the equivalence class of the diagonal connection matrix in the monodromy manifold given by setting $c=0$ in any of the above two formulas.

By Theorem \ref{thm:main_moduli}, the monodromy manifold inherits any topological properties of the space $\mathcal{S}^*(\kappa,t_0)$ via the mapping $\mathcal{P}$. Diagonal monodromy, or anti-diagonal monodromy, form singularities on the monodromy manifold, which is the content of our third main result, proven in Section \ref{subsection:smoothness}.
\begin{theorem}\label{thm:main_smooth}
 If the non-splitting conditions \eqref{eq:intro_irreducibleparameter} hold true, then the monodromy manifold $\mathcal{M}(\kappa,t_0)$ is a smooth complex surface.

On the other hand, if one or more of the non-splitting conditions are violated, then the set
\begin{equation*}
\mathcal{M}_{sing}:=\{[C(z)]\in \mathcal{M}(\kappa,t_0):\text{$C(z)$ is diagonal or anti-diagonal}\},
\end{equation*}
is non-empty (but finite), its elements form singularities of the monodromy manifold and away from them the monodromy manifold is smooth.
\end{theorem}
\begin{remark}\label{remark:ohyamaetal}
We note that the above theorem implies the assertion in Conjecture 7.10 of Ohyama, Ramis and Sauloy \cite{ohyamaramissualoy}. This conjecture is made under the conditions \eqref{eq:non_res}, \eqref{eq:intro_irreducibleparameter} and additional assumptions on the parameters, but our proof shows that the result holds without these additional assumptions. 
\end{remark}

In our fourth and final result we identify the monodromy manifold with an explicit affine algebraic surface via an embedding into $\mathbb{C}^6$.
To construct this embedding, let us denote by
    \begin{equation*}
T'(p)=T_{12}'\rho_1\rho_2+T_{13}'\rho_1\rho_3+T_{14}'\rho_1\rho_4+T_{23}'\rho_2\rho_3+T_{24}'\rho_2\rho_4+T_{34}'\rho_3\rho_4,
\end{equation*}
the quadratic polynomial $T(p)=T(p;\kappa,t_0)$ after replacing $\kappa_0\mapsto 1$.

Take $1\leq i<j\leq 4$ and consider the coordinate
\begin{equation}\label{eq:eta_defi}
    \eta_{ij}:=\frac{T_{ij}\rho_i\rho_j}{\theta_q(\kappa_0,\kappa_0^{-1})T'(\rho)}.
\end{equation}
So, for example, $\eta_{12}$ is given by
\begin{equation*}
    \eta_{12}=\frac{1}{\theta_q(\kappa_0,\kappa_0^{-1})}\frac{T_{12}\rho_1^x\rho_2^x\rho_3^y\rho_4^y}{T_{12}'\rho_1^x\rho_2^x\rho_3^y\rho_4^y+T_{13}'\rho_1^x\rho_2^y\rho_3^x\rho_4^y+\ldots+T_{34}'\rho_1^y\rho_2^y\rho_3^x\rho_4^x},
\end{equation*}
in homogeneous coordinates.

Note that $\eta_{ij}$ is invariant under scalar multiplication $\rho\mapsto c \rho$, $c\in\mathbb{C}^*$. Furthermore, the denominator of $\eta_{ij}$ does not vanish on $\mathcal{S}^*(\kappa,t_0)$, as any such point $[\rho]$ would necessarily lie on the curve $\mathcal{X}$, see equation \eqref{eq:subvariety}.

This means that the $\eta_{ij}$, $1\leq i<j\leq 4$, are six well-defined coordinates on $\mathcal{S}^*(\kappa,t_0)$, and thus on the monodromy manifold $\mathcal{M}(\kappa,t_0)$, which lie in $\mathbb{C}^6$.
Furthermore, by construction, they satisfy the following four equations,
\begin{subequations}\label{eq:eta_equations}
\begin{align}
    &\eta_{12}+\eta_{13}+\eta_{14}+\eta_{23}+\eta_{24}+\eta_{34}=0,\label{eq:eta_equationsa}\\
    &a_{12}\eta_{12}+a_{13}\eta_{13}+a_{14}\eta_{14}+a_{23}\eta_{23}+a_{24}\eta_{24}+a_{34}\eta_{34}=1,\label{eq:eta_equationsb}\\
    &\eta_{13}\eta_{24}-\eta_{12}\eta_{34}b_1=0,\label{eq:eta_equationsc}\\ 
    &\eta_{14}\eta_{23}-\eta_{12}\eta_{34}b_2=0,\label{eq:eta_equationsd}
\end{align}
\end{subequations}
where the coefficients $a_{ij}=T_{ij}'/T_{ij}$, $1\leq i<j\leq 4$, read
\begin{align*}
    &a_{12}=\prod_{\epsilon=\pm 1}\frac{\theta_q\big(\kappa_0^\epsilon\big)\theta_q\big(\kappa_\infty^{-1}t_0\big)}{\theta_q\big(\kappa_0^{\epsilon}\kappa_\infty^{-1}t_0\big)}, &
    &a_{34}=\prod_{\epsilon=\pm 1}\frac{\theta_q\big(\kappa_0^\epsilon\big)\theta_q\big(\kappa_\infty t_0\big)}{\theta_q\big(\kappa_0^{\epsilon}\kappa_\infty t_0\big)},\\
    &a_{13}=\prod_{\epsilon=\pm 1}\frac{\theta_q\big(\kappa_0^\epsilon\big)\theta_q\big(\kappa_t\kappa_1\kappa_\infty^{-1}\big)}{\theta_q\big(\kappa_0^{\epsilon}\kappa_t\kappa_1\kappa_\infty^{-1}\big)}, &
    &a_{24}=\prod_{\epsilon=\pm 1}\frac{\theta_q\big(\kappa_0^\epsilon\big)\theta_q\big(\kappa_t\kappa_1\kappa_\infty \big)}{\theta_q\big(\kappa_0^{\epsilon}\kappa_t\kappa_1\kappa_\infty \big)},\\
    &a_{14}=\prod_{\epsilon=\pm 1}\frac{\theta_q\big(\kappa_0^\epsilon\big)\theta_q\big(\kappa_t^{-1}\kappa_1\kappa_\infty\big)}{\theta_q\big(\kappa_0^{\epsilon}\kappa_t^{-1}\kappa_1\kappa_\infty\big)}, &
    &a_{23}=\prod_{\epsilon=\pm 1}\frac{\theta_q\big(\kappa_0^\epsilon\big)\theta_q\big(\kappa_t\kappa_1^{-1}\kappa_\infty \big)}{\theta_q\big(\kappa_0^{\epsilon}\kappa_t\kappa_1^{-1}\kappa_\infty \big)},\\
\end{align*}
and
\begin{equation*}
    b_1=\frac{T_{13}T_{24}}{T_{12}T_{34}},\qquad b_2=\frac{T_{14}T_{23}}{T_{12}T_{34}}.
\end{equation*}

\begin{definition}\label{def:affine_variety}
We denote by $\mathcal{F}(\kappa,t_0)$ the affine algebraic surface in
\begin{equation*}
    \{(\eta_{12},\eta_{13},\eta_{14},\eta_{23},\eta_{24},\eta_{34})\in\mathbb{C}^6\}
\end{equation*}
defined by equations \eqref{eq:eta_equations}. We correspondingly denote by 
\begin{equation*}
 \Phi:\mathcal{S}^*(\kappa,t_0)\rightarrow \mathcal{F}(\kappa,t_0), [\rho]\rightarrow \eta,
\end{equation*}
the mapping defined through the $\eta$-coordinates \eqref{eq:eta_defi} and write
\begin{equation*}
 \Phi_{\mathcal{M}}=\Phi\circ \mathcal{P}
 :\mathcal{M}(\kappa,t_0)\rightarrow \mathcal{F}(\kappa,t_0), [ C(z)]\rightarrow \eta,
\end{equation*}
where $\mathcal{P}$ is the mapping defined in equation \eqref{eq:coordinate_mapping}.
\end{definition}

Our fourth and final main result is given by the following theorem, which is proved in Section \ref{subsection:embeddings}.
\begin{theorem}\label{thm:main_affine}
Let $\kappa$ and $t_0$ be parameters satisfying the non-resonance conditions \eqref{eq:non_res} and the non-splitting conditions \eqref{eq:intro_irreducibleparameter}. Then the mapping $\Phi_{\mathcal{M}}$, given in Definition \ref{def:affine_variety}, is an isomorphism between the
monodromy manifold $\mathcal{M}(\kappa,t_0)$ and the affine algebraic surface $\mathcal{F}(\kappa,t_0)$.
\end{theorem}
\begin{remark}
We note that the algebraic surface $\mathcal{F}(\kappa,t_0)$ is invariant under the translations
\begin{equation*}
   t_0\mapsto q\, t_0,\qquad \kappa_j\mapsto q\,\kappa_j\quad (j=0,t,1,\infty),
\end{equation*}
since the coefficients in equations \eqref{eq:eta_equations} are invariant under them.
\end{remark}

The surface $\mathcal{F}(\kappa,t_0)$ can be identified with the intersection of two quadrics in $\mathbb{C}^4$. This can be seen by using equations \eqref{eq:eta_equationsa} and \eqref{eq:eta_equationsb} to eliminate any two of the six variables.

For example, consider eliminating $\{\eta_{24},\eta_{34}\}$ from \eqref{eq:eta_equations} using \eqref{eq:eta_equationsa} and \eqref{eq:eta_equationsb}. The relevant determinant is given by
\begin{equation*}
    \begin{vmatrix}
    1 & 1\\
    a_{24} & a_{34}\\
    \end{vmatrix}=\kappa_t\kappa_1\kappa_\infty \theta_q(\kappa_t^{-1}\kappa_1^{-1}t_0,\kappa_t\kappa_1\kappa_\infty^2t_0)\prod_{\epsilon=\pm 1}{\frac{\theta_q(\kappa_0^\epsilon)^2}{\theta_q(\kappa_0^\epsilon \kappa_\infty t_0, \kappa_0^\epsilon \kappa_t \kappa_1 \kappa_\infty)}}.
\end{equation*}

Let us assume that $\kappa_t\kappa_1\kappa_\infty^2t_0\notin q^\mathbb{Z}$. If not, then we can instead choose another pair of coordinates to eliminate. The non-resonance conditions \eqref{eq:non_res} and non-splitting conditions \eqref{eq:intro_irreducibleparameter} now guarantee that the above determinant is non-zero. Upon eliminating $\{\eta_{24},\eta_{34}\}$,
equations  \eqref{eq:eta_equationsc} and  \eqref{eq:eta_equationsd} respectively become
\begin{equation} \label{eq:quadrics}
\begin{aligned}
    u_{0} \eta_{12}^2+u_{1} \eta_{12}\eta_{13}+u_{2} \eta_{12}\eta_{14} +u_{3} \eta_{12}\eta_{23}+u_{4}\eta_{14}\eta_{23} +u_5 \eta_{12}=0,\\
    v_{0}\hspace{0.3mm} \eta_{13}^2+v_{1} \hspace{0.3mm}\eta_{12}\eta_{13}+v_{2} \hspace{0.3mm}\eta_{13}\eta_{14} +v_{3}\hspace{0.3mm} \eta_{13}\eta_{23}+v_{4}\hspace{0.3mm}\eta_{14}\eta_{23} +v_5\hspace{0.3mm} \eta_{13}=0,
\end{aligned}
\end{equation}
with coefficients given by
\begin{align*}
    u_0&=-\kappa_t^2\kappa_1^2\kappa_\infty^2\, \theta_q\left(t_0\kappa_t\kappa_1,\frac{t_0}{\kappa_t\kappa_1\kappa_\infty^2}\right)\prod_{\epsilon=\pm 1}\frac{\theta_q\left(\kappa_0^\epsilon\right)}{\theta_q\left(\frac{t_0}{\kappa_\infty}\kappa_0^\epsilon\right)},\\
    u_1&=\kappa_t^2\kappa_1^2\,\theta_q\left(\kappa_t^2\kappa_1^2,\kappa_\infty^2\right)\prod_{\epsilon=\pm 1}\frac{\theta_q\left(\kappa_0^\epsilon\right)}{\theta_q\left(\frac{\kappa_t\kappa_1}{\kappa_\infty}\kappa_0^\epsilon\right)},\\
    u_2&=\kappa_1^2\kappa_\infty^2\,\theta_q\left(\kappa_1^2\kappa_\infty^2,\kappa_t^2\right)\prod_{\epsilon=\pm 1}\frac{\theta_q\left(\kappa_0^\epsilon\right)}{\theta_q\left(\frac{\kappa_1\kappa_\infty}{\kappa_t}\kappa_0^\epsilon\right)},\\
    u_3&=\kappa_t^2\kappa_\infty^2\,\theta_q\left(\kappa_t^2\kappa_\infty^2,\kappa_1^2\right)\prod_{\epsilon=\pm 1}\frac{\theta_q\left(\kappa_0^\epsilon\right)}{\theta_q\left(\frac{\kappa_t\kappa_\infty}{\kappa_1}\kappa_0^\epsilon\right)},\\
    u_4&=-\kappa_\infty^4 
    \frac{\theta_q\left(\kappa_t^2,\kappa_1^2\right)^2 \theta_q(t_0\kappa_t\kappa_1\kappa_\infty^2)}{\theta_q(t_0\kappa_t\kappa_1)^2\, \theta_q\left(\frac{t_0}{\kappa_t\kappa_1}\right)}
    \prod_{\epsilon=\pm 1}\frac{\theta_q\left(\kappa_0^\epsilon,\frac{t_0}{\kappa_\infty}\kappa_0^\epsilon\right)}{\theta_q\left(\frac{\kappa_1\kappa_\infty}{\kappa_t}\kappa_0^\epsilon,\frac{\kappa_t\kappa_\infty}{\kappa_1}\kappa_0^\epsilon\right)},\\
    u_5&=\kappa_t\kappa_1\kappa_\infty \prod_{\epsilon=\pm 1}\frac{\theta_q\left(\kappa_t\kappa_1\kappa_\infty\kappa_0^\epsilon\right)}{\theta_q\left(\kappa_0^\epsilon\right)},
\end{align*}
and
\begin{align*}
    v_0&=-\kappa_t^2\kappa_1^2\, \theta_q\left(t_0\kappa_t\kappa_1,\frac{t_0 \kappa_\infty^2}{\kappa_t\kappa_1}\right)\prod_{\epsilon=\pm 1}\frac{\theta_q\left(\kappa_0^\epsilon\right)}{\theta_q\left(\frac{\kappa_t\kappa_1}{\kappa_\infty}\kappa_0^\epsilon\right)},\\
    v_1&=t_0\kappa_t\kappa_1\,\theta_q\left(t_0^2,\kappa_\infty^2\right)\prod_{\epsilon=\pm 1}\frac{\theta_q\left(\kappa_0^\epsilon\right)}{\theta_q\left(\frac{t_0}{\kappa_\infty}\kappa_0^\epsilon\right)},\\
    v_2&=\kappa_1^2\kappa_\infty^2\,\theta_q\left(\frac{t_0 \kappa_t}{\kappa_1},\frac{t_0\kappa_1 \kappa_\infty^2}{\kappa_t}\right)\prod_{\epsilon=\pm 1}\frac{\theta_q\left(\kappa_0^\epsilon\right)}{\theta_q\left(\frac{\kappa_1\kappa_\infty}{\kappa_t}\kappa_0^\epsilon\right)},\\
    v_3&=\kappa_t^2\kappa_\infty^2\,\theta_q\left(\frac{t_0 \kappa_1}{\kappa_t},\frac{t_0\kappa_t \kappa_\infty^2}{\kappa_1}\right)\prod_{\epsilon=\pm 1}\frac{\theta_q\left(\kappa_0^\epsilon\right)}{\theta_q\left(\frac{\kappa_t\kappa_\infty}{\kappa_1}\kappa_0^\epsilon\right)},\\
    v_4&=\kappa_\infty^4 
    \frac{\theta_q\left(\frac{t_0\kappa_t}{\kappa_1},\frac{t_0\kappa_1}{\kappa_t}\right)^2 \theta_q(t_0\kappa_t\kappa_1\kappa_\infty^2)}{\theta_q(t_0\kappa_t\kappa_1)^2\, \theta_q\left(\frac{t_0}{\kappa_t\kappa_1}\right)}
    \prod_{\epsilon=\pm 1}\frac{\theta_q\left(\kappa_0^\epsilon,\frac{\kappa_t\kappa_1}{\kappa_\infty}\kappa_0^\epsilon\right)}{\theta_q\left(\frac{\kappa_1\kappa_\infty}{\kappa_t}\kappa_0^\epsilon,\frac{\kappa_t\kappa_\infty}{\kappa_1}\kappa_0^\epsilon\right)},\\
    v_5&=\kappa_t\kappa_1\kappa_\infty \prod_{\epsilon=\pm 1}\frac{\theta_q\left(t_0\kappa_\infty\kappa_0^\epsilon\right)}{\theta_q\left(\kappa_0^\epsilon\right)},
\end{align*}

Thus, for generic parameter values, the monodromy manifold of $q\Psix$ is isomorphic to the intersection of the two quadrics defined by equations \eqref{eq:quadrics} in $\mathbb{C}^4$. Intersections of two quadrics in $\mathbb{P}^4$ are known as Segre surfaces and it is well-known that they are isomorphic to Del Pezzo surfaces of degree four, see e.g. \cite{griffithsharris}.

It is interesting to contrast this with the monodromy manifolds of the classical Painlev\'e equations. They are isomorphic to affine cubic surfaces \cite{putsaito2009}. In particular, their corresponding projective completions are Del Pezzo surfaces of degree three \cite{griffithsharris}.

We further note that Chekhov et al. \cite{mazzoccoconjectures} conjectured explicit affine Del Pezzo surfaces of degree three as the monodromy manifolds of the $q$-Painlev\'e equations higher up in Sakai's classification scheme \cite{s:01} than $q\Psix$.

From Corollary \ref{cor:monodromy_mapping}, Theorem \ref{thm:main_smooth} and Theorem \ref{thm:main_affine}, we obtain the following corollary.

\begin{corollary}
Let $\kappa$ and $t_0$ be such that the non-resonance conditions \eqref{eq:non_res} and non-splitting conditions \eqref{eq:intro_irreducibleparameter} are fulfilled. Then,
composition of the monodromy mapping with $\Phi_{\mathcal{M}}$, defined in Definition \ref{def:affine_variety}, yields a bijective mapping from the solution space of $q\Psix(\kappa_,t_0)$ to the smooth algebraic surface $\mathcal{F}(\kappa,t_0)$,
\begin{equation}\label{eq:embeddingsolutionspace}
\{(f,g)\text{ solution of }q\Psix(\kappa_,t_0)\}\rightarrow \mathcal{F}(\kappa,t_0).
\end{equation}
In particular, we may write the general solution of $q\Psix(\kappa,t_0)$ as
\begin{subequations}\label{eq:generalsol}
\begin{align*}
    f(t)&=f(t;\kappa,t_0,\eta),\\
    g(t)&=g(t;\kappa,t_0,\eta),
\end{align*}
\end{subequations}
with $t\in q^\mathbb{Z} t_0$ and $\eta$ varying in $\mathcal{F}(\kappa,t_0)$.
\end{corollary}
\begin{remark}
By identifying the domain of the mapping \eqref{eq:embeddingsolutionspace} with the initial value space of $q\Psix$ at $t=t_0$, the mapping becomes a bijective correspondence between complex (algebraic) surfaces. One can show that this correspondence is a biholomorphism using standard arguments. Namely, one observes that the matrix functions $\Psi_j(z,t_0)$, $j=0,\infty$, defined in equations \eqref{eq:linear_sys_solutions}, can be chosen locally analytically in $(f,g)$ as long as one stays away from the exceptional lines above the base points $b_7$ and $b_8$. The corresponding connection matrix is then locally analytic in $(f,g)$ and, consequently, so are the $\eta$-coordinates. To prove the latter statement around points on the exceptional lines above $b_7$ and $b_8$, one simply applies the argument with $t=q\,t_0$ rather than $t=t_0$, recalling that the time-evolution is a biholomorphism beween the initial value spaces at $t=t_0$ and $t=q\, t_0$. It follows that the mapping \eqref{eq:embeddingsolutionspace} is a bijective holomorphism and thus biholomorphism.
\end{remark}

\begin{remark}
By specialising to the parameter setting
\begin{equation}\label{eq:parametercollapse}
    \kappa_0=\kappa_t,\quad \kappa_\infty=p^{-1}\kappa_1,\quad p=q^{\frac{1}{2}},
\end{equation}
the $q\Psix(\kappa)$ equation collapses to its symmetric form
\begin{equation*}
    q\text{SP}_{\text{VI}}:\quad \widetilde{h}\undertilde{h}=\frac{(h-\kappa_t^{+1} t)(h-\kappa_t^{-1} t )}{(h-\kappa_1^{+1})(h-\kappa_1^{-1})},
\end{equation*}
where
\begin{equation*}
    h=h(t),\quad \widetilde{h}=h(p\,t),\quad \undertilde{h}=h(t/p),
\end{equation*}
and $h$ is related to $(f,g)$ as
\begin{equation*}
h(p^{2m} t_0)=f(q^mt_0),\quad h(p^{2m-1}t_0)=g(q^mt_0)\quad (m\in\mathbb{Z}).
\end{equation*}
As both the non-resonance and non-splitting conditions \eqref{eq:non_res} and \eqref{eq:intro_irreducibleparameter} are generically not violated by  \eqref{eq:parametercollapse}, all the aspects of our treatment of $q\Psix$ can be carried over to $q\text{SP}_{\text{VI}}$. We further note that $q\text{SP}_{\text{VI}}$  is also known as $q\text{P}_\text{III}$ in the literature \cite{gramram}.
\end{remark}

\begin{remark}
Regarding Painlev\'e VI, and its associated standard linear problem, the corresponding monodromy mapping was thoroughly studied by Inaba et al. \cite{inaba2006}. The associated monodromy manifold can be identified with an explicit affine cubic surface,
a fact which first appeared in Fricke and Klein \cite{frickeklein} and was rediscovered by Jimbo \cite{jimbo1982} in the context of Painlev\'e VI. Our construction of the surface $\mathcal{F}(\kappa,t_0)$, in Theorem \ref{thm:main_affine}, may be considered as a $q$-analog of this.
Iwasaki \cite{iwasaki} studied the smoothness of the Painlev\'e VI monodromy manifold and associated cubic. Theorem \ref{thm:main_smooth} can be considered a $q$-analog of \cite{iwasaki}[Theorem 1] in the non-resonant parameter regime.
\end{remark}

\section{The Linear Problem}\label{sec:linear_problem}
Consider the linear system
\begin{equation}\label{eq:linear_system}
    Y(qz)=A(z)Y(z),
\end{equation}
where $A(z)$  is a complex $2\times 2$ matrix polynomial of degree two,
\begin{equation*}
    A(z)=A_0+zA_1+z^2 A_2,
\end{equation*}
with both $A_0$ and $A_2$ invertible and semi-simple.

Jimbo and Sakai \cite{jimbosakai} showed that isomonodromic deformation of such a linear system, as the eigenvalues of $A_0$ as well as two of the zeros of the determinant of $A(z)$ evolve via multiplication by $q$, defines an evolution of the coefficient matrix $A(z)$ which is birationally equivalent to $q\Psix^{aux}$.

In Section \ref{subsec:normalising_linear}, we show that the linear system \eqref{eq:linear_system} can always be normalised to the standard form \eqref{eq:linear_problem} we use in this paper. Then, in Section \ref{subsec:isomonodromic_deformation}, we formulate the main results of Jimbo and Sakai \cite{jimbosakai} regarding isomonodromic deformation of the linear system \eqref{eq:linear_problem} and prove Lemma \ref{lem:isomonodromic}.

Finally, in Section \ref{sec:isorhp}, we show how the linear system \eqref{eq:linear_problem} can be recovered from RHP I, defined in Definition \ref{def:RHPmain}, yielding in particular Lemma \ref{lem:monodromy_mapping_injective}.

\subsection{Normalising the linear system}\label{subsec:normalising_linear}
In this section we normalise the linear system \eqref{eq:linear_system} to the standard form \eqref{eq:linear_problem}.

Recall that $A_0$ and $A_2$ are semi-simple and we denote their eigenvalues by $\{\sigma_1,\sigma_2\}$ and $\{\mu_1,\mu_2\}$ respectively. By means of gauging the linear system with a constant matrix, $Y(z)\mapsto G Y(z)$, so that $A(z)\mapsto G A(z)G^{-1}$, we may ensure that $A_2=\operatorname{diag}(\mu_1,\mu_2)$ is diagonal.

We further denote the zeros of the determinant of $A(z)$ by $x_k$, $1\leq k\leq 4$, so that
\begin{equation}\label{eq:detAA}
    |A(z)|=\mu_1\mu_2(z-x_1)(z-x_2)(z-x_3)(z-x_4).
\end{equation}
Evaluating this determinant at $z=0$ gives the identity
\begin{equation*}
    \sigma_1\sigma_2=\mu_1\mu_2 x_1 x_2 x_3 x_4.
\end{equation*}

By means of a scalar gauge as well as a scaling of the independent variable,
\begin{equation*}
    Y(z)\mapsto g(z)Y(cz),\quad g(z):=z^{\log_q(s)},\quad c,s\in\mathbb{C}^*,
\end{equation*}
so that the linear system transforms as $A(z)\mapsto s A(cz)$, we may ensure that
\begin{equation*}
    \mu_1\mu_2=1,\quad x_3x_4=1,\quad \sigma_1\sigma_2=x_1 x_2.
\end{equation*}
We introduce a time variable $t$, satisfying $t^2=\sigma_1\sigma_2$, and four nonzero parameters $\kappa=(\kappa_0,\kappa_t,\kappa_1,\kappa_\infty)$, through
\begin{align*}
&\sigma_1=\kappa_0^{+1}t, & &x_1=\kappa_t^{+1}t, & &x_3=\kappa_1^{+1}, &  &\mu_1=\kappa_\infty^{+1},\\
&\sigma_2=\kappa_0^{-1}t, & &x_2=\kappa_t^{-1}t, &
&x_4=\kappa_1^{-1}, & &\mu_2=\kappa_\infty^{-1},
\end{align*}
and note that the linear system \eqref{eq:linear_system} has now been normalised to the form \eqref{eq:linear_problem}.

\subsection{Isomonodromic deformation of the linear system} \label{subsec:isomonodromic_deformation}
In this section we state important results by Jimbo and Sakai \cite{jimbosakai} on the isomonodromic deformation of the linear system \eqref{eq:linear_problem}. Here we recall that isomonodromic deformation stands for deformation as $t\rightarrow q\,t$ such that $P(z,qt)=P(z,t)$, or equivalently, such that the connection matrix satisfies
\begin{equation}\label{eq:connection_deformation}
    C(z,qt)=z\,C(z,t)
\end{equation}

\begin{theorem}[Jimbo and Sakai \cite{jimbosakai}] \label{thm:isomonodromic}
Considering the linear system \eqref{eq:linear_problem}, equation \eqref{eq:connection_deformation} holds if and only if both $Y_0(z,t)$ and $Y_\infty(z,t)$, defined in equations \eqref{eq:linear_sys_solutions}, satisfy
\begin{equation}\label{eq:time_deformation}
    Y(z,qt)=B(z,t)Y(z,t),
\end{equation}
for an (a posteriori unique), rational in $z$, matrix function $B(z,t)$, which takes the form
\begin{equation*}
    B(z,t)=\frac{z^2 I+zB_0(t)}{(z-q\kappa_t^{+1}t)(z-q\kappa_t^{-1}t)}.
\end{equation*}
\end{theorem}
We proceed in making the time-evolution defined by \eqref{eq:time_deformation} more explicit.
Note that compatibility of the linear system \eqref{eq:linear_problem} and time deformation \eqref{eq:time_deformation} amounts to the following evolution of the coefficient matrix $A$,
\begin{equation}\label{eq:linear_system_evolution}
    A(z,qt)B(z,t)=B(qz,t)A(z,t),
\end{equation}
as well as the following evolution of the diagonalising matrix $H(t)$ in \eqref{eq:diagonal},
\begin{equation}\label{eq:diagonal_evolution}
    H(qt)=B_0(t)H(t).
\end{equation}

We use the standard coordinates $f=f(t),g=g(t)$ and $w=w(t)$, defined by equations \eqref{eq:coordinates_linear}, on the linear system, whose definition we repeat here for convenience of the reader,
\begin{align}
    A_{12}(z,t)&=\kappa_\infty^{-1} w(z-f),\label{eq:param1}\\
    A_{22}(f,t)&=q(f-\kappa_1)(f-\kappa_1^{-1})g.\nonumber
\end{align}
Then the linear system is given in terms of $\{f,g,w\}$ by
\begin{equation*}
    A(z,t)=\begin{pmatrix}
    \kappa_\infty ((z-f)(z-\alpha)+g_1) & \kappa_\infty^{-1} w(z-f)\\
    \kappa_\infty w^{-1}(\gamma z+\delta) & \kappa_\infty^{-1}((z-f)(z-\beta)+g_2)
    \end{pmatrix},
\end{equation*}
where
\begin{align}
g_1&=q^{-1}\kappa_\infty^{-1}(f-\kappa_t t)(f-\kappa_t^{-1}t)g^{-1},\label{eq:g1}\\ 
g_2&=q\kappa_\infty(f-\kappa_1)(f-\kappa_1^{-1})g,\nonumber
\end{align}
and, temporarily using the notation $\mathring{\kappa}=\kappa+\kappa^{-1}$,
\begin{align*}
\alpha&=\frac{1}{(1-\kappa_\infty^2)f}
\left(\kappa_\infty^2 g_1-\kappa_\infty\mathring{\kappa}_0t+g_2+(\mathring{\kappa}_tt+\mathring{\kappa}_1)f-2 f^2\right),\\
\beta&=\frac{1}{(\kappa_\infty^2-1)f}
\left(\kappa_\infty^2 g_1-\kappa_\infty\mathring{\kappa}_0t+g_2+\kappa_\infty^2(\mathring{\kappa}_tt+\mathring{\kappa}_1)f-2 \kappa_\infty^2f^2\right),\\
\gamma&=g_1+g_2+f^2+2(\alpha+\beta)f+\alpha\beta -(t^2+\mathring{\kappa}_t\mathring{\kappa}_1t+1),\\           
\delta&=f^{-1}(t^2-(g_1+\alpha f)(g_2+\beta f)),
\end{align*}

Equation \eqref{eq:linear_system_evolution} is equivalent to the following conditions on the matrix $B_0(t)$,
\begin{equation*}
\begin{aligned}\label{eq:overdeterminedB0}
    &A(q\kappa_t^{\pm 1}t, qt)(q\kappa_t^{\pm 1}t I+B_0(t))=0,\\
    &(q\kappa_t^{\pm 1}t I+B_0(t))A(\kappa_t^{\pm 1}t, t)=0,\\
    &A_0(qt)B_0(t)=qB_0(t)A_0(t).
\end{aligned}
\end{equation*}
The first two equations follow from the fact that both the left and right-hand side of \eqref{eq:linear_system_evolution} are necessarily analytic in $z\in\mathbb{C}$ and the third follows from equating the degree one terms in $z$ of both sides of equation \eqref{eq:linear_system_evolution}.

These equations form an over-determined system for $B_0=B_0(t)$. They allow one to express $B_0$ explicitly in terms of $\{f,g,w\}$, for example
\begin{equation*}
    B_0=\begin{pmatrix}
    \frac{q}{1-q}(\overline{f}+\overline{\beta}-f-\beta) & -\frac{q(\overline{w}-w)}{q\kappa_\infty^2-1}\\
    \frac{q\kappa_\infty^2}{\kappa_\infty^2-q}\left(\frac{\overline{\gamma}}{\overline{w}}-\frac{\gamma}{w}\right) & \frac{q}{1-q}(\overline{f}+\overline{\alpha}-f-\alpha)
    \end{pmatrix},
\end{equation*}
and Jimbo and Sakai \cite{jimbosakai} showed that equations \eqref{eq:overdeterminedB0} are then equivalent to the $q\Psix^{aux}$ time evolution of $(f,g,w)$.

Furthermore, by means of a direct computation, one can check that equations \eqref{eq:diagonal} and \eqref{eq:diagonal_evolution} translate to the elements of the diagonalising matrix $H=(h_{ij})_{1\leq i,j\leq 2}$ satisfying
\begin{subequations}\label{eq:diagonalisation_evolution}
\begin{align}
\frac{\overline{h}_{11}}{h_{11}}&=-\frac{q t}{f}\frac{\kappa_\infty(\overline{g}-\kappa_0 t)}{\kappa_0(\overline{g}-\kappa_\infty)},\\
\frac{\overline{h}_{12}}{h_{12}}&=-\frac{q t}{f}\kappa_0\kappa_\infty\frac{(\overline{g}-t/\kappa_0 )}{(\overline{g}-\kappa_\infty)},\\
\frac{h_{21}}{h_{11}}&=\kappa_\infty\frac{\kappa_\infty g_1+\kappa_\infty f \alpha-t\kappa_0}{fw},\label{eq:diagonalisation_evolution3}\\
\frac{h_{22}}{h_{12}}&=\kappa_\infty\frac{\kappa_\infty g_1+\kappa_\infty f \alpha-t\kappa_0^{-1}}{fw}.\label{eq:diagonalisation_evolution4}
\end{align}
\end{subequations}

We are now in a position to prove Lemma  \ref{lem:isomonodromic}.
\begin{proof}[Proof of Lemma \ref{lem:isomonodromic}]
We start by showing that the linear system $A=A(z,t)$ is regular in $t$ away from values where $(f,g)=(\infty,\kappa_\infty)$. To this end, consider the parametrisation of $A=A(z,t)$ with respect to $(f,g,w)$.
By direct inspection, one can see that this parametrisation is regular for all values of $(f,g)\in \mathbb{C}^*\times \mathbb{C}^*$ and $w\in\mathbb{C}^*$. The same is true near each of the six basepoints $b_k$, $1\leq k\leq 6$, defined in equation \eqref{eq:basepoints}.

For example, consider the basepoint $b_3=(\kappa_t t,0)$.
We apply a change of variables,
\begin{equation*}
    f-\kappa_t t=FG,\quad g=G,
\end{equation*}
so that  $\{F\in\mathbb{C},G=0\}$ lies on the exceptional line above $b_3$, after a local blow up. The parametrisation of the matrix polynomial $A$ is regular at $G=0$, and takes the form
\begin{equation*}
    A(z,t)=\begin{pmatrix}
    \kappa_\infty ((z-\kappa_t t)(z-\alpha)+g_1) & \kappa_\infty^{-1} w(z-\kappa_t t)\\
    \kappa_\infty w^{-1}(\gamma z+\delta) & \kappa_\infty^{-1}(z-\kappa_t t)(z-\beta)
    \end{pmatrix},
\end{equation*}
with
\begin{equation*}
    g_1=q^{-1}\kappa_\infty^{-1}(\kappa_1-\kappa_1^{-1})t F.
\end{equation*}
Geometrically, the line $\{F\in\mathbb{C},G=0\}$, above $b_3$, parametrises coefficient matrices $A$ whose second column vanishes at $z=\kappa_t t$. The one remaining point on the exceptional line above $b_3$, which does not lie on this line, is an inaccessible initial value. Namely, the corresponding formal solution of $q\Psix$ never takes value in $\mathbb{C}^*\times \mathbb{C}^*$ and is thus not a genuine solution. We conclude that $A$ is regular for $(f,g)$ near $b_3$.
Similarly, it is shown that $A$ is regular near the other basepoints $b_k$, $1\leq k\leq 6$, $k\neq 3$.

The situation is slightly more involved for the remaining base-points $b_7$ and $b_8$, as the auxiliary equation \eqref{eq:auxiliary} is singular at these points. Firstly, as $(f,g)$ approaches $b_8=(\infty,\kappa_\infty^{-1}q^{-1})$, $\overline{g}$ approaches $\kappa_\infty$ and consequently $w$ vanishes, due to the auxiliary equation. Consider thus the change of variables
\begin{equation*}
    f=F^{-1},\quad g-\kappa_\infty^{-1}q^{-1}=FG,\quad w=FW.
\end{equation*}
In the local chart $\{F,G,W\}$, the coefficient matrix $A$ is regular at $F=0$.
Geometrically, the line $\{F=0,G\in\mathbb{C}\}$, above $b_8$, parametrises coefficient matrices $A$ for which the entry $A_{12}(z)$ is constant.
In particular, $A$ is regular near $b_8$.

Finally, by the same reasoning, it follows that $w\rightarrow \infty$, as $(f,g)$ approaches $b_7=(\infty,\kappa_\infty)$, and that the coefficient matrix $A$ is thus singular there.

We conclude that $A(z,t)$ is singular at $t=t_*$ if and only if $(f(t_*),g(t_*))=(\infty,\kappa_\infty)$. Correspondingly, we write
\begin{equation}
\mathfrak{M}=\left\{m\in\mathbb{Z}:(f(q^mt_0),g(q^mt_0))\neq (\infty,\kappa_\infty)\}\right\}.
\end{equation}

For every $t\in q^{\mathfrak{M}}t_0$, we choose any $H(t)$  satisfying \eqref{eq:diagonal}, but not necessarily \eqref{eq:diagonal_evolution}, and let $C(z,t)$ denote the corresponding connection matrix.
We proceed with proving equation \eqref{eq:C0definition} in the lemma. 

To prove \eqref{eq:C0definition}, it is enough to show that, for any $m\in \mathfrak{M}$, \begin{equation*}
    C(z,qt_m)=z \Delta C(z,t_m),
\end{equation*}
for some diagonal matrix $\Delta$, if $m+1\in \mathfrak{M}$, and
\begin{equation}\label{eq:cevolution2}
    C(z,q^2t_m)=z^2 \Delta C(z,t_m),
\end{equation}
for some diagonal matrix $\Delta$, if $m+1\notin \mathfrak{M}$ (so that necessarily $m+2\in \mathfrak{M}$).

The first case is a direct consequence of Theorem \ref{thm:isomonodromic}. We may further ensure that $\Delta=I$ by imposing equation \eqref{eq:diagonal_evolution} at $t=t_m$.

As to the second case, we note that, analogues to the proof of Theorem \ref{thm:isomonodromic} by Jimbo and Sakai \cite{jimbosakai}, one can show that $P(z,q^2t)=P(z,t)$ if and only if 
 $Y_0(z,t)$ and $Y_\infty(z,t)$ both satisfy
\begin{equation*}
    Y(z,q^2t)=F(z,t)Y(z,t),
\end{equation*}
for an (a posteriori unique), rational in $z$, matrix function $F(z,t)$ which takes the form
\begin{equation*}
    F(z,t)=\frac{z^4 I+z^3 F_1(t)+z^2F_0(t)}{(z-\kappa_t^{+1}qt)(z-\kappa_t^{-1}qt)(z-\kappa_t^{+1}t)(z-\kappa_t^{-1}t)}.
\end{equation*}
The corresponding time evolution of the coefficient matrix
\begin{equation*}
    A(z,q^2t)=F(qz,t)A(z,t)F(z,t)^{-1},
\end{equation*}
is equivalent to two iterations of $q\Psix^{aux}$, and
\begin{equation*}
    F(z,t)=B(z,qt)B(z,t).
\end{equation*}
By specialising to $t=t_m$, we obtain \eqref{eq:cevolution2}. We may further ensure that $\Delta=I$, by imposing
\begin{equation*}
    H(q^2 t_m)=F_0(t_m)H(t_m).
\end{equation*}
This establishes equation \eqref{eq:C0definition}.

The last statement of the lemma, follows from the fact that, rescaling $H(t)\mapsto H(t)D(t)$, yields $\Psi_0(z,t)\mapsto \Psi_0(z,t)D(t)$ and thus $C(z,t)\mapsto D(t)^{-1} C(z,t)$.
\end{proof}

\subsection{On the $q\Psix$ RHP}\label{sec:isorhp}
In Section \ref{sec:rhp}, we formulated the main Riemann-Hilbert problem for the $q\Psix$ equation, RHP I, in Definition \ref{def:RHPmain}. Let $(f,g)$ be a solution of $q\Psix(\kappa,t_0)$ and $[C(z)]$ be its corresponding monodromy in the monodromy manifold via the monodromy mapping, see Definition \ref{def:monodromy_mapping}. Then equation \eqref{eq:rhpsolution} defines a solution of RHP I. In this section, we show now we may reconstruct the solution $(f,g)$ from the solution of RHP I, giving in particular formulas \eqref{eq:rhptofgw}. This furthermore yields a proof of
Lemma \ref{lem:monodromy_mapping_injective}.

Firstly, we prove Lemma \ref{lem:uniqueness}.
\begin{proof}[Proof of Lemma \ref{lem:uniqueness}]
Note that the determinant of $z^{m}C(z)$ may be written as
\begin{equation*}
z^{2m}|C(z)|=c_m^{-1}\theta_q\left(\kappa_t^{+1}\frac{z}{t_m},\kappa_t^{-1}\frac{z}{t_m},\kappa_1^{+1}z,\kappa_1^{-1}z\right),\quad t_m=q^mt_0,
\end{equation*}
for some $c_m\in\mathbb{C}^*$.
Assume we have a solution $\Psi^{\B{m}}(z)$ of RHP I, defined in Definition \ref{def:RHPmain}. Then its determinant $\Delta^{\B{m}}(z)$ is analytic on $\mathbb{C}\setminus \gamma^{\B{m}}$, it satisfies the jump condition
\begin{equation*}
    \Delta_+^{\B{m}}(z)=\Delta_-^{\B{m}}(z)\, c_m^{-1}\theta_q\left(\kappa_t^{+1}\frac{z}{t_m},\kappa_t^{-1}\frac{z}{t_m},\kappa_1^{+1}z,\kappa_1^{-1}z\right)\quad (z\in \gamma^{\B{m}}),
\end{equation*}
and $\Delta^{\B{m}}(z)=1+\mathcal{O}(z^{-1})$ as $z\rightarrow \infty$.\\
This scalar RHP is uniquely solved by
\begin{equation}\label{eq:rhpdeterminant}
|\Delta^{\B{m}}(z)|=
\begin{cases}
\left(\kappa_t^{+1}\frac{qt}{z},\kappa_t^{-1}\frac{qt}{z},\kappa_1^{+1}\frac{q}{z},\kappa_1^{-1}\frac{q}{z};q\right)_\infty & \text{ if }z\in D_+,\\
c_m\left(\kappa_t^{+1}\frac{z}{t},\kappa_t^{-1}\frac{z}{t},\kappa_1^{+1}z,\kappa_1^{-1}z;q\right)_\infty^{-1} & \text{ if }z\in D_-.
    \end{cases}
\end{equation}
Indeed, the right-hand side satisfies this scalar RHP and, denoting the quotient of the left- and right-hand side of \eqref{eq:rhpdeterminant} by $g(z)$, it follows that $g(z)$ is an entire function on the complex plane satisfying $g(z)\rightarrow 1$ as $z\rightarrow \infty$. By Liouville's theorem, $g(z)\equiv 1$, which yields equation \eqref{eq:rhpdeterminant}.
In particular, the solution $\Psi^{\B{m}}(z)$ is globally invertible on $\mathbb{C}$.

Suppose we have another solution $\widetilde{\Psi}^{\B{m}}(z)$ of RHP I, then the quotient
\begin{equation*}
    R(z)=\widetilde{\Psi}^{\B{m}}(z)\Psi^{\B{m}}(z)^{-1},
\end{equation*}
is analytic on $\mathbb{C}\setminus \gamma^{\B{m}}$. Furthermore, $R(z)$ has a trivial jump on $\gamma^{\B{m}}$, i.e. $R_+(z)=R_-(z)$. Therefore, $R(z)$ extends to an analytic function on the entire complex plane. Finally, we know that $R(z)=I+\mathcal{O}(z^{-1})$ as $z\rightarrow \infty$, thus $R(z)\equiv I$, again by Liouville's theorem, and the lemma follows.
\end{proof}

Starting with a solution of $q\Psix$, we showed how to obtain a connection matrix in Section \ref{sec:rhp}. Therefore, we obtain a solution of RHP I -- see \eqref{eq:rhpsolution}. We now describe how conversely, any solution of RHP I leads to a solution of $q\Psix$.

Take a connection matrix $C(z)\in \mathfrak{C}(\kappa,t_0)$ and suppose RHP I has a solution for at least one $m\in\mathbb{Z}$. 
We write
\begin{equation*}
    \mathfrak{M}:=\{m\in\mathbb{Z}:\Psi^{\B{m}}(z) \text{ exists}\}.
\end{equation*}
For $m\in\mathfrak{M}$, define $A(z,q^mt_0)$ by equation \eqref{eq:rhptolinear}. Due to the jump conditions of $\Psi^{\B{m}}(z)$ in RHP I, the matrix $A(z,q^mt_0)$ has trivial jumps on $\gamma^{\B{m}}$ and $q^{-1}\gamma^{\B{m}}$ and thus extends to a single-valued function on the complex $z$-plane. Furthermore, it follows from the global analyticity and invertibility of $\Psi^{\B{m}}(z)$, see Lemma \ref{lem:uniqueness}, that $A(z,q^mt_0)$ is entire. Finally, as $\Psi^{\B{m}}(z)=I+\mathcal{O}(z^{-1})$ as $z\rightarrow \infty$, it follows that $A(z,q^mt_0)$
is a degree two matrix polynomial satisfying
\begin{align*}
    A(z,q^mt_0)&=z^2\kappa_\infty^{\sigma_3}+\mathcal{O}(z)\quad (z\rightarrow \infty),\\
    A(0,q^mt_0)&=H(q^m t_0) q^m t_0\kappa_0^{\sigma_3} H(q^m t_0)^{-1},\quad H(q^m t_0):=\Psi^{\B{m}}(0),
    \end{align*}
and, due to equations \eqref{eq:rhpdeterminant} and \eqref{eq:rhptolinear},
\begin{equation*}
|A(z,q^mt_0)|=(z-\kappa_tq^m t_0)(z-\kappa_t^{-1}q^m t_0)(z-\kappa_1)(z-\kappa_1^{-1}).
\end{equation*}
Thus, $A(z,q^mt_0)$ is a coefficient matrix of the form \eqref{eq:matrix_polynomial}, for $m\in \mathfrak{M}$. By construction, the connection matrix associated with $A(z,q^mt_0)$ is given by $z^mC(z)$, $m\in \mathfrak{M}$.

For all $m\in \mathfrak{M}$, assume that 
\begin{equation}\label{eq:assumption}
   A_{12}(z,q^mt_0)\not\equiv 0.
\end{equation}
Then the corresponding coordinates $(f,g,w)$ are well-defined on $A$, via equations \eqref{eq:coordinates_linear}, and they form a solution of $q\Psix^{\text{aux}}(\kappa,t_0)$. Furthermore, we can read the values of $(f,g,w)$ directly from the solution $\Psi^{\B{m}}(z)$ of the RHP through formulas \eqref{eq:rhptofgw}. 

These formulas are derived as follows.
By expanding equation \eqref{eq:rhptolinear} around $z=\infty$, and considering the $(1,2)$ and $(1,1)$ entry, we respectively obtain
\begin{equation}\label{eq:wuderivation}
    w=(q^{-1}-\kappa_\infty^2)u_{12},\quad \alpha=(1-q^{-1})u_{11}-f.
\end{equation}
The first equation is precisely equation \eqref{eq:rhptofgw:w} for $w$.
The formula \eqref{eq:rhptofgw:f} for $f$ follows by subtracting \eqref{eq:diagonalisation_evolution3} from \eqref{eq:diagonalisation_evolution4} and solving for $f$.
By substituting $\alpha=(1-q^{-1})u_{11}-f$ in equation \eqref{eq:diagonalisation_evolution3} we obtain equation \eqref{eq:rhptofgw:g1} for $g_1$. Finally formula \eqref{eq:rhptofgw:g} for $g$ now follows from equation \eqref{eq:g1}.

We are now in a position to prove Lemma \ref{lem:monodromy_mapping_injective}.
\begin{proof}[Proof of Lemma \ref{lem:monodromy_mapping_injective}]
We have shown that, for any solution $(f,g)$ of $q\Psix(\kappa,t_0)$, there exists a connection matrix $C(z)\in\mathfrak{C}(\kappa,t_0)$, such that the values of $(f,g)$ may be read directly from the solution $\Psi^{\B{m}}(z)$ of RHP I in Definition \ref{def:RHPmain}, via equations \eqref{eq:rhptolinear}. Here $[C(z)]=\textsc{M}\in\mathcal{M}(\kappa,t_0)$ is the monodromy attached to $(f,g)$ via the monodromy mapping.

To prove the lemma, it remains to show be shown that these formulas are invariant under choosing a different representation $[\widetilde{C}(z)]=\textsc{M}$ of the monodromy, so that $(f,g)$ indeed only depends on the class $\textsc{M}$.  We proceed in proving this statement.

As $[\widetilde{C}(z)]=[C(z)]$, there exist invertible diagonal matrices $D_{1,2}$ such that
\begin{equation*}
    \widetilde{C}(z)=D_1C(z)D_2.
\end{equation*}
Thus, the solution $\widetilde{\Psi}^{\B{m}}(z)$ of RHP I, with $C(z)\rightarrow \widetilde{C}(z)$, is related to $\Psi^{\B{m}}(z)$ by
\begin{equation*}
    \widetilde{\Psi}^{\B{m}}(z)=\begin{cases} D_2^{-1}\Psi^{\B{m}}(z)D_2 & \text{if }z\in D_+^{\B{m}},\\
    D_2^{-1}\Psi^{\B{m}}(z)D_1^{-1} & \text{if }z\in D_-^{\B{m}}.\\
    \end{cases}
    \end{equation*}
Consequently, the matrix function $\widetilde{H}$ and $\widetilde{U}$, defined by equations  \eqref{eq:defiH} and \eqref{eq:defiU} for $\widetilde{\Psi}^{\B{m}}(z)$, are related to $H$ and $U$ by
\begin{equation*}
    \widetilde{H}(t)=D_2^{-1}H(t)D_1^{-1},\quad \widetilde{U}(t)=D_2^{-1}U(t)D_2.
\end{equation*}
The formulas \eqref{eq:rhptofgw:f} and \eqref{eq:rhptofgw:g} for $f$ and $g$ are invariant under such rescaling and the lemma follows.
\end{proof}

We finish this section with some remarks on assumption \eqref{eq:assumption}. Firstly, note that this is a necessary assumption for the coordinates $(f,g,w)$ to be well-defined. 
Now, suppose that $A_{12}(z,q^mt_0)\equiv 0$, for some $m\in \mathfrak{M}$, and write $t_m=q^m t_0$. Then, we have
\begin{equation*}
    A_{11}(z,t_m)=\kappa_\infty(z-v_1)(z-v_2),\quad A_{22}(z,t_m)=\kappa_\infty^{-1}(z-v_3)(z-v_4),
\end{equation*}
where, by equation \eqref{eq:detA},
\begin{equation*}
    \{v_1,v_2,v_3,v_4\}=\{\kappa_t^{+1}t_m,\kappa_t^{-1}t_m,\kappa_1^{+1},\kappa_1^{-1}\}.
\end{equation*}
Furthermore, as the eigenvalues of $A(0,t_m)$ are $\kappa_0^{\pm 1}t_m$, necessarily
\begin{equation*}
    \{\kappa_\infty v_1v_2,\kappa_\infty^{-1}v_3 v_4\}=\{A_{11}(0,t_m),A_{22}(0,t_m)\}=\{\kappa_0 t_m,\kappa_0^{-1}t_m\}.
\end{equation*}
By comparing the different possible values of $v_{1},\ldots, v_4$ in the above two equations, it follows that the parameters must satisfy
\begin{equation*}
  \kappa_0^{\epsilon_0}\kappa_t^{\epsilon_t} \kappa_1^{\epsilon_1} \kappa_\infty^{\epsilon_\infty}=1\quad \text{or}\quad \kappa_0^{\epsilon_0} \kappa_\infty^{\epsilon_\infty}t_m=1,
\end{equation*}
for some $\epsilon_j\in\{\pm 1\}$, $j=0,t,1,\infty$. So, at least one of the non-splitting conditions \eqref{eq:intro_irreducibleparameter} is violated.

Furthermore, from the defining equations of $\Psi_{0}$ and $\Psi_{\infty}$, equations \eqref{eq:linear_sys_solutions}, it follows that $\Psi_{\infty}(z,t_m)$ is lower-triangular and either $\left(\Psi_{0}\right)_{11}(z,t_m)$ or $\left(\Psi_{0}\right)_{12}(z,t_m)$ is identically zero. In particular, either $C_{12}(z)\equiv 0$ or $C_{22}(z)\equiv 0$, which means that $C(z)$ is reducible, see Definition \ref{def:reducible}.

We discuss RHP I with reducible monodromy in further detail in Section \ref{subsection:reducible}.

\section{Solvability, Reducible Monodromy and Orthogonal Polynomials}\label{sec:solvability}
In this section we study the solvability of RHP I, defined in Definition \ref{def:RHPmain}, and consequently the invertibility of the monodromy mapping introduced in Definition \ref{def:monodromy_mapping}.
In Section \ref{subsection:solvability}, we prove Lemma \ref{lem:irreducible} and Theorem \ref{thm:main_solvability}. In Section \ref{subsection:reducible}, we discuss RHP I with reducible monodromy.

\subsection{Solvability}\label{subsection:solvability}
We start this section by proving Lemma \ref{lem:irreducible}. To this end, we briefly recall some fundamental properties of $q$-theta functions, i.e. analytic functions $\theta(z)$ on $\mathbb{C}^*$ such that $\theta(z)/\theta(qz)$ is a monomial. 
For $\alpha\in\mathbb{C}^*$ and $n\in\mathbb{N}$, we denote by $V_n(\alpha)$ the set of all analytic functions $\theta(z)$ on $\mathbb{C}^*$, satisfying
\begin{equation}\label{eq:theta}
\theta(qz)=\alpha z^{-n}\theta(z).
\end{equation}
We note that $V_n(\alpha)$ is a vector space of dimension $n$ if $n\geq 1$, see e.g. \cite{phdroffelsen}.
 
 For $r\in\mathbb{R}_+$, we call
\begin{equation*}
	D_q(r):=\{|q|r\leq|z|<r\},
\end{equation*}
a fundamental annulus. As described in the following lemma, $q$-theta functions are, up to scaling, completely determined by the location of their zeros within any fixed fundamental annulus. 
\begin{lemma}\label{lem:classificationtheta}
	Let $\alpha\in\mathbb{C}^*$, $n\in\mathbb{N}$ and  $\theta(z)$ be a nonzero element of $V_n(\alpha)$. Then, within any fixed fundamental annulus, $\theta(z)$ has precisely $n$ zeros, counting multiplicity, say $\{a_1,\ldots,a_n\}$, and there exist unique $c\in\mathbb{C}^*$ and $s\in\mathbb{Z}$ such that
	\begin{equation}\label{eq:thetaformula}
	\theta(z)=c z^s\theta_q(z/a_1,\ldots,z/a_{n}),\quad \alpha=(-1)^n q^s a_1\cdot\ldots\cdot a_n.
	\end{equation}	
	Conversely, for any choice of the parameters, equation \eqref{eq:thetaformula} defines an element of $V_n(\alpha)$.
\end{lemma}
\begin{proof}
	See for instance \cite{phdroffelsen}.
\end{proof}

We proceed in proving Lemma \ref{lem:irreducible}. 
\begin{proof}[Proof of Lemma \ref{lem:irreducible}]
Take a connection matrix $C(z)\in\mathfrak{C}(\kappa,t_0)$ and suppose that $C(z)$ is reducible. Then $C(z)$ is triangular or anti-triangular.

Assume $C(z)$ is triangular, then
\begin{equation*}
C_{11}(z)C_{22}(z)=|C(z)|=c\theta_q(z\kappa_t t_0^{-1}  ,z\kappa_t^{-1} t_0^{-1},z\kappa_1,z\kappa_1^{-1}),
\end{equation*}
for some $c\in\mathbb{C}^*$, where the second equality follows from Definition \ref{def:connection_matrix_space}. Writing
\begin{equation*}
    (x_1,x_2,x_3,x_4)=(\kappa_t t_0,\kappa_t^{-1} t_0,\kappa_1 ,\kappa_1^{-1}),
\end{equation*}
it follows from Lemma \ref{lem:classificationtheta} that
\begin{equation}\label{eq:reducibleentries}
    C_{11}(z)=c_{11}\theta_q(z/x_i,z/x_j)z^n,\quad C_{22}(z)=c_{22}\theta_q(z/x_k,z/x_l)z^{-n},
\end{equation}
for some labeling $\{i,j,k,l\}=\{1,2,3,4\}$, $c_{11},c_{22}\in\mathbb{C}^*$ and $n\in\mathbb{Z}$.

Furthermore, by Definition \ref{def:connection_matrix_space},
\begin{equation*}
    \frac{C_{11}(qz)}{C_{11}(z)}=z^{-2}\frac{\kappa_0}{\kappa_\infty}t_0,\quad \frac{C_{22}(qz)}{C_{22}(z)}=z^{-2}\frac{\kappa_\infty}{\kappa_0}t_0,
\end{equation*}
which implies
\begin{equation}\label{eq:violation}
\frac{\kappa_0}{\kappa_\infty}t_0=x_i x_j q^n,\quad \frac{\kappa_\infty}{\kappa_0}t_0=x_k x_l q^{-n},
\end{equation}
violating the non-splitting conditions \eqref{eq:intro_irreducibleparameter}.

Similarly, if $C(z)$ is anti-triangular, then
\begin{equation}\label{eq:violation2}
\kappa_0 \kappa_\infty t_0=x_i x_j q^n,\quad \frac{1}{\kappa_0 \kappa_\infty}t_0=x_k x_l q^{-n},
\end{equation}
for some re-labeling $\{i,j,k,l\}=\{1,2,3,4\}$ and $n\in\mathbb{Z}$,  again violating the non-splitting conditions \eqref{eq:intro_irreducibleparameter}.

Conversely, if the non-splitting conditions \eqref{eq:intro_irreducibleparameter} do not hold true, then either equalities \eqref{eq:violation} or equalities \eqref{eq:violation2} can be realised by a re-labeling $\{i,j,k,l\}=\{1,2,3,4\}$, for some $n\in\mathbb{Z}$. In the former case, equations \eqref{eq:reducibleentries} with $C_{12}(z)\equiv C_{21}(z)\equiv 0$ define a reducible connection matrix in $\mathfrak{C}(\kappa,t_0)$.

It follows similarly that $\mathfrak{C}(\kappa,t_0)$ contains reducible monodromy in the latter case and the lemma follows.
\end{proof}

To study the solvability of RHP I, in Definition \ref{def:RHPmain}, it is helpful to consider the following slightly more general RHP.
\begin{definition}[RHP II]\label{def:RHPgeneral}
Given a connection matrix $C\in \mathfrak{C}(\kappa,t_0)$ and a family of admissable curves $(\gamma^{\B{m}})_{m\in\mathbb{Z}}$,
 for $m,n\in\mathbb{Z}$, find a matrix function $\Psi^{\B{m,n}}(z)$ which satisfies the following conditions.
  \begin{enumerate}[label={{\rm (\roman *)}}]
  \item $\Psi^{\B{m,n}}(z)$ is analytic on $\mathbb{C}\setminus\gamma^{\B{m}}$.
    \item $\Psi^{\B{m,n}}(z')$ has continuous boundary values $\Psi_-^{\B{m,n}}(z)$ and $\Psi_+^{\B{m,n}}(z)$ as $z'$ approaches $z\in \gamma^{\B{m}}$ from $D_-^{\B{m}}$ and $D_+^{\B{m}}$ respectively, related by
		\begin{equation*}
		\Psi_+^{\B{m,n}}(z)=\Psi_-^{\B{m,n}}(z)z^mC(z),\quad z\in \gamma^{\B{m}}.
              \end{equation*}
            \item $\Psi^{\B{m,n}}(z)$ satisfies
              \begin{equation*}
		\Psi^{\B{m,n}}(z)=\left(I+\mathcal{O}\left(z^{-1}\right)\right)z^{n\sigma_3}\quad z\rightarrow \infty.
              \end{equation*}
              \end{enumerate}
  \end{definition}
By comparison with RHP I in Definition \ref{def:RHPmain}, we can identify $\Psi^{\B{m,0}}(z)=\Psi^{\B{m}}(z)$. More generally, for any fixed $n\in\mathbb{Z}$, RHP II is equivalent to RHP I, with $C(z)$ replaced by $C(z)z^{-n\sigma_3}$.
In particular, we have the following analog of Lemma \ref{lem:uniqueness}.
 \begin{lemma}\label{lem:uniquenessII}
For any fixed $m,n\in\mathbb{Z}$, if RHP II in Definition \ref{def:RHPgeneral} has a solution $\Psi^{\B{m,n}}(z)$, then this solution is globally invertible on the complex plane and unique.
\end{lemma}
\begin{proof}
The proof is analogous to that of Lemma \ref{lem:uniqueness}.
\end{proof}
Given the uniqueness in the above lemma, we say that $\Psi^{\B{m,n}}(z)$ exists if and only if RHP II has a solution for that value of $m,n\in\mathbb{Z}$.

The main reason for considering the more general RHP above, is that we have the following result due to Birkhoff \cite{birkhoffgeneralized1913}.
\begin{lemma}\label{lem:existence}
For any fixed $m\in\mathbb{Z}$, the solution $\Psi^{\B{m,n}}(z)$ to RHP II, in Definition \ref{def:RHPgeneral}, exists for at least one $n\in\mathbb{Z}$.
\end{lemma}
\begin{proof}
See Birkhoff \cite{birkhoffgeneralized1913}[\textsection 21] or the proof of Lemma 4.4 in \cite{joshiroffelseniv}.
\end{proof}

Our next step is to study the dynamics of $\Psi^{\B{m,n}}(z)$ as $n$ varies, with the ultimate goal to obtain criteria for the existence of $\Psi^{\B{m,n}}(z)$ at $n=0$, as these will allow us to prove solvability of RHP I and thus prove Theorem \ref{thm:main_solvability}.

To this end, if $\Psi^{\B{m,n}}(z)$ exists, we denote its expansion around $z=\infty$ by
\begin{equation}\label{eq:expansioninfty}
\Psi^{\B{m,n}}(z)=\left(I+z^{-1}U^{\B{m,n}}+z^{-2}V^{\B{m,n}}+z^{-3}W^{\B{m,n}}+\mathcal{O}(z^{-4})\right)z^{n\sigma_3},
\end{equation}
as $z\rightarrow \infty$,
and associate a coefficient matrix $A^{\B{m,n}}(z)$
as in equation \eqref{eq:rhptolinear},
\begin{equation}\label{eq:rhtolineargeneral}
    A^{\B{m,n}}(z)=\begin{cases}
    z^2 \Psi^{\B{m,n}}(qz)\kappa_\infty^{\sigma_3} \Psi^{\B{m,n}}(z)^{-1} & \text{if } z\in q^{-1}(D_+^{\B{m}}\cup \gamma^{\B{m}}),\\
    q^m t_0 \Psi^{\B{m,n}}(qz)\kappa_0^{\sigma_3}C(z) \Psi^{\B{m,n}}(z)^{-1}
    & \text{if } z\in D_+^{\B{m}}\cap q^{-1}D_-^{\B{m}},\\
    q^m t_0 \Psi^{\B{m,n}}(qz)\kappa_0^{\sigma_3} \Psi^{\B{m,n}}(z)^{-1}
    & \text{if } z\in D_-^{\B{m}}\cup \gamma^{\B{m}}.
    \end{cases}
\end{equation}
Then $A^{\B{m,n}}(z)$ is a degree two matrix polynomial of the form \eqref{eq:matrix_polynomial} except for a generally different normalisation at $z=\infty$,
\begin{equation*}
    A^{\B{m,n}}(z)=z^2 (q^n\kappa_\infty)^{\sigma_3}+\mathcal{O}(z)\quad (z\rightarrow \infty).
\end{equation*}
In particular, the corresponding coordinates $f^{\B{n}}(q^m t_0)$, $g^{\B{n}}(q^m t_0)$ and $w^{\B{n}}(q^m t_0)$ define a solution of $q\Psix(\kappa^{\B{n}},t_0)$ with
\begin{equation*}
    \kappa^{\B{n}}=(\kappa_0,\kappa_t,\kappa_1,q^n\kappa_\infty),
\end{equation*}
if RHP II is solvable in $m$ for that value of $n$.

We have the following lemma regarding solvability of RHP II as $n$ varies.

\begin{lemma}\label{lem:ndependencerhp}
Fix $m,n\in\mathbb{Z}$ and suppose that the solution $\Psi^{\B{m,n}}(z)$ of RHP II in Definition \ref{def:RHPgeneral} exists. Then, recalling the definition of the matrices $U=(u_{ij})$ and $V=(v_{ij})$ in equation \eqref{eq:expansioninfty}, either
\begin{enumerate}[label={{\rm (\roman *)}}]
    \item $u_{12}^{\B{m,n}}\neq 0$, in which case $\Psi^{\B{m,n+1}}(z)$ exists.
    \item $u_{12}^{\B{m,n}}= 0$ but $v_{12}^{\B{m,n}}\neq 0$, in which case $\Psi^{\B{m,n+1}}(z)$ does not exist but $\Psi^{\B{m,n+2}}(z)$ does exist.
    \item $u_{12}^{\B{m,n}}= 0$ and $v_{12}^{\B{m,n}}= 0$, in which case $\Psi^{\B{m,n+k}}(z)$ does not exist for any $k>0$ and necessarily $C_{12}(z)\equiv 0$ or $C_{22}(z)\equiv 0$.
\end{enumerate}
Similarly, either
\begin{enumerate}[label={{\rm (\Roman *)}}]
    \item $u_{21}^{\B{m,n}}\neq 0$, in which case $\Psi^{\B{m,n-1}}(z)$ exists.
    \item $u_{21}^{\B{m,n}}= 0$ but $v_{21}^{\B{m,n}}\neq 0$, in which case $\Psi^{\B{m,n-1}}(z)$ does not exist but $\Psi^{\B{m,n-2}}(z)$ does exist.
    \item $u_{21}^{\B{m,n}}= 0$ and $v_{21}^{\B{m,n}}= 0$, in which case $\Psi^{\B{m,n-k}}(z)$ does not exist for any $k>0$ and necessarily $C_{11}(z)\equiv 0$ or $C_{21}(z)\equiv 0$.
\end{enumerate}
\end{lemma}
\begin{proof}
We start with the fundamental observation that, for any $k\in\mathbb{Z}$, the solution $\Psi^{\B{m,n+k}}(z)$ exists if and only if there exists a matrix polynomial $R(z)$ which satisfies
\begin{equation}\label{eq:Rrelation}
    R(z)\Psi^{\B{m,n}}(z)=(I+\mathcal{O}(z^{-1}))z^{(n+k)\sigma_3}.
\end{equation}
Indeed, if such a matrix $R(z)$ exists, then $\Psi^{\B{m,n+k}}(z)=R(z)\Psi^{\B{m,n}}(z)$ solves RHP II. Conversely, suppose $\Psi^{\B{m,n+k}}(z)$ exists, define
\begin{equation*}
R(z)=\Psi^{\B{m,n+k}}(z)\Psi^{\B{m,n}}(z)^{-1},
\end{equation*}
then $R(z)$ has a trivial jump on $\gamma^{\B{m}}$ and consequently extends to an analytic matrix function on the whole complex plane, satisfying
\begin{equation*}
    R(z)=(I+\mathcal{O}(z^{-1}))z^{k\sigma_3}(I+\mathcal{O}(z^{-1}))\quad (z\rightarrow \infty).
\end{equation*}
It follows that $R(z)$ is a matrix polynomial and equation \eqref{eq:Rrelation} follows directly from the normalisation of $\Psi^{\B{m,n+k}}(z)$ at $z=\infty$.

By the above observation, the existence of $\Psi^{\B{m,n+k}}(z)$ can be studied through examining the solvability of equation \eqref{eq:Rrelation}, which is how we proceed in establishing the lemma.
 
Firstly, we consider $k=1$. The matrix $R(z)$ must take the form
\begin{equation*}
    R(z)=z\begin{pmatrix}1 & 0\\
    0 & 0\end{pmatrix}+
    \begin{pmatrix}r_{11} & r_{12}\\
    r_{21} & 0\end{pmatrix},
\end{equation*}
and equation \eqref{eq:Rrelation} reduces to the following linear system of equations,
\begin{equation*}
    \begin{pmatrix}
    0 & 1 & 0\\
    u_{12}^{\B{m,n}} & u_{22}^{\B{m,n}} & 0 \\
    0 & 0 & u_{12}^{\B{m,n}} \\
    \end{pmatrix}
    \begin{pmatrix}
    r_{11}\\
    r_{12}\\
    r_{21}\\
    \end{pmatrix}=
    \begin{pmatrix}
    -u_{12}^{\B{m,n}}\\
    -v_{12}^{\B{m,n}}\\
    1\\
    \end{pmatrix}
\end{equation*}
This system is solvable if and only if $u_{12}^{\B{m,n}}\neq 0$. Consequently, $\Psi^{\B{m,n+1}}(z)$ exists if and only if $u_{12}^{\B{m,n}}\neq 0$. This establishes part (i) of the lemma.

Next, assume $u_{12}^{\B{m,n}}=0$ and we proceed in studying the solvability of equation \eqref{eq:Rrelation} with $k=2$. The matrix $R(z)$ must take the form
\begin{equation*}
    R(z)=z^2\begin{pmatrix}1 & 0\\
    0 & 0\end{pmatrix}+z
    \begin{pmatrix}r_{11}^{\B{1}} & 0\\
    r_{21}^{\B{1}} & 0\end{pmatrix}+
    \begin{pmatrix}r_{11}^{\B{0}} & r_{12}^{\B{0}}\\
    r_{21}^{\B{0}} & 0\end{pmatrix},
\end{equation*}
and \eqref{eq:Rrelation} reduces to
\begin{equation*}
\begin{pmatrix}
0 & 1 & 0 & 0 & 0\\
0 & u_{22}^{\B{m,n}} & 0 & v_{12}^{\B{m,n}} & 0\\
0 & 0 & 0 & 0 & v_{12}^{\B{m,n}}\\
v_{12}^{\B{m,n}} & v_{22}^{\B{m,n}} & 0 & w_{12}^{\B{m,n}} & 0\\
0 & 0 & v_{12}^{\B{m,n}} & 0 & w_{12}^{\B{m,n}}
\end{pmatrix}
\begin{pmatrix}
    r_{11}^{\B{0}}\\
    r_{12}^{\B{0}}\\
    r_{21}^{\B{0}}\\
    r_{11}^{\B{1}}\\
    r_{21}^{\B{1}}\\
    \end{pmatrix}=
    \begin{pmatrix}
    -v_{12}^{\B{m,n}}\\
    -w_{12}^{\B{m,n}}\\
    0\\
    0\\
    1\\
    \end{pmatrix}.
\end{equation*}
It follows from direct computation that the above linear system has a solution if and only if $v_{12}^{\B{m,n}}\neq 0$. We therefore conclude that, if $u_{12}^{\B{m,n}}=0$, then $\Psi^{\B{m,n+2}}(z)$ exists if and only if $v_{12}^{\B{m,n}}\neq 0$. This establishes part (ii) of the lemma.

Finally, consider the case when both $u_{12}^{\B{m,n}}=0$ and $v_{12}^{\B{m,n}}=0$. Then it follows directly from equation \eqref{eq:rhtolineargeneral} that the entry $A_{12}^{\B{m,n}}(z)$ of the matrix polynomial $A^{\B{m,n}}(z)$
is identically zero, by considering its expansion around $z=\infty$. Furthermore, as
\begin{equation*}
\Psi^{\B{m,n}}(qz)=z^{-2}A^{\B{m,n}}(z) \Psi^{\B{m,n}}(z)\kappa_\infty^{-\sigma_3},
\end{equation*}
for $z\in q^{-1}D_+^{\B{m}}$, it follows that $\Psi_{12}^{\B{m,n}}(z)\equiv 0$ on $D_+^{\B{m}}$. 

Now, consider equation \eqref{eq:Rrelation} for any $k>0$. Its $(2,2)$-entry reads
\begin{equation}\label{eq:R22unsolvable}
 R_{22}(z)\Psi_{22}^{\B{m,n}}(z)=z^{-n-k}(1+\mathcal{O}(z^{-1})),
\end{equation}
as $z\rightarrow \infty$. However, recall that
\begin{equation*}
\Psi_{22}^{\B{m,n}}(z)=z^{-n}(1+\mathcal{O}(z^{-1})),
\end{equation*}
and thus equation \eqref{eq:R22unsolvable} has no polynomial solution $R_{22}(z)$. It follows that $\Psi^{\B{m,n+k}}(z)$ does not exist, for any $k>0$.

Finally, we prove that one of the entries of $C(z)$ must be identically zero. To this end, note that
\begin{equation}\label{eq:qdifpsi}
\Psi^{\B{m,n}}(qz)=q^m t_0A^{\B{m,n}}(z) \Psi^{\B{m,n}}(z)\kappa_0^{-\sigma_3},
\end{equation}
for $z\in D_-^{\B{m}}$. There are two options, either
\begin{equation*}
    A_{11}^{\B{m,n}}(0)=q^m t_0\kappa_0,\quad A_{22}^{\B{m,n}}(0)=q^m t_0\kappa_0^{-1},
\end{equation*}
in which case it follows from \eqref{eq:qdifpsi} that $\Psi_{12}^{\B{m,n}}(z)\equiv 0$ on $D_-^{\B{m}}$ and consequently that $C_{12}(z)\equiv 0$; or
\begin{equation*}
    A_{11}^{\B{m,n}}(0)=q^m t_0\kappa_0^{-1},\quad A_{22}^{\B{m,n}}(0)=q^m t_0\kappa_0,
\end{equation*}
in which case it follows from \eqref{eq:qdifpsi} that $\Psi_{11}^{\B{m,n}}(z)\equiv 0$ on $D_-^{\B{m}}$ and consequently that $C_{22}(z)\equiv 0$. This proves part (iii) of the lemma.

Parts (I)-(III) of the lemma are proven analogously.
\end{proof}

Note that we have the following immediate corollary from Lemmas \ref{lem:existence} and \ref{lem:ndependencerhp}.
\begin{corollary}\label{cor:existencerhpn}
Consider RHP II in Definition \ref{def:RHPgeneral} and assume $C(z)$ is irreducible. Then, for any $m,n\in\mathbb{Z}$, the solution $\Psi^{\B{m,n}}(z)$ or $\Psi^{\B{m,n+1}}(z)$ exists.
\end{corollary}

We now have all the ingredients to prove Theorem \ref{thm:main_solvability}.

\begin{proof}[Proof of Theorem \ref{thm:main_solvability}]
Take an irreducible connection matrix $C(z)\in \mathfrak{C}(\kappa,t_0)$. In RHP II, see Definition \ref{def:RHPgeneral}, we have an additional integer parameter $n$ and we denote its solution by $\Psi^{\B{m,n}}(z)$, when it exists. For $n=0$, this RHP is precisely RHP I. Proving the first part of Theorem \ref{thm:main_solvability}, is thus equivalent to showing that, for any fixed $m\in\mathbb{Z}$, the solution $\Psi^{\B{m,0}}(z)$ or $\Psi^{\B{m+1,0}}(z)$ of RHP II exists. We do this via a proof by contradiction.

Take  $m\in\mathbb{Z}$ and suppose that neither 
$\Psi^{\B{m,0}}(z)$ nor $\Psi^{\B{m+1,0}}(z)$ exists. As $C(z)$ is irreducible,  Corollary \ref{cor:existencerhpn} implies that $\Psi^{\B{m,-1}}(z)$ and $\Psi^{\B{m+1,-1}}(z)$ necessarily exist.

To deduce a contradiction, we define the following matrix function
\begin{equation}\label{eq:definitionB}
    B(z)=\begin{cases}
    \Psi^{\B{m+1,-1}}(z)\Psi^{\B{m,-1}}(z)^{-1} & \text{if }z\in D_+^{\B{m}}\cup \gamma^{\B{m}},\\
    \Psi^{\B{m+1,-1}}(z)z^{-m}C(z)^{-1}\Psi^{\B{m,-1}}(z)^{-1} & \text{if }z\in D_-^{\B{m}}\cap D_+^{\B{m+1}},\\
    z\Psi^{\B{m+1,-1}}(z)\Psi^{\B{m,-1}}(z)^{-1} & \text{if }z\in D_-^{\B{m+1}}\cup \gamma^{\B{m+1}}.
    \end{cases}
\end{equation}
The jump conditions of RHP II that $\Psi^{\B{m+1,-1}}(z)$ and $\Psi^{\B{m,-1}}(z)$ satisfy imply that $B(z)$  has only trivial jumps on $\gamma^{\B{m}}$ and $\gamma^{\B{m+1}}$. 
Consequently, $B(z)$ extends to a meromorphic function on the complex plane.

The only possible source of singularities (i.e., poles) on the right side of Equation \eqref{eq:definitionB}, is the term $C(z)^{-1}$. (Note that  $\Psi^{(m,n)}$ are analytic functions of $z$, which moreover are  invertible for all $z$, see Lemma \ref{lem:uniquenessII}.) In $D_-^{\B{m}}\cap D_+^{\B{m+1}}$, we know that the determinant of $C(z)$ only vanishes at $z=\kappa_t^{\pm 1}q^{m+1} t_0$, so that $C(z)^{-1}$ has (simple) poles there. Therefore, $B(z)$ has simple poles at $z=\kappa_t^{\pm 1}q^{m+1} t_0$. This, combined with the fact that $B(0)=0$ and $B(\infty)=I$, yields
\begin{equation*}
    B(z)=\frac{z^2 I+z B_0}{(z-\kappa_tq^{m+1} t_0)(z-\kappa_t^{-1}q^{m+1} t_0)},
\end{equation*}
for a constant matrix $B_0$\footnote{Note that this is essentially the derivation of the forward implication of Theorem \ref{thm:isomonodromic}.}.

We now turn our attention to the coefficient matrices $A^{\B{m,-1}}(z)$ and $A^{\B{m+1,-1}}(z)$ related to $\Psi^{\B{m,-1}}(z)$ and $\Psi^{\B{m+1,-1}}(z)$
via equation \eqref{eq:rhtolineargeneral}. It follows from the defining equation of $B(z)$, equation \eqref{eq:definitionB}, that these coefficient matrices are related by
\begin{equation}\label{eq:compatiblerhp}
    A^{\B{m+1,-1}}(z)B(z)=B(qz)A^{\B{m,-1}}(z).
\end{equation}
To deduce this, it suffices to note that, for $z\in q^{-1}D_+^{\B{m}}$, compatibility of the first rows of the right-hand sides of equations \eqref{eq:rhtolineargeneral} and \eqref{eq:definitionB}, yields equation \eqref{eq:compatiblerhp}. By analytic continuation, equation \eqref{eq:compatiblerhp} holds globally.

Now, recall that $\Psi^{\B{m,n}}(z)$ has an asymptotic expansion at infinity, see equation \eqref{eq:expansioninfty}, of the form
\begin{equation*}
\Psi^{\B{m,n}}(z)=\left(I+z^{-1}U^{\B{m,n}}+z^{-2}V^{\B{m,n}}+\mathcal{O}(z^{-3})\right)z^{n\sigma_3},\quad U=(u_{ij}),\quad V=(v_{ij}).
\end{equation*}
Due to part (i) of Lemma \ref{lem:ndependencerhp}, we know that $u_{12}^{\B{m,-1}}=0$ and $u_{12}^{\B{m+1,-1}}=0$. We will proceed in showing that also $v_{12}^{\B{m,-1}}=0$, which, due to part (iii) of Lemma \ref{lem:ndependencerhp}, means that $C(z)$ is not irreducible, giving us the desired contradiction.

To get there, we first note that, by considering the expansion of $B(z)$ as $z\rightarrow \infty$ in the first row of equation \eqref{eq:definitionB}, we obtain $(B_0)_{12}=0$.

Similarly, as $u_{12}^{\B{m,-1}}=0$, it follows from the first row of the right-hand side of equation \eqref{eq:rhtolineargeneral} that the $(1,2)$ entry of $A$ satisfies $A_{12}^{\B{m,-1}}(z)=\mathcal{O}(1)$ as $z\rightarrow \infty$. Namely
\begin{equation}\label{eq:A12constant}
    A_{12}^{\B{m,-1}}(z)\equiv c,
\end{equation}
where $c$ is a constant.

We now show that, equations \eqref{eq:compatiblerhp}, \eqref{eq:A12constant} and the fact that $(B_0)_{12}=0$ imply that $v_{12}^{\B{m,-1}}=0$.

Firstly, by comparing the determinants of the left and right-hand sides of equation \eqref{eq:compatiblerhp}, we obtain
\begin{equation*}
    |zI+B_0|=(z-\kappa_tq^{m+1} t_0)(z-\kappa_t^{-1}q^{m+1} t_0).
\end{equation*}
As $(B_0)_{12}=0$, this implies the following dichotomy: either
\begin{enumerate}[label=(\Roman*)]
    \item $B_0=\begin{pmatrix}
    -\kappa_tq^{m+1} t_0 & 0\\
    b_{21} & -\kappa_t^{-1}q^{m+1} t_0
    \end{pmatrix}$, \hspace{0.5cm}or
    \item $B_0=\begin{pmatrix}
    -\kappa_t^{-1}q^{m+1} t_0 & 0\\
    b_{21} & -\kappa_t q^{m+1} t_0
    \end{pmatrix}$,
\end{enumerate}
for some $b_{21}\in\mathbb{C}$.

Secondly, the left-hand side of equation \eqref{eq:compatiblerhp} is analytic at $z=\kappa_t^{\pm 1}q^{m} t_0$, but $B(qz)$, on the right-hand side, has a pole at those two points. This means that
\begin{equation}\label{eq:BArelation}
    (\kappa_t^{\pm 1}q^{m+1}t_0 I+B_0)A^{\B{m,-1}}(\kappa_t^{\pm 1}q^{m}t_0)=0.
\end{equation}
We now consider the $(1,2)$-entry of equation \eqref{eq:BArelation} for the two choices of the sign $\pm$. The positive choice leads to a tautology in Case (I), while the negative choice gives
\begin{equation*}
    (\kappa_t^{-1}-\kappa_t)q^{m+1} t_0 c=0.
\end{equation*}
On the other hand, the positive choice in Case (II) gives 
\begin{equation*}
    (\kappa_t-\kappa_t^{-1})q^{m+1} t_0 c=0,
\end{equation*}
while the negative choice is a tautology.
Due to the non-resonance conditions \eqref{eq:non_res}, $\kappa_t^2\neq 1$, and so it follows from the above results that $c=0$. Therefore, $A_{12}^{\B{m,-1}}(z)$ is identically zero, by equation \eqref{eq:A12constant}.

Since $A$ is lower triangular, it follows from equation \eqref{eq:rhtolineargeneral} that $\Psi^{\B{m,-1}}(z)$ must be lower triangular for $z\in D_+^{\B{m}}$. In particular, $u_{12}^{\B{m,-1}}=v_{12}^{\B{m,-1}}=0$, which, due to part (iii) of Lemma \ref{lem:ndependencerhp}, means that $C(z)$ is not irreducible, giving us the desired contradiction. 

We conclude that solution $\Psi^{\B{m}}(z)$ or $\Psi^{\B{m+1}}(z)$ of RHP I exists for any $m\in\mathbb{Z}$, establishing the first part of the theorem.

Let $(f,g,w)$ be the corresponding solution of $q\Psix^{aux}(\kappa,t_0)$ via \eqref{eq:rhptofgw}. The second part of the theorem asserts that for $m\in\mathbb{Z}$, $\Psi^{\B{m}}(z)$ fails to exist if and only if
 $(f(t_m),g(t_m))=(\infty,\kappa_\infty)$.
 
So, suppose $m\in\mathbb{Z}$ is such that $\Psi^{\B{m}}(z)$ fails to exist. If $(f(t_m),g(t_m))\neq (\infty,\kappa_\infty)$, then it follows from Lemma \ref{lem:isomonodromic} that the coefficient matrix $A(z,t_m)$ is well-defined. But then equation \eqref{eq:rhpsolution} would yield a solution of RHP I, that is, $\Psi^{\B{m}}(z)$ exists, which contradicts our assumption. Thus $(f(t_m),g(t_m))= (\infty,\kappa_\infty)$.

On the other hand, if $\Psi^{\B{m}}(z)$ exists, then $A(z,t_m)$ is well-defined, via equation \eqref{eq:rhptolinear}, and consequently $(f(t_m),g(t_m))\neq (\infty,\kappa_\infty)$, by Lemma \ref{lem:isomonodromic}. So, indeed, $\Psi^{\B{m}}(z)$ fails to exist if and only if
 $(f(t_m),g(t_m))=(\infty,\kappa_\infty)$. This completes the proof of the theorem.
\end{proof}

\subsection{Reducible monodromy, orthogonal polynomials and special function solutions} \label{subsection:reducible}

In the case of $\Psix$, it is well-known that reducible monodromy yields special function solutions -- see Mazzocco \cite{mazzoccorational}. Furthermore, in such case, the solution of the standard RHP for $\Psix$, when solvable, can be solved explicitly in terms of certain orthogonal polynomials \cites{dubrovinkapaev,fokasitskapaev}.

In this subsection, we show that the same phenomenon occurs for $q\Psix$. Recall that the monodromy manifold contains reducible monodromy if and only if conditions \eqref{eq:intro_irreducibleparameter1} or conditions  \eqref{eq:intro_irreducibleparameter2} are violated.
We discuss one example from each of these two sets of non-splitting conditions.

Firstly, we consider the case where
\begin{equation*}
    \kappa_0=\kappa_t\kappa_1 \kappa_\infty,
\end{equation*}
violating one of the conditions in \eqref{eq:intro_irreducibleparameter1}, and consider RHP II, defined in Definition \ref{def:RHPgeneral}, with the following upper-triangular connection matrix $C(z)\in\mathfrak{C}(\kappa,t_0)$,
\begin{equation}\label{eq:triangularconnection}
    C(z)=\begin{pmatrix}
    \theta_q\left(\frac{z}{\kappa_t t_0},\frac{z}{\kappa_1}\right) & c\,\theta_q\left(\frac{z}{\nu t_0},\frac{z\nu}{\kappa_0\kappa_\infty}\right)\\
    0 & \theta_q\left(\frac{z\kappa_t}{t_0},z\kappa_1\right)
    \end{pmatrix}.
\end{equation}
Here $c\in\mathbb{C}$ and $\nu\in\mathbb{C}^*$ are two monodromy datums that can be chosen at pleasure.

Writing $t_m=q^m t_0$, the jump matrix of $\Psi^{\B{m,n}}(z)$ in RHP II can be written as
\begin{equation*}
    z^mC(z)=(-1)^m q^{\frac{1}{2}m(m+1)}t_0^m\begin{pmatrix}
    \kappa_t^m\theta_q\left(\frac{z}{\kappa_t t_m},\frac{z}{\kappa_1}\right) & c\nu^m\theta_q\left(\frac{z}{\nu t_m},\frac{z\nu}{\kappa_0\kappa_\infty}\right)\\
    0 & \kappa_t^{-m}\theta_q\left(\frac{z\kappa_t}{t_m},z\kappa_1\right)
    \end{pmatrix}.
\end{equation*}
We bring RHP II into the standard Fokas-Its-Kitaev RHP form \cites{fokasitskitaev1,fokasitskitaev2} for orthogonal polynomials, by applying a transformation
\begin{equation*}
  Y^{\B{m,n}}(z)=\begin{cases}
  D_1^{-1}\Psi^{\B{m,n}}(z)F_\infty^{\B{m}}(z)^{-1}D_1  & \text{ if $z\in D_+^{\B{m}}$,}\\
   D_1^{-1}\Psi^{\B{m,n}}(z)F_0^{\B{m}}(z)^{-1}D_2 & \text{ if $z\in D_-^{\B{m}}$,}\\
  \end{cases}
\end{equation*}
where $D_1$ and $D_2$ diagonal matrices and $F_0^{\B{m}}(z)$ and $F_\infty^{\B{m}}(z)$ analytic and invertible matrix functions on respectively $D_-^{\B{m}}$ and $D_+^{\B{m}}$.

After such a transformation, the jump matrix of $Y^{\B{m,n}}(z)$ reads
\begin{equation*}
    J^{\B{m}}(z)=D_2^{-1}F_0^{\B{m}}(z)z^m C(z)F_\infty^{\B{m}}(z)^{-1}D_1,
\end{equation*}
and we wish to choose $D_{1,2}$ and $F_{0,\infty}$ such that this jump matrix is upper-triangular with diagonal entries constant and equal to $1$.
To this end, we choose $F_{0,\infty}$ so that they cancel the $q$-theta functions on the diagonal,
\begin{align*}
F_\infty^{\B{m}}(z)&=
\begin{pmatrix}
\left(\frac{q\kappa_tt_m}{z},\frac{q\kappa_1}{z};q\right)_\infty & 0\\
0 & \left(\frac{qt_m}{\kappa_t z},\frac{q}{\kappa_1z};q\right)_\infty\\
\end{pmatrix},\\
F_0^{\B{m}}(z)&=\begin{pmatrix}
\left(\frac{z}{\kappa_t t_m},\frac{z}{\kappa_1};q\right)_\infty & 0\\
0 & \left(\frac{z\kappa_t}{t_m},z\kappa_1;q\right)_\infty\\
\end{pmatrix},
\end{align*}
and we choose $D_1$ and $D_2$ to normalise the now constant diagonal entries so that they equal $1$,
\begin{equation*}
    D_1=\begin{pmatrix}
    \nu^{m} & 0\\
    0 & \kappa_t^{m}\\
    \end{pmatrix},\quad
    D_2=(-1)^m q^{\frac{1}{2}m(m+1)}t_0^m\begin{pmatrix}
    \nu^{m}\kappa_t^m & 0\\
    0 & 1\\
    \end{pmatrix}.
\end{equation*}
Then the jump matrix reads
\begin{equation*}
    J^{\B{m}}(z)=\begin{pmatrix}
    1 & w(z,t_m)\\
    0 & 1
    \end{pmatrix},
\end{equation*}
where
\begin{equation}\label{eq:weight_function}
     w(z,t)=\frac{c\,\theta_q\left(\frac{z}{\nu t},\frac{z \nu}{\kappa_0\kappa_\infty}\right)}{\left(\frac{z}{\kappa_t t},\frac{z}{\kappa_1};q\right)_\infty\left(\frac{qt}{\kappa_t z},\frac{q}{\kappa_1z};q\right)_\infty},
\end{equation}
and $Y^{\B{m,n}}(z)$ solves the following RHP, if it exists.
\begin{definition}[RHP III]\label{def:RHPorthogonal}
For $m,n\in\mathbb{Z}$, find a matrix function $Y^{\B{m,n}}(z)$ which satisfies the following conditions.
  \begin{enumerate}[label={{\rm (\roman *)}}]
  \item $Y^{\B{m,n}}(z)$ is analytic on $\mathbb{C}\setminus\gamma^{\B{m}}$.
    \item $Y^{\B{m,n}}(z')$ has continuous boundary values $Y_-^{\B{m,n}}(z)$ and $Y_+^{\B{m,n}}(z)$ as $z'$ approaches $z\in \gamma^{\B{m}}$ from $D_-^{\B{m}}$ and $D_+^{\B{m}}$ respectively, related by
		\begin{equation*}
		Y_+^{\B{m,n}}(z)=Y_-^{\B{m,n}}(z)\begin{pmatrix}
    1 & w(z,t_m)\\
    0 & 1
    \end{pmatrix}\qquad (z\in \gamma^{\B{m}}),
              \end{equation*}
    where $w(z,t)$ is the weight function defined in equation \eqref{eq:weight_function}.      
            \item $Y^{\B{m,n}}(z)$ satisfies
              \begin{equation*}
		Y^{\B{m,n}}(z)=\left(I+\mathcal{O}\left(z^{-1}\right )\right)z^{n\sigma_3}\quad z\rightarrow \infty.
              \end{equation*}
              \end{enumerate}
  \end{definition}
 RHP III is the standard Fokas-Its-Kitaev RHP for orthogonal polynomials on the contour $\gamma^{\B{m}}$ with respect to the weight function $w(z,t_m)$. We refer to Deift \cite{deiftorthogonal} for more background information on the theory of orthogonal polynomials and corresponding RHPs.

We proceed to draw some immediate conclusions from the equivalence between the RHPs II and III, given in Definitions \ref{def:RHPgeneral} and  \ref{def:RHPorthogonal} respectively, and the theory of orthogonal polynomials.
 If $n<0$, then RHP III is unsolvable for every $m\in\mathbb{Z}$ and thus the same holds true for RHP II.
 
 When $n=0$, RHP III is solvable for every $m\in\mathbb{Z}$ and the solution is explicitly given by
 \begin{equation*}
    Y^{\B{m,0}}(z)=\begin{pmatrix}
    1 & -\mathcal{C}^{\B{m}}\left[w(\cdot,t_m)\right](z)\\
    0 & 1
    \end{pmatrix},
 \end{equation*}
 where $\mathcal{C}^{\B{m}}$ denotes the Cauchy operator on  $\gamma^{\B{m}}$,
 \begin{equation*}
	\mathcal{C}^{\B{m}}\left[h(\cdot)\right](z)=\frac{1}{2\pi i}\oint_{\gamma^{\B{m}}}\frac{h(x)}{x-z}dx\quad (h(\cdot)\in L^2(\gamma^{\B{m}})).
	\end{equation*}
 
 When $n>0$, RHP III is solvable if and only if the Hankel determinant of moments
\begin{equation*}
 \Delta_n(t_m):=\begin{vmatrix} 
 \mu_{0} & \mu_{1} & \ldots & \mu_{n-1}\\
 \mu_{1} & \mu_{2} &\ldots  & \mu_{n}\\
 \vdots & \vdots  &\ddots & \vdots\\
 \mu_{n-1} & \mu_n & \ldots & \mu_{2n-2}
 \end{vmatrix},\quad \mu_k:=\frac{1}{2\pi i}\oint_{\gamma^{\B{m}}}z^kw(z,t_m)dz\quad (k\in\mathbb{Z}),
 \end{equation*}
is nonzero, in which case the solution of the RHP is explicitly given by
\begin{equation*}
    Y^{\B{m,n}}(z)=\frac{1}{\Delta_n(t_m)}\begin{pmatrix}
    p_n(z;t_m) & -\mathcal{C}^{\B{m}}\left[p_n(\cdot;t_m)w(\cdot,t_m)\right](z)\\
    p_{n-1}(z;t_m) & 
    -\mathcal{C}^{\B{m}}\left[p_{n-1}(\cdot;t_m)w(\cdot,t_m)\right](z)
    \end{pmatrix},
 \end{equation*}
where $p_n(z;t_m)$, for $n\geq 0$, denotes the (generically) degree $n$ polynomial
 \begin{equation*}
 p_n(z;t_m)=\begin{vmatrix} 
 \mu_0 & \mu_1 & \ldots & \mu_{n}\\
 \mu_1 & \mu_2 &\ldots  & \mu_{n+1}\\
 \vdots & \vdots  &\ddots & \vdots\\
 \mu_{n-1} & \mu_n & \ldots & \mu_{2n-1}\\
 1 & z &\ldots & z^n
 \end{vmatrix}.
 \end{equation*}
The latter polynomials satisfy the orthogonality condition
 \begin{equation*}
  \frac{1}{2\pi i}\oint_{\gamma^{\B{m}}}p_l(z,t_m)p_n(z,t_m)w(z,t_m)dz=\Delta_n(t_m)\Delta_{n+1}(t_m)\delta_{l,n}\quad (l,n\in\mathbb{N}),
 \end{equation*}
 and thus form a sequence of orthogonal polynomials with respect to the complex functional
 \begin{equation}\label{eq:complex_functional}
     \mathbb{C}[z]\rightarrow \mathbb{C},p(z)\mapsto \frac{1}{2\pi i}\oint_{\gamma^{\B{m}}}p(z)w(z,t_m)dz,
 \end{equation}
 when none of the Hankel determinants vanish.
 
 We denote
 \begin{equation*}
    \mathfrak{M}_n:=\{m\in\mathbb{Z}:\Psi^{\B{m,n}}(z) \text{ exists}\},
\end{equation*}
and assume $c\neq 0$. We may employ a similar argument as in the proof of Theorem \ref{thm:main_solvability}, to show that, for any $m\in\mathbb{Z}$, $m\in \mathfrak{M}_n$ or $m+1\in \mathfrak{M}_n$. We thus obtain a corresponding solution $(w^{\B{n}},f^{\B{n}},g^{\B{n}})$ of $q\Psix^\text{aux}(\kappa^{\B{n}},t_0)$, where
\begin{equation*}
    \kappa^{\B{n}}=(\kappa_0,\kappa_t,\kappa_1,\kappa_{\infty,n}),\quad \kappa_{\infty,n}:=q^n\kappa_\infty, \quad \kappa_0=\kappa_t\kappa_1\kappa_\infty,
\end{equation*}
 for $n\geq 0$.
 
We proceed to derive explicit formulas for $f^{\B{n}}$ and $g^{\B{n}}$. To this end, we note that the next to highest order coefficient, in the asymptotic expansion
 \begin{equation*}
     Y^{\B{m,n}}(z)=\left(I+z^{-1}Y_1^{\B{m,n}}+\mathcal{O}\left(z^{-2}\right )\right)z^{n\sigma_3}\quad z\rightarrow \infty,
 \end{equation*}
 can be written explicitly as
 \begin{equation*}
     Y_1^{\B{m,n}}=\displaystyle\begin{pmatrix}
     -\frac{\Gamma_n(t_m)}{\Delta_n(t_m)} & \frac{\Delta_{n+1}(t_m)}{\Delta_n(t_m)}\\
     \frac{\Delta_{n-1}(t_m)}{\Delta_n(t_m)} & \frac{\Gamma_n(t_m)}{\Delta_n(t_m)}
     \end{pmatrix},
 \end{equation*}
where $\Delta_n(t_m)$ denotes the $n$-th Hankel determinant of moments and
\begin{equation*}
 \Gamma_n(t_m):=\begin{vmatrix} 
	\mu_{0} & \mu_{1} & \ldots & \mu_{n-2}& \mu_{n}\\
	\mu_{1} & \mu_{2} &\ldots  & \mu_{n-1}& \mu_{n+1}\\
	\vdots & \vdots  &\ddots & \vdots\\
	\mu_{n-2} & \mu_{n-1} & \ldots & \mu_{2n-4}& \mu_{2n-2}\\
	\mu_{n-1} & \mu_n & \ldots & \mu_{2n-3}& \mu_{2n-1}
\end{vmatrix},
 \end{equation*}
 with $\Gamma_1(t_m)=\mu_1$ and $\Gamma_0(t_m)=0$.
 
 By direct substitution of the corresponding asymptotic expansion of $\Psi^{\B{m,n}}(z)$ around $z=\infty$ into equation \eqref{eq:rhptolinear}, we find
 \begin{equation*}
w^{\B{n}}(t_m)=-q^{-1}(q\kappa_{\infty,n}^2-1)\left(\frac{\nu}{\kappa_t}\right)^m \frac{\Delta_{n+1}(t_m)}{\Delta_n(t_m)},
 \end{equation*}
 and
 \begin{equation}\label{eq:fexplicit}
 f^{\B{n}}(t)=\frac{\kappa_{\infty,n}^2-1}{q\kappa_{\infty,n}^2-1}\frac{\Gamma_n(t)}{\Delta_n(t)}-
 \frac{q^2\kappa_{\infty,n}^2-1}{q\kappa_{\infty,n}^2-1}\frac{\Gamma_{n+1}(t)}{\Delta_{n+1}(t)}+L(t),
 \end{equation}
 where $t=t_m$ and the linear term $L$ reads
 \begin{equation*}
     L(t)=\kappa_t t+\kappa_1+\frac{\kappa_t(\kappa_1^2-1)+\kappa_1(\kappa_t^2-1)t}{\kappa_t\kappa_1(q\kappa_{\infty,n}^2-1)}.
 \end{equation*}
 Upon substituting the explicit formula for $w^{\B{n}}$ into the auxiliary equation \eqref{eq:auxiliary}, and solving for $g$, we obtain
 \begin{equation}\label{eq:gexplicit}
 g^{\B{n}}(t)=\kappa_{\infty,n}\frac{\nu \Delta_n(t/q)\Delta_{n+1}(t)-\kappa_t\Delta_n(t)\Delta_{n+1}(t/q)\textcolor{white}{q\kappa_{\infty,n}^2}}{\nu \Delta_n(t/q)\Delta_{n+1}(t)-\kappa_t\Delta_n(t)\Delta_{n+1}(t/q)q\kappa_{\infty,n}^2}.
 \end{equation}
 Note that, by the above formulas, $\Delta_n(t)=0$ if and only if $f^{\B{n}}(t)=\infty$ and $g^{\B{n}}(t)=\kappa_{\infty,n}$, consistent with Theorem \ref{thm:main_solvability}.

 Furthermore, the moments $\mu_k=\mu_k(t)$ can be expressed explicitly in terms of Heine's basic hypergeometric functions. Indeed, a residue computation  yields that the $k$-th moment equals
 	\begin{align*}
 	\mu_k=&S_1+S_2,\\
	S_1=&\frac{c\,\kappa_0^2\, \theta_q(q\kappa_t\nu)}{(q;q)_\infty(q/\kappa_t^2;q)_\infty}\frac{\left(q^{1+k}\frac{q\kappa_0^2}{\kappa_t^2};q\right)_\infty}{\left(q^{1+k}\kappa_0^2;q\right)_\infty}\left(\frac{qt}{\kappa_t}\right)^{k+1}\frac{\theta_q\left(\frac{\kappa_1\nu t}{\kappa_0^2}\right)}{\theta_q\left(\frac{\kappa_1 t}{\kappa_t}\right)}\\
	&\times
	\;_{2}\phi_1\left[\begin{matrix} 
		\kappa_1^2, q^{1+k}\kappa_0^2 \\ 
		q^{2+k}\frac{\kappa_0^2}{\kappa_t^2} \end{matrix} 
	; q,\frac{qt}{\kappa_t\kappa_1} \right],\\
	S_2=&\frac{c\,\kappa_0^2\, \theta_q\left(\frac{\kappa_t\nu}{\kappa_0^2}\right)}{\nu \kappa_t(q;q)_\infty(q/\kappa_1^2;q)_\infty}\frac{\left(q^{1+k}\frac{q\kappa_0^2}{\kappa_1^2};q\right)_\infty}{\left(q^{1+k}\kappa_0^2;q\right)_\infty}\left(\frac{q}{\kappa_1}\right)^{k+1}\frac{\theta_q\left(\kappa_1\nu t\right)}{\theta_q\left(\frac{\kappa_1 t}{\kappa_t}\right)}\\
	&\times\;_{2}\phi_1 \left[\begin{matrix} 
		\kappa_t^2, q^{1+k}\kappa_0^2 \\ 
		q^{2+k}\frac{\kappa_0^2}{\kappa_1^2} \end{matrix} 
	; q,\frac{q}{\kappa_t\kappa_1t} \right].
\end{align*}	

Sakai \cite{sakaicasorati} first derived special function solutions of $q\Psix$, written in terms of Casorati determinants of Heine's basic hypergeometric functions, which correspond to setting $\nu=\kappa_t^{-1}$ or $\nu=\kappa_0^2/\kappa_t$ in the above, so that $S_1=0$ or $S_2=0$ respectively.

Ormerod et al. \cite{forresterqorthogonal} related a family of semi-classical orthogonal polynomials to $q\Psix$, via the Jimbo-Sakai linear system, and derived formulas similar to \eqref{eq:fexplicit} and \eqref{eq:gexplicit} above.
To relate the orthogonal polynomials in this section to those in \cite{forresterqorthogonal}, we write the complex functional \eqref{eq:complex_functional} in terms of $q$-Jackson integrals. Assuming that $|\kappa_0|<|q|^{-\frac{1}{2}}$, a residue computation gives
 \begin{align*}
\frac{1}{2\pi i}\oint_{\gamma^{\B{m}}} p(z)w(z,t_m)dz=&\alpha_1(t_m;c,\nu) \int_0^{q \kappa_t^{-1}t_m}p(z)W(z,t_m)d_qz\\
&+\alpha_2(t_m;c,\nu) \int_0^{q \kappa_1^{-1}}p(z)W(z,t_m)d_qz,
 \end{align*}
for any entire function $p(z)$, where the right-hand side integrals are standard Jackson integrals, $W(z,t)$ is the weight function
 \begin{equation*}
     W(z,t):=z^\sigma \frac{\left(\frac{\kappa_t z}{ t},\kappa_1 z;q\right)_\infty}{\left(\frac{z}{\kappa_t  t},\frac{z}{\kappa_1};q\right)_\infty},\quad \sigma:=2\log_q\kappa_0,
 \end{equation*}
and the dependence of the integral operator on the monodromy data $\{c,\nu\}$ is hidden in the coefficients in front of the Jackson integrals,
\begin{align*}
\alpha_1(t;c,\nu)&=\frac{c\,(t/\kappa_t)^{-\sigma}}{(1-q)(q;q)_\infty^2}\frac{\theta_q\left(q\kappa_t\nu,\frac{\kappa_1 \nu  t}{\kappa_0^2}\right)}{\theta_q\left(\frac{\kappa_1 t}{\kappa_t}\right)},\\
\alpha_2(t;c,\nu)&=\frac{c\,\kappa_1^{\sigma}}{(1-q)(q;q)_\infty^2}\frac{\theta_q\left(q\kappa_1 \nu t,\frac{\kappa_t \nu }{\kappa_0^2}\right)}{\theta_q\left(\frac{\kappa_t}{\kappa_1 t}\right)}.
\end{align*}
Note that both coefficients satisfy $\alpha(qt)=\frac{1}{\kappa_t \nu}\alpha(t)$ and the orthogonal polynomials in Ormerod et al. \cite{forresterqorthogonal} then coincide with the polynomials $p_n$ above, up to scalar multiplication, in the case when $\nu$ is chosen such that $\alpha_1(t)=-\alpha_2(t)$. In other words, $\nu=\nu(t_0)$ is chosen such that
\begin{equation*}
    \left(\frac{\kappa_1 t_0}{\kappa_t}\right)^{1+\sigma}=\frac{\theta_q\left(q\kappa_t\nu,\frac{ \kappa_1 \nu t_0}{\kappa_0^2}\right)}{\theta_q\left(q\nu\kappa_1  t_0,\frac{\kappa_t \nu}{\kappa_0^2}\right)}.
\end{equation*}

Next, we briefly consider an example coming from one of the conditions in \eqref{eq:intro_irreducibleparameter2} being violated. Namely, we set
\begin{equation*}
    \kappa_0=\kappa_\infty t_0,
\end{equation*}
and consider RHP I, defined in Definition \ref{def:RHPmain}, with a corresponding upper-triangular connection matrix of the form,
\begin{equation*}
    C(z)=\begin{pmatrix}
    \theta_q\left(\frac{z}{\kappa_t t_0},\frac{z\kappa_t}{t_0}\right) & c\,\theta_q(\frac{z}{\kappa_0}\nu,\frac{z}{\kappa_0}\nu^{-1})\\
    0 & \theta_q\left(\frac{z}{\kappa_1},z\kappa_1\right)
    \end{pmatrix},
\end{equation*}
where the monodromy datums $c\in\mathbb{C}$ and $\nu\in\mathbb{C}^*$ can again be chosen at pleasure.

We note that the jump matrix $z^mC(z)$ can be rewritten as
\begin{equation*}
    z^mC(z)=
    \begin{pmatrix}
    q^{m(m+1)}t_0^{2m}\theta_q\left(\frac{z}{\kappa_t t_m},\frac{z\kappa_t}{t_m}\right) & c\,\theta_q(\frac{z}{\kappa_0}\nu,\frac{z}{\kappa_0}\nu^{-1})\\
    0 & \theta_q\left(\frac{z}{\kappa_1},z\kappa_1\right)
    \end{pmatrix}z^{-m\sigma_3},
\end{equation*}
where we denoted $t_m:=q^m t_0$.

We apply the transformation
\begin{equation*}
  Y^{\B{m}}(z)=\begin{cases}
  D_1^{-1}\Psi^{\B{m}}(z)F_\infty^{\B{m}}(z)^{-1}D_1 z^{m\sigma_3}  & \text{ if $z\in D_+^{\B{m}}$,}\\
   D_1^{-1}\Psi^{\B{m}}(z)F_0^{\B{m}}(z)^{-1} & \text{ if $z\in D_-^{\B{m}}$,}\\
  \end{cases}
\end{equation*}
where
\begin{align*}
F_\infty^{\B{m}}(z)&=
\begin{pmatrix}
\left(\frac{q\kappa_tt_m}{z},\frac{qt_m}{\kappa_tz};q\right)_\infty & 0\\
0 & \left(\frac{q\kappa_1}{z},\frac{q}{\kappa_1z};q\right)_\infty,\\
\end{pmatrix}\\
F_0^{\B{m}}(z)&=\begin{pmatrix}
\left(\frac{z}{\kappa_t t_m},\frac{\kappa_t z}{ t_m};q\right)_\infty & 0\\
0 & \left(\frac{z}{\kappa_1},z\kappa_1;q\right)_\infty\\
\end{pmatrix},
\end{align*}
and
\begin{equation*}
    D_1=\begin{pmatrix}
    t_0^{-2m}q^{-m(m+1)} & 0\\
    0 & 1\\
    \end{pmatrix}.
\end{equation*}
Then the jump matrix for $Y^{\B{m}}(z)$ reads
\begin{equation*}
    J^{\B{m}}(z)=\begin{pmatrix}
    1 & \widehat{w}(z,t_m)\\
    0 & 1
    \end{pmatrix},
\end{equation*}
where
\begin{equation}\label{eq:weight_function2}
     \widehat{w}(z,t)=\frac{c\,\theta_q(\frac{z}{\kappa_0}\nu,\frac{z}{\kappa_0}\nu^{-1})}{\left(\frac{z}{\kappa_t t},\frac{\kappa_t z}{ t};q\right)_\infty\left(\frac{q\kappa_1}{z},\frac{q}{\kappa_1z};q\right)_\infty},
\end{equation}
and $Y^{\B{m}}(z)$ solves the following RHP, if it exists.
\begin{definition}[RHP IV]\label{def:RHPorthogonal2}
For $m\in\mathbb{Z}$, find a matrix function $Y^{\B{m}}(z)$ which satisfies the following conditions.
  \begin{enumerate}[label={{\rm (\roman *)}}]
  \item $Y^{\B{m}}(z)$ is analytic on $\mathbb{C}\setminus\gamma^{\B{m}}$.
    \item $Y^{\B{m}}(z')$ has continuous boundary values $Y_-^{\B{m}}(z)$ and $Y_+^{\B{m}}(z)$ as $z'$ approaches $z\in \gamma^{\B{m}}$ from $D_-^{\B{m}}$ and $D_+^{\B{m}}$ respectively, related by
		\begin{equation*}
		Y_+^{\B{m}}(z)=Y_-^{\B{m}}(z)\begin{pmatrix}
    1 & \widehat{w}(z,t_m)\\
    0 & 1
    \end{pmatrix}\qquad (z\in \gamma^{\B{m}}),
              \end{equation*}
    where $\widehat{w}(z,t)$ is the weight function defined in equation \eqref{eq:weight_function2}.      
            \item $Y^{\B{m}}(z)$ satisfies
              \begin{equation*}
		Y^{\B{m}}(z)=\left(I+\mathcal{O}\left(z^{-1}\right )\right)z^{m\sigma_3}\quad z\rightarrow \infty.
              \end{equation*}
              \end{enumerate}
  \end{definition}
 This RHP takes the form of the Fokas-Its-Kitaev RHP for orthogonal polynomials, but with the contour $\gamma^{\B{m}}$ and weight function $\widehat{w}(z,t_m)$ scaling with the `degree' $m$ of the corresponding orthogonal polynomials.
 In particular, RHP IV is unsolvable for $m<0$ and thus so is RHP I in Definition \ref{def:RHPmain}.
 
 For $m=0$, RHP IV is solvable and its solution is given by
  \begin{equation*}
    \Psi^{\B{0}}(z)=\begin{pmatrix}
    1 & -\mathcal{C}^{\B{0}}\left[w(\cdot,t_0)\right](z)\\
    0 & 1
    \end{pmatrix}.
 \end{equation*}
 From equation \eqref{eq:rhptolinear} it follows that the corresponding linear system $A(z,t)$ at $t=t_0$ takes the upper-triangular form
 \begin{equation}\label{eq:Aexplicit}
    A(z,t_0)=\begin{pmatrix}
    \kappa_\infty(z-\kappa_t t_0)(z-\kappa_t^{-1} t_0) & A_{12}(z,t_0)\\
    0 & \kappa_\infty^{-1}(z-\kappa_1)(z-\kappa_1^{-1})
    \end{pmatrix}.
\end{equation}

For $m\geq 0$, RHP IV is solvable if and only if the $m$th Hankel determinant of moments for the weight function $\widehat{w}(z,t_m)$, with respect to the contour $\gamma^{\B{m}}$, is nonzero.
We denote
\begin{equation*}
    \mathfrak{M}:=\{m\in\mathbb{Z}:\Psi^{\B{m}}(z) \text{ exists}\},
\end{equation*}
then $\mathfrak{M}\subseteq \mathbb{N}$ and, if $c\neq 0$, then, by the same argument as in the proof of Theorem \ref{thm:main_solvability}, we may show that, for any $m\geq 0$, $m\in \mathfrak{M}$ or $m+1\in \mathfrak{M}$.
Thus the domain of the corresponding solution $(f,g)$ is given by the semi $q$-spiral $q^{\mathbb{N}}t_0$.

Note that, by equation \eqref{eq:Aexplicit}, the value of $g$ at $t=t_0$ is given by
\begin{equation*}
    g(t_0)=q^{-1} \kappa_\infty^{-1}=q^{-1}\kappa_0^{-1}t_0,
\end{equation*}
and thus equation \eqref{eq:qpvi} has a singularity at $t=q^{-1} t_0$ which cannot be resolved. In particular, there exists no isomonodromic continuation of the solution past $t=t_0$, see also \cite{dreyfus2020degeneration}[Prop. 4.1].
 
This phenomenon has also been observed for solutions of other discete Painlev\'e equations associated with orthogonal polynomials, see e.g. Assche \cite{asschebook}.

We emphasise that, also in this case, one can derive explicit expressions for $f(q^mt_0)$ and $g(q^mt_0)$, $m\geq 0$, in terms of determinants of moments, but with the sizes of the determinants growing with $m$.

Finally, note that, if we set $c=0$, so that $C(z)$ is diagonal, we have $\widehat{w}(z,t)=0$ and $A_{12}(z,t_0)\equiv 0$. In this singular case, $\mathfrak{M}=\{0\}$ and there is no solution $(f,g)$ of $q\Psix(\kappa,t_0)$ corresponding to this monodromy.

We finish this section by noting that, in general, the domain where RHP I, defined in Definition \ref{def:RHPmain}, is solvable,
\begin{equation*}
    \mathfrak{M}:=\{m\in\mathbb{Z}:\Psi^{\B{m}}(z) \text{ exists}\},
\end{equation*}
can take one of five particular forms when $C(z)$ is reducible, characterised by
\begin{enumerate}
    \item $\forall m\in\mathbb{Z}$, $m\in \mathfrak{M}$ or $m+1\in \mathfrak{M}$;
    \item $\exists{m_0\in\mathfrak{M}}$ such that $\mathfrak{M}\subseteq \mathbb{Z}_{\geq m_0}$ and $\forall m\geq m_0$:  $m\in \mathfrak{M}$ or $m+1\in \mathfrak{M}$;
    \item $\exists{m_0\in\mathfrak{M}}$ such that $\mathfrak{M}\subseteq \mathbb{Z}_{\leq m_0}$ and $\forall m\leq m_0$:  $m\in \mathfrak{M}$ or $m-1\in \mathfrak{M}$;
    \item $\exists{m_0\in\mathfrak{M}}$ such that $\mathfrak{M}=\{m_0\}$;
    \item $\mathfrak{M}=\emptyset$.
\end{enumerate}
In the first example of this section, we saw cases (1) and (5). In the second example, we saw cases (2) and (4) with $m_0=0$.

\section{The Monodromy Manifold} \label{sec:monodromy_surface}
This section is devoted to the monodromy manifold defined in Definition \ref{def:monodromy_mapping}. In Sections \ref{subsec:moduli}, \ref{subsection:smoothness} and \ref{subsection:embeddings} we prove Theorems \ref{thm:main_moduli}, \ref{thm:main_smooth} and \ref{thm:main_affine} respectively.

\subsection{On the embedding of the monodromy manifold} \label{subsec:moduli}
In Section \ref{sec:results_monodromy_surface}, see equation \eqref{eq:coordinate_mapping}, we defined a mapping $\mathcal{P}$ of the monodromy manifold to $(\mathbb{P}^1)^4/\mathbb{C}^*$. In this section we show that this mapping is an embedding and determine its image, proving Theorem \ref{thm:main_moduli}.

Firstly, we have the following lemma.
\begin{lemma}\label{lem:embedding}
The mappings $\mathcal{P}$, defined in equation \eqref{eq:coordinate_mapping}, is injective.
\end{lemma}
\begin{proof}
Take any two connection matrices $C(z),\widetilde{C}(z)\in\mathfrak{C}(\kappa,t_0)$ and suppose that their respective coordinate values $\rho$ and $\widetilde{\rho}$ are identical up to scaling, i.e. $\widetilde{\rho}=c\rho$, for some $c\in\mathbb{C}^*$. Then the matrix function
\begin{equation*}
    D(z)=\widetilde{C}(z)\begin{pmatrix} 1 & 0\\
    0 & c\\
    \end{pmatrix} C(z)^{-1},
\end{equation*}
is analytic on $\mathbb{C}^*$. But $D(z)$ satisfies
\begin{equation*}
    D(qz)=\kappa_0^{\sigma_3}D(z)\kappa_0^{-\sigma_3},
\end{equation*}
and, as $\kappa_0^2\notin q^\mathbb{Z}$, it follows from the general theory of $q$-theta functions, see e.g. Lemma \ref{lem:classificationtheta}, that $D(z)\equiv D$ must be a constant diagonal matrix. Therefore, $[\widetilde{C}(z)]$ and $[C(z)]$ represent the same point on the monodromy manifold. The thesis follows.
\end{proof}

To determine the image of the monodromy manifold under $\mathcal{P}$, it is convenient to consider a related embedding into $(\mathbb{P}^1)^4$ of a finer quotient of the space $\mathfrak{C}(\kappa,t_0)$, given in the following definition.
\begin{definition}\label{def:monodromymanifoldM}
We define $M(\kappa,t_0)$ to be the space of connection matrices $\mathfrak{C}(\kappa,t_0)$ quotiented by arbitrary left-multiplication by invertible diagonal matrices. We denote the equivalence class of $C(z)\in \mathfrak{C}(\kappa,t_0)$ in $M(\kappa,t_0)$ by $\llbracket C(z)\rrbracket$
and denote by
\begin{equation*}
\iota_M:   M(\kappa,t_0)\rightarrow   \mathcal{M}(\kappa,t_0),\llbracket C(z)\rrbracket \rightarrow [C(z)],
\end{equation*}
the quotient mapping of $M(\kappa,t_0)$ onto the monodromy manifold.
\end{definition}

Note that the coordinates $\rho=(\rho_1,\rho_2,\rho_3,\rho_4)$ introduced in Section \ref{sec:results_monodromy_surface}, i.e.
\begin{equation*}
\rho_k=\pi(C(x_k)),\quad (1\leq k\leq 4),\quad (x_1,x_2,x_3,x_4)=(\kappa_t t_0,\kappa_t^{-1} t_0,\kappa_1 ,\kappa_1^{-1}),
\end{equation*}
are invariant under left-multiplication by diagonal matrices and are thus well defined on equivalence classes in $M(\kappa,t_0)$. We thus obtain a mapping
\begin{equation}\label{eq:coordinate_mappingfine}
    P:M(\kappa,t_0)\rightarrow (\mathbb{P}^1)^4,\llbracket C(z)\rrbracket \mapsto \rho.
\end{equation}
This mapping is an embedding, by the same argument as given in the proof of Lemma \ref{lem:embedding}, with $c$ set equal to $1$.

Let $\iota_\mathbb{P}$ denote the quotient mapping
\begin{equation}\label{eq:defi_iota}
    \iota_\mathbb{P}:(\mathbb{P}^1)^4\rightarrow (\mathbb{P}^1)^4/\mathbb{C}^*.
\end{equation}
The proof of Theorem \ref{thm:main_moduli} revolves around the diagram
\begin{equation}\label{eq:commutative_diagram}
\begin{tikzcd}
M(\kappa,t_0) \arrow[r, "P"] \arrow[d, "\iota_M"]
& (\mathbb{P}^1)^4 \arrow[d, "\iota_\mathbb{P}"] \\
\mathcal{M}(\kappa,t_0) \arrow[r, "\mathcal{P}"]
& (\mathbb{P}^1)^4/\mathbb{C}^*,
\end{tikzcd}
\end{equation}
which is commutative, because right multiplication by a diagonal matrix translates to scalar multiplication of $\rho$ as shown in equation \eqref{eq:scaling}. We first determine the image of $M(\kappa,t_0)$ under $P$, following the technique developed in our previous paper \cite{joshiroffelseniv}, and then obtain Theorem \ref{thm:main_moduli} by projecting this image into $(\mathbb{P}^1)^4/\mathbb{C}^*$ via $\iota_\mathbb{P}$.

To describe the image of $M(\kappa,t_0)$ under $P$, we make the following definition.
\begin{definition}\label{def:modulispacefine}
Recall the definition of the quadratic polynomial $T(\rho:\kappa,t_0)$ as well as its homogeneous form $T_{hom}$ in Definition \ref{def:modulispace}. Using homogeneous coordinates $\rho_k=[\rho_k^x: \rho_k^y]\in \mathbb P^1$, $1\le k\le 4$, the equation
  \begin{equation*}
      T_{hom}(\rho_1^x,\rho_1^y,\rho_2^x,\rho_2^y,\rho_3^x,\rho_3^y,\rho_4^x,\rho_4^y)=0
  \end{equation*}
  defines a threefold in $(\mathbb{P}^1)^4$, which we denote by
  \begin{equation*}
  S(\kappa,t_0)=\{\rho\in (\mathbb{P}^1)^4:T(\rho:\kappa,t_0)=0\} .
  \end{equation*}
\end{definition}

Regarding the image of $M(\kappa,t_0)$ under $P$, we have the following result.
\begin{proposition}\label{prop:moduli}
Denote by $\widehat{\kappa}$ the tuple of complex parameters $\kappa$ after replacing $\kappa_0\mapsto 1$.
 The image of  $M(\kappa,t_0)$ under the mapping $P$, defined in equation \eqref{eq:coordinate_mappingfine}, is given by the threefold
$S(\kappa,t_0)$ minus the codimension one subspace
\begin{equation}
X(\kappa,t_0):=S(\kappa,t_0)\cap S(\widehat{\kappa},t_0)
=\bigcap_{\lambda_0\in\mathbb{C}^*}S(\lambda_0,\kappa_t,\kappa_1,\kappa_\infty,t_0).\label{eq:Xdefi}
\end{equation}
We denote by $S^*(\kappa,t_0)$ the space obtained by cutting this subspace from $S(\kappa,t_0)$,
then the mapping
\begin{equation*}
 M(\kappa,t_0)\rightarrow S^*(\kappa,t_0),\ \textrm{where}\ \llbracket C(z)\rrbracket \mapsto P(\llbracket C(z)\rrbracket),
\end{equation*}
is a bijection.
\end{proposition}
\begin{proof}
Let us take a connection matrix $C(z)\in\mathfrak{C}(\kappa,t_0)$.
It will be convenient to work with the following uniform notation,
\begin{subequations}\label{eq:notationsigmamu}
\begin{align}
&(\sigma_1,\sigma_2)=(\kappa_0 t_0,\kappa_0^{-1} t_0),\quad (\mu_1,\mu_2)=(\kappa_\infty,\kappa_\infty^{-1}), \\
&(x_1,x_2,x_3,x_4)=(\kappa_t t_0,\kappa_t^{-1} t_0,\kappa_1 ,\kappa_1^{-1}).
\end{align}
\end{subequations}
For any $1\leq i,j\leq 2$, the matrix-entry $C_{ij}(z)$ is an element of the two-dimensional vector space
\begin{equation*}
    V_{ij}:=\left\{\text{analytic functions } \theta:\mathbb{C}^*\rightarrow \mathbb{C}\text{ satisfying }\theta(qz)=\frac{\sigma_i}{\mu_j}z^{-2}\theta(z)\right\},
\end{equation*}
see equation \eqref{eq:theta},
and we know that
\begin{equation}\label{eq:det0}
    C_{11}(z)C_{22}(z)-C_{12}(z)C_{22}(z)=c\theta_q(z/x_1,z/x_2,z/x_3,z/x_4),
\end{equation}
for some $c\in\mathbb{C}^*$.

For each $1\leq k\leq 4$, the equation $\pi(C(x_k))=\rho_k$ translates to
\begin{equation}\label{eq:rho_relations}
\rho_k^y C_{11}(x_k)-\rho_k^x C_{12}(x_k)=0,\quad \rho_k^y C_{21}(x_k)-\rho_k^x C_{22}(x_k)=0,      
\end{equation}
where we used homogeneous coordinates $\rho_k=[\rho_k^x:\rho_k^y]$.

We proceed in studying equations \eqref{eq:rho_relations} by choosing explicit bases of the vector spaces $V_{ij}$, $1\leq i,j\leq 2$. To this end, we introduce the following eight $q$-theta functions,
\begin{align*}
    u_1^{11}(z)&=\theta_q\left(z/x_1,z x_1\frac{\mu_1}{\sigma_1}\right), & u_2^{11}(z)&=\theta_q\left(z/x_2,z x_2\frac{\mu_1}{\sigma_1}\right),\\
    u_1^{12}(z)&=\theta_q\left(z/x_3,z x_3\frac{\mu_2}{\sigma_1}\right), & u_2^{12}(z)&=\theta_q\left(z/x_4,z x_4\frac{\mu_2}{\sigma_1}\right),\\
    u_1^{21}(z)&=\theta_q\left(z/x_1,z x_1\frac{\mu_1}{\sigma_2}\right), & u_2^{21}(z)&=\theta_q\left(z/x_2,z x_2\frac{\mu_1}{\sigma_2}\right),\\
    u_1^{22}(z)&=\theta_q\left(z/x_3,z x_3\frac{\mu_2}{\sigma_2}\right), & u_2^{22}(z)&=\theta_q\left(z/x_4,z x_4\frac{\mu_2}{\sigma_2}\right).
\end{align*}
For any $1\leq i,j\leq 2$, the collection $\{u_1^{ij}(z),u_2^{ij}(z)\}$ forms a basis of $V_{ij}$. We may thus write
\begin{equation}\label{eq:connectionentries}
    C_{ij}(z)=\alpha_{1}^{ij}u_1^{ij}(z)+\alpha_{2}^{ij}u_2^{ij}(z),
\end{equation}
for some coefficients $\alpha_{1}^{ij},\alpha_{2}^{ij}\in\mathbb{C}$.

Equations \eqref{eq:rho_relations} now translate to eight equations among the coefficients in \eqref{eq:connectionentries}, which we group into the following two homogeneous systems,
\begin{equation}
    \begin{pmatrix}
    0  &   \rho_1^y u_2^{11}(x_1) & -\rho_1^x u_1^{12}(x_1) & -\rho_1^x u_2^{12}(x_1)\\
    \rho_2^y u_{1}^{11}(x_2)   &   0 & -\rho_2^x u_1^{12}(x_2) & -\rho_2^x u_2^{12}(x_2)\\
    \rho_3^y u_{1}^{11}(x_3)   &   \rho_3^y u_2^{11}(x_3) & 0 & -\rho_3^x u_2^{12}(x_3)\\
    \rho_4^y u_{1}^{11}(x_4)   &   \rho_4^y u_2^{11}(x_4) & -\rho_4^x u_1^{12}(x_4) & 0\\
    \end{pmatrix}
    \begin{pmatrix}
    \alpha_1^{11}\\
    \alpha_2^{11}\\
    \alpha_1^{12}\\
    \alpha_2^{12}\\
    \end{pmatrix}=
    \begin{pmatrix}
    0\\
    0\\
    0\\
    0\\
    \end{pmatrix}, \label{eq:rho_linear_1}
\end{equation}
and
\begin{equation}
        \begin{pmatrix}
    0   &   \rho_1^y u_2^{21}(x_1) & -\rho_1^x u_1^{22}(x_1) & -\rho_1^x u_2^{22}(x_1)\\
    \rho_2^y u_{1}^{21}(x_2)   &   0 & -\rho_2^x u_1^{22}(x_2) & -\rho_2^x u_2^{22}(x_2)\\
    \rho_3^y u_{1}^{21}(x_3)   &   \rho_3^y u_2^{21}(x_3) & 0 & -\rho_3^x u_2^{22}(x_3)\\
    \rho_4^y u_{1}^{21}(x_4)   &   \rho_4^y u_2^{21}(x_4) & -\rho_4^x u_1^{22}(x_4) & 0\\
    \end{pmatrix}
    \begin{pmatrix}
    \alpha_1^{21}\\
    \alpha_2^{21}\\
    \alpha_1^{22}\\
    \alpha_2^{22}\\
    \end{pmatrix}=
    \begin{pmatrix}
    0\\
    0\\
    0\\
    0\\
    \end{pmatrix}.\label{eq:rho_linear_2}
\end{equation}

As the determinant of $C(z)$ cannot be identically zero, we know that both vectors on the left-hand side of equations \eqref{eq:rho_linear_1} and \eqref{eq:rho_linear_2} are nonzero. This in turn implies that the determinants of the $4\times 4$ matrices on the left-hand side are zero. By means of a lengthy calculation, one can check that both determinants, coincide, up to some nonzero scalar multipliers, with the equation
\begin{equation*}
      T_{hom}(\rho_1^x,\rho_1^y,\rho_2^x,\rho_2^y,\rho_3^x,\rho_3^y,\rho_4^x,\rho_4^y)=0,
\end{equation*}
where $T_{hom}$ is defined in Definition \ref{def:modulispace}. We refer the interested reader to our previous work \cite{joshiroffelseniv}[Appendix B] where an analogous computation is given.

It follows that $P$ embeds $M(\kappa,t_0)$ into the threefold $S(\kappa,t_0)$.

We proceed to determine those coordinate-values in $S(\kappa,t_0)$ which cannot be realised by any connection matrix $C(z)\in\mathfrak{C}(\kappa,t_0)$.

Take any $\rho\in S(\kappa,t_0)$, then we know that both homogeneous equations \eqref{eq:rho_linear_1} and \eqref{eq:rho_linear_2} have non-trivial solutions. Let us take a solution of each respectively,
\begin{equation}\label{eq:connection_construction}
    \left(\alpha_1^{11},
    \alpha_2^{11},
    \alpha_1^{12},
    \alpha_2^{12}\right)^T,\quad \left(a_1^{21},
    \alpha_2^{21},
    \alpha_1^{22},
    \alpha_2^{22}\right)^T,
\end{equation}
and let $C(z)$ denote the corresponding matrix function via equations \eqref{eq:connectionentries}.

Then we know that $C(z)$ is analytic on $\mathbb{C}^*$, it satisfies
\begin{equation}\label{eq:qdifferenceC}
    C(qz)=z^{-2}t_0\kappa_0^{\sigma_3}C(z)\kappa_\infty^{-\sigma_3},
\end{equation}
and $|C(x_k)|=0$ for $1\leq k \leq 4$. Furthermore, by construction, 
\begin{subequations}\label{eq:entriesnonzero}
\begin{align}
   &C_{11}(z)\not\equiv 0\quad   \text{ or }\quad C_{12}(z)\not\equiv 0, \text{ and }\\
   &C_{21}(z)\not\equiv 0\quad  \text{ or }\quad C_{22}(z)\not\equiv 0.
\end{align}
\end{subequations}

There are two options, either equation \eqref{eq:det0} holds for some $c\in\mathbb{C}^*$, which means that $C(z)\in \mathfrak{C}(\kappa,t_0)$ and thus $\rho$ lies inside the range of $P$; or the determinant of $C(z)$ is identically zero,
\begin{equation}\label{eq:det_identically_zero}
    C_{11}(z)C_{22}(z)=C_{12}(z)C_{21}(z).
\end{equation}
In the latter case, $\rho$ does not lie inside the range of $P$. To show this, suppose on the contrary that there is a $\widetilde{C}(z)\in\mathfrak{C}(\kappa,t_0)$ with $\pi(\widetilde{C}(x_k))=\rho_k$ for $1\leq k\leq 4$. Then, by the same argument as in the proof of Lemma \ref{lem:embedding}, we have
\begin{equation}\label{eq:connection_equality}
    C(z)=D\widetilde{C}(z),
\end{equation}
for some diagonal matrix $D$. However, as the determinant of $C(z)$ is identically zero, we must have $|D|=0$. Consequently, equation \eqref{eq:connection_equality}  contradicts equations \eqref{eq:entriesnonzero}. It follows that, in the case when the determinant of $C(z)$ is identically zero, $\rho$ indeed does not lie in the range of $P$.

Therefore, to prove the proposition, it remains to be shown that the determinant of the matrix $C(z)$, constructed above, is identically equal to zero if and only if the coordinate-values $\rho$ lie in $X=X(\kappa,t_0)$, and that this space $X$ is a codimension one subspace of $S(\kappa,t_0)$.

To this end, let us note that equations \eqref{eq:entriesnonzero} and \eqref{eq:det_identically_zero} imply that either
\begin{enumerate}[label=(\roman*)]
\item $C_{11}(z)\equiv 0$ and $C_{21}(z)\equiv 0$,
\item $C_{12}(z)\equiv 0$ and $C_{22}(z)\equiv 0$, or
\item $C_{11}(z)C_{22}(z)=C_{12}(z)C_{21}(z)\not\equiv 0$.
\end{enumerate}
Case (i) corresponds, via equations \eqref{eq:rho_linear_1} and \eqref{eq:rho_linear_2}, to the four lines
\begin{equation}\label{eq:lineszero}
    \{\rho\in(\mathbb{P}^1)^4:\rho_i=\rho_j=\rho_k=0\}\quad (1\leq i<j<k\leq 4).
\end{equation}
Indeed, $C_{11}(z)\equiv 0$ implies that the coefficients $\alpha_{1}^{11},\alpha_2^{11}$ in equation \eqref{eq:rho_linear_1} are zero. A non-trivial solution of \eqref{eq:rho_linear_1} with these constraints exists if and only if the coordinate-values $\rho$ lies inside one of the above four lines.

Similarly, case (ii) corresponds to the four lines
\begin{equation}\label{eq:linesinfty}
    \{\rho\in(\mathbb{P}^1)^4:\rho_i=\rho_j=\rho_k=\infty\}\quad (1\leq i<j<k\leq 4).
\end{equation}

Note that the eight lines, defined in \eqref{eq:lineszero} and \eqref{eq:linesinfty}, indeed lie inside $X$.

Finally, in case (iii), $C(z)$ must take the form
\begin{equation*}
    C(z)=\begin{pmatrix}
    c_{11}\theta_q(z/u_1)\theta_q(z/v_1) & c_{12}\theta_q(z/u_1)\theta_q(z/u_2)\\
    c_{21}\theta_q(z/v_1)\theta_q(z/v_2) & c_{22}\theta_q(z/u_2)\theta_q(z/v_2)
    \end{pmatrix},
\end{equation*}
with
\begin{equation*}
u_1=\kappa_0 t_0 \tau,\quad u_2=\kappa_\infty \tau^{-1},\quad
v_1=\kappa_\infty^{-1}\tau^{-1},\quad
v_2=\kappa_0^{-1} t_0\tau,
\end{equation*}
for some $\tau\in\mathbb{C}^*$ and nonzero constant multipliers satisfying $c_{11}c_{22}=c_{12}c_{21}$.
 The corresponding $\rho$-coordinates of this matrix are given by
\begin{equation}\label{eq:rho_coord_singular}
    \rho_k=c\,\phi(\tau x_k)\quad (1\leq k\leq 4),\quad \phi(x):=\frac{\theta_q(x\kappa_\infty)}{\theta_q(x/\kappa_\infty)},\quad c=\frac{c_{11}}{c_{12}}\in\mathbb{C}^*.
\end{equation}
Consequently, for any choice of $c,\tau\in\mathbb{C}^*$, equation \eqref{eq:rho_coord_singular} defines a point on the threefold $S(\kappa,t_0)$. We now make the important observation that formulae \eqref{eq:rho_coord_singular} are $\kappa_0$-independent. That is, equation \eqref{eq:rho_coord_singular} defines a point on $S(\lambda_0,\kappa_t,\kappa_1,\kappa_\infty,t_0)$, for any value of $\lambda_0$. Thus these points lie in the subspace $X$.

To prove the proposition, it suffices to show that, conversely, any point in $X$ lies either on one of the eight lines \eqref{eq:lineszero} and \eqref{eq:linesinfty}, or is given by \eqref{eq:rho_coord_singular} for a choice of $c,\tau\in\mathbb{C}^*$. To this end, let us take a point $\rho\in X$ which is not on one of the eight lines.
 Construct a corresponding matrix $C(z)$ via equations \eqref{eq:rho_linear_1} and \eqref{eq:rho_linear_2}, see equation \eqref{eq:connection_construction}. So $C(z)$ is analytic on $\mathbb{C}^*$, it satisfies \eqref{eq:qdifferenceC} and equations \eqref{eq:entriesnonzero} hold true.

As $\rho\in X\subseteq S(1,\kappa_t,\kappa_1,\kappa_\infty,t_0)$, we can similarly construct a matrix $\widetilde{C}(z)$ via equations \eqref{eq:rho_linear_1} and \eqref{eq:rho_linear_2},
which satisfies
\begin{equation*}
    \widetilde{C}(qz)=z^{-2}t_0\widetilde{C}(z)\kappa_\infty^{-\sigma_3}.
\end{equation*}
This matrix function is also analytic on $\mathbb{C}^*$, and satisfies \eqref{eq:entriesnonzero}.

Now suppose, for the sake of obtaining a contradiction, that $\rho$ is not given by \eqref{eq:rho_coord_singular} for some $c,\tau\in\mathbb{C}^*$, so that $|C(z)|\not\equiv 0$. Consider the quotient $D(z)=\widetilde{C}(z)C(z)^{-1}$. As, by construction, $C(z)$ and $\widetilde{C}(z)$ have the same $\rho$-coordinate values, it follows that $D(z)$ is an analytic function on $\mathbb{C}^*$. However, $D(z)$ satisfies the $q$-difference equation
\begin{equation*}
    D(qz)=D(z)\kappa_0^{-\sigma_3},
\end{equation*}
and therefore, by Lemma \ref{lem:classificationtheta}, $D(z)\equiv 0$ and consequently $\widetilde{C}(z)\equiv 0$, which contradicts the fact that $\widetilde{C}(z)$ satisfies \eqref{eq:entriesnonzero}. 

We conclude that the subspace $X$ is explicitly parametrised by
\begin{equation}\label{eq:parametrisationX}
    X=\operatorname{cl}\left(\{(c\,\phi(\tau x_1),c\,\phi(\tau x_2),c\,\phi(\tau x_3),c\,\phi(\tau x_4)):c,\tau\in\mathbb{C}^*\}\right),
\end{equation}
where $\phi$ is the function defined in \eqref{eq:rho_coord_singular} and
 the closure is taken in $(\mathbb{P}^1)^4$. Thus $X$ is a codimension one closed subspace of $S(\kappa,t_0)$. Furthermore, we have shown that $X$ consists precisely of the points in the threefold $S(\kappa,t_0)$ that cannot be realised as coordinate-values $\rho$ of any connection matrix $C(z)\in\mathfrak{C}(\kappa,t_0)$. Thus the image of $M(\kappa,t_0)$ under the embedding $P$ is given by $S(\kappa,t_0)\setminus X$ and the proposition follows.
\end{proof}

\begin{proof}[Proof of Theorem \ref{thm:main_moduli}] Recall from Definition \ref{def:monodromymanifoldM} that elements of $M(\kappa,t_0)$ are connection matrices $C$ equivalent under left multiplication by a diagonal matrix, while the entries of $\mathcal{M}(\kappa,t_0)$ are those equivalent under right and left multiplication by diagonal matrices. Note that the desired bijection is already proved for $M(\kappa,t_0)$ in  Proposition \ref{prop:moduli}. So the proof of the present theorem will follow under an appropriate quotient mapping $M(\kappa,t_0)$ to $\mathcal{M}(\kappa,t_0)$ and the corresponding quotient from $S^*(\kappa,t_0)$ to $\mathcal{S}^*(\kappa,t_0)$. Recall that  Definition \ref{def:monodromymanifoldM} denotes the former quotient by $\iota_M$. The latter quotient is denoted by $\iota_{\mathbb{P}}$, defined in equation \eqref{eq:defi_iota}.

Now, consider the commutative diagram \eqref{eq:commutative_diagram}.
By Proposition \ref{prop:moduli}, the image of $P$ is given by $S^*(\kappa,t_0)$. Therefore, the image of the composition $\iota_{\mathbb{P}}\circ P$ is given by $\mathcal{S}^*(\kappa,t_0)$. As $\iota_M$ is surjective, it follows from the commutativity of diagram \eqref{eq:commutative_diagram} that the image of $\mathcal{P}$ is given by $\mathcal{S}^*(\kappa,t_0)$.

In Lemma \ref{lem:embedding}, it was shown that $\mathcal{P}$ is injective and it thus follows that $\mathcal{P}$ is a bijection, which proves the theorem.
\end{proof}

\begin{proof}[Proof of Remark \ref{remark:curveX}]
We note that, by equation \eqref{eq:parametrisationX}, we have the following explicit parametrisation of the curve $\mathcal{X}=\mathcal{X}(\kappa,t_0)$,
\begin{equation}\label{eq:parametrisationXcurve}
\begin{aligned}
    \mathcal{X}&=\operatorname{cl}\left(\{(\phi(\tau x_1),\phi(\tau x_2),\phi(\tau x_3),\phi(\tau x_4)):\tau\in\mathbb{C}^*\}\right),\\
    \phi(x)&=\frac{\theta_q(x\kappa_\infty)}{\theta_q(x/\kappa_\infty)},
\end{aligned}
\end{equation}
where the closure is taken in $(\mathbb{P}^1)^4/\mathbb{C}^*$ and $x_k,1\leq k\leq 4$, are as defined in equation \eqref{eq:intro_xnotation}. Note that this parametrisation is $\kappa_0$-independent, which implies
\begin{equation*}
\mathcal{X}\subseteq
\bigcap_{\lambda_0\in\mathbb{C}^*}\mathcal{S}(\lambda_0,\kappa_t,\kappa_1,\kappa_\infty,t_0)
.\end{equation*}
By the definition of $\mathcal{X}$, equation \eqref{eq:subvariety}, the right-hand side is also a subset of $\mathcal{X}$ and they are therefore equal, yielding the desired result, equation \eqref{eq:subvariety_infty}.
\end{proof}

\subsection{Smoothness of the monodromy manifold}\label{subsection:smoothness}
In this subsection, we study the smoothness of the monodromy manifold $\mathcal{M}(\kappa,t_0)$ and prove Theorem \ref{thm:main_smooth}.

The monodromy manifold does not naturally come with a topology. However, 
due to Theorem \ref{thm:main_moduli} and Proposition \ref{prop:moduli}, we have the following refined version of a commutative diagram \eqref{eq:commutative_diagram},
\begin{equation}\label{eq:commutative_diagram_2}
\begin{tikzcd}
M(\kappa,t_0) \arrow[r, "P"] \arrow[d, "\iota_M"]
& S^*(\kappa,t_0) \arrow[d, "\iota_{S^*}"] \\
\mathcal{M}(\kappa,t_0) \arrow[r, "\mathcal{P}"]
& \mathcal{S}^*(\kappa,t_0),
\end{tikzcd}
\end{equation}
where both $P$ and $\mathcal{P}$ are bijective, and $\iota_{S^*}$ denotes the quotient mapping $\iota_{\mathbb{P}}$ restricted to $S^*(\kappa,t_0)$.
The monodromy manifold inherits a topology from $\mathcal{S}^*(\kappa,t_0)$ via $\mathcal{P}$.
Similarly, $M(\kappa_0,t_0)$ inherits a topology from the threefold $S^*(\kappa,t_0)$.

To prove Theorem \ref{thm:main_smooth}, we first study the smoothness of the space $S^*(\kappa,t_0)$. We then deduce corresponding results for the surface $\mathcal{S}^*(\kappa,t_0)$, by taking the quotient with respect to scalar multiplication. Finally, we translate the results for $\mathcal{S}^*(\kappa,t_0)$ to the monodromy manifold.

The following proposition describes the singular set of the space $S^*(\kappa,t_0)$ and shows that it is empty if and only if the non-splitting conditions hold.
\begin{proposition}\label{prop:smooth}
The space $S^*(\kappa,t_0)$ is a complex $3$-manifold  singularities at points in the finite set
\begin{equation}\label{eq:defiSsing}
S^*_{sing}:=S^*(\kappa,t_0)\cap \Theta,
\end{equation}
where
\begin{align}\label{eq:theta_set}
\Theta:=\{&(0,0,\infty,\infty),(0,\infty,0,\infty),(0,\infty,\infty,0),\\
&(\infty,0,0,\infty),(\infty,0,\infty,0),(\infty,\infty,0,0)\}.
\end{align}
Furthermore, all these singularities are ordinary double-point singularities.

In particular, the following statements are equivalent.
\begin{enumerate}[label={{\rm (\roman *)}}]
    \item The space $S^*(\kappa,t_0)$ is smooth.
    \item The set $S^*_{sing}$ is empty.
    \item The non-splitting conditions \eqref{eq:intro_irreducibleparameter} hold true.
\end{enumerate}
\end{proposition}
\begin{proof}
Recall that the space $S^*(\kappa,t_0)$ is defined as $S(\kappa,t_0)\setminus X(\kappa,t_0)$, where $S(\kappa,t_0)$ is the zero locus of the polynomial $T(\rho;\kappa,t_0)$ in $(\mathbb{P}^1)^4$ and $X(\kappa,t_0)$ denotes a subspace of $S(\kappa,t_0)$, defined in equation \eqref{eq:Xdefi}.
From here on, we will often suppress the explicit parameter dependence on $(\kappa,t_0)$ of $T(\rho),S,X$ and $S^*=S\setminus X$.

Firstly, as $X$ is, by definition, the zero locus of two polynomials, it is closed in $S$. Hence, $S^*$ is open in $S$. To prove the first part of the proposition, we study whether the gradient of $T(\rho)$ vanishes anywhere on the open subset $S^*$ of $S$.

We start by considering whether $S^*$ has any singularities in its affine part $S^*\cap \mathbb{C}^4$. The zero locus of the gradient of $T(\rho_1,\rho_2,\rho_3,\rho_4)$ is characterised by the linear equation
\begin{equation}\label{eq:gradientlocus}
    H\cdot (\rho_1,\rho_2,\rho_3,\rho_4)^T=0,
\end{equation}
where $H$ is the Hessian matrix of $T$, i.e.
\begin{equation}\label{eq:hessian}
    H=\begin{pmatrix}
    0 & T_{12} & T_{13} & T_{14}\\
    T_{12} &0  & T_{23} & T_{24}\\
    T_{13} & T_{23} & 0 & T_{34}\\
    T_{14} & T_{24} & T_{34} & 0\\
    \end{pmatrix}.
\end{equation}

We proceed to show that the determinant of $H$ is nonzero. This implies that equation \eqref{eq:gradientlocus} has only one solution $\underline{0}:=(0,0,0,0)\in X$, which does not lie in
$S^*$. In particular, $S^*$ has no singularities in its affine part. 

In fact, we will prove that the determinant of $H$ is given explicitly by
\begin{equation}\label{eq:Hdet}
    |H|=\kappa_0^{-2}\kappa_t^{2}\kappa_1^{2}\kappa_\infty^{2}\theta_q\left(\kappa_0^2,\kappa_t^2,\kappa_1^2,\kappa_\infty^2\right)^2\theta_q\left(\kappa_t\kappa_1 t_0,
    \kappa_t\kappa_1^{-1} t_0,
    \kappa_t^{-1}\kappa_1 t_0,
    \kappa_t^{-1}\kappa_1^{-1} t_0\right)^2,
\end{equation}
so that $|H|\neq 0$, due to the non-resonance conditions \eqref{eq:non_res}.

To this end, we first note that $|H|$ depends analytically on each of the parameters $\kappa_j\in \mathbb{C}^*$, $j=0,t,1,\infty$, and $t_0\in\mathbb{C}^*$. We begin by studying the dependence of the determinant on $\kappa_0$ and denote
\begin{equation*}
    h=h(\kappa_0):=|H|.
\end{equation*}
Since each of the entries of $H$ satisfies the $q$-difference equation
\begin{equation*}
	T_{ij}(q\,\kappa_0)=q^{-1}\kappa_0^{-2}T_{ij}(\kappa_0),
\end{equation*}
$1\leq i<j\leq 4$, we have
\begin{equation}\label{eq:difeq}
	h(q\,\kappa_0)=q^{-4}\kappa_0^{-8}h(\kappa_0).
\end{equation} 
It follows from Lemma \ref{lem:classificationtheta} that $h$ has precisely eight zeros, counting multiplicity, in $\{\kappa_0\in \mathbb{C}^*\}$, modulo $q^\mathbb{Z}$.
We further note the following helpful symmetries,
\begin{equation}\label{eq:hsym}
    h(\kappa_0^{-1})=h(\kappa_0),\qquad h'(\kappa_0^{-1})=-\kappa_0^{2}\, h'(\kappa_0).
\end{equation}

A direct calculation yields that $h$, evaluated at $\kappa_0=1$, formally factorises as
\begin{align*}
h(1)=
	\prod_{\epsilon_1,\epsilon_2\in\{\pm 1\}}\big[&
	+\kappa_\infty \theta_q(\kappa_t^2,\kappa_1^2,\kappa_\infty^{-1}t_0,\kappa_\infty t_0)\\
	&+
	\epsilon_1  \kappa_\infty \theta_q(\kappa_t\kappa_1\kappa_\infty^{-1},\kappa_t\kappa_1\kappa_\infty,\kappa_t \kappa_1^{-1}t_0,\kappa_1 \kappa_t^{-1}t_0)\\
	&+
	\epsilon_2  \kappa_t\kappa_1 \theta_q(\kappa_t^{-1}\kappa_1\kappa_\infty,\kappa_t\kappa_1^{-1}\kappa_\infty,\kappa_t \kappa_1 t_0,\kappa_1^{-1} \kappa_t^{-1}t_0)\big].
\end{align*}
The factor with $\epsilon_1=\epsilon_2=-1$ vanishes identically by the addition law for theta functions, hence $h(1)=0$. it furthermore follows from symmetries \eqref{eq:hsym} that $h'(1)=0$, so that $\kappa_0=1$ is at least a double zero of $h$.

An analogous computation shows that $\kappa_0=-1$ is at least a double zero of $h$.

Similarly, it follows that $\kappa_0=q^{\frac{1}{2}}$ is a zero of $h$. To show that it is at least a double zero, we take the derivative of equation \eqref{eq:difeq},
\begin{equation*}
	qh'(q\,\kappa_0)=q^{-4}\kappa_0^{-8}h'(\kappa_0)-8\,q^{-4}\kappa_0^{-9}h(\kappa_0).
\end{equation*} 
By evaluating this identity, and the second equation in \eqref{eq:hsym}, at $\kappa_0=q^{-\frac{1}{2}}$, it follows that 
$h'(q^{\frac{1}{2}})=0$ so that $\kappa_0=q^{\frac{1}{2}}$ is at least a double zero of $h$. The same statement follows analogously for $\kappa_0=-q^{\frac{1}{2}}$.

In conclusion, we have found four zeros of $h$, $\kappa_0=\pm 1,\pm q^{\frac{1}{2}}$, each at least of degree two. But $h$ is a degree $8$ theta function. It follows from this, and equation \eqref{eq:difeq}, that
\begin{equation*}
h=\kappa_0^{-2}\theta_q\left(\kappa_0^2\right)^2 \widetilde{h},
\end{equation*}
where $\widetilde{h}$ is a function independent of $\kappa_0$.

By following the same procedure with respect to the variables $\kappa_t,\kappa_1,\kappa_\infty$, we obtain
\begin{equation*}
	h=c\,\kappa_0^{-2}\kappa_t^{2}\kappa_1^{2}\kappa_\infty^{2}\theta_q\left(\kappa_0^2,\kappa_t^2,\kappa_1^2,\kappa_\infty^2\right)^2\theta_q\left(\kappa_t\kappa_1 t_0,
	\kappa_t\kappa_1^{-1} t_0,
	\kappa_t^{-1}\kappa_1 t_0,
	\kappa_t^{-1}\kappa_1^{-1} t_0\right)^2,
\end{equation*}
for some constant $c$ which may only depend on $t_0$ and $q$.

At this point, one simply evaluates both sides at $\kappa_0=\kappa_t=\kappa_1=\kappa_\infty=i$, to obtain $c=1$, which yields equation \eqref{eq:Hdet}.

We now return to the proof of the proposition. We have already established that $\mathcal{S}^*$ has no singularities in its affine part.
It remains to study whether $S^*$ has singularities with
one or more of their coordinates equal to $\infty$. Note that we only have to check the cases where one or two of their coordinates are equal to $\infty$, as points, with more than two coordinates equal to $\infty$, lie in $X$ and thus not in $S^*$.

Let us start by considering whether there are any singularities in
\begin{equation}\label{eq:surface_intersection}
    S^*\cap\{(\rho_1,\rho_2,\rho_3,\infty):\rho_{k}\in\mathbb{C}\text{ for }1\leq k\leq 3\}.
\end{equation}
To this end, we evaluate the gradient of
\begin{equation*}
F=\rho_4^yT\left(\rho_1^x,\rho_2^x,\rho_3^x,\frac{1}{\rho_4^y}\right)
  \end{equation*}
at $\rho_4^y=0$, yielding
\begin{equation*}
\nabla F|_{\rho_4^y=0}=(T_{14},T_{24},T_{34},T_{14}\rho_1^x+T_{24}\rho_2^x+T_{34}\rho_3^x)^T.
\end{equation*}
For this gradient to vanish, it is required that $T_{14}=T_{24}=T_{34}=0$, which cannot be realised without violating one of the non-resonance conditions \eqref{eq:non_res}. Therefore,  $S^*$ has no singularities with $\rho_4=\infty$ and the remaining coordinates finite. Applying the same argument in the three other cases, it follows that the manifold $S^*$ has no singularities with precisely one of their coordinates equal to $\infty$.

Next, we consider the existence of singularities on $S^*$ with two of their coordinates infinite.
Let us for example consider $\rho_3=\rho_4=\infty$ with $\rho_1$ and $\rho_2$ finite. Setting $\rho_3^y=\rho_4^y=0$ in
\begin{equation*}
\rho_1^y \rho_2^y \rho_3^y \rho_4^yT\left(\frac{\rho_1^x}{\rho_1^y},\frac{\rho_2^x}{\rho_2^y},\frac{\rho_3^x}{\rho_3^y},\frac{\rho_4^x}{\rho_4^y}\right)=0,
\end{equation*}
reduces it to
\begin{equation*}
    T_{34}\rho_1^y \rho_2^y\rho_3^x\rho_4^x=0.
\end{equation*}
Therefore,
\begin{equation}\label{eq:T34}
    \left\{\rho\in S^*:\rho_3=\rho_4=\infty\right\}=\begin{cases}
    \{(\rho_1,\rho_2,\infty,\infty):\rho_{1},\rho_2\in\mathbb{C}\} & \text{ if }T_{34}=0,\\
    \emptyset & \text{ otherwise.}
    \end{cases}
\end{equation}
In turn, $T_{34}=0$ if and only if $\kappa_0^{+1}\kappa_\infty t_0\in q^\mathbb{Z}$ or $\kappa_0^{-1}\kappa_\infty t_0\in q^\mathbb{Z}$. Thus $T_{34}\neq 0$ when the non-splitting conditions \eqref{eq:intro_irreducibleparameter} hold true.

More generally, if the non-splitting conditions \eqref{eq:intro_irreducibleparameter} hold true, then all of the coefficients $T_{ij}$, $1\leq i<j\leq 4$, are nonzero and consequently there are no points on $S^*$ with two coordinates equal to $\infty$. Thus we can conclude that $S^*$ is smooth when conditions the non-splitting conditions hold true.

Returning to the example above, i.e. $\rho_3=\rho_4=\infty$, under the assumption that $T_{34}=0$, evaluation of the gradient
of
\begin{equation*}
F=\rho_3^y\rho_4^yT\left(\rho_1^x,\rho_2^x,\frac{1}{\rho_3^y},\frac{1}{\rho_4^y}\right)
  \end{equation*}
at $\rho_3^y=\rho_4^y=0$, yields
\begin{equation*}
\nabla F|_{\rho_3^y,\rho_4^y=0}=(0,0,T_{14}\rho_1^x+T_{24}\rho_2^x,T_{13}\rho_1^x+T_{23}\rho_2^x)^T,
\end{equation*}
which vanishes at $\rho_1^x=\rho_2^x=0$, and only at this point, as
\begin{equation*}
    \begin{vmatrix}
    T_{14} & T_{24}\\
    T_{13} & T_{23}
    \end{vmatrix}^2=|H|\neq 0,
\end{equation*}
where $H$ the Hessian of $T$ defined in equation \eqref{eq:hessian}.

The determinant of the Hessian of $F$ at the point $(\rho_1^x,\rho_2^x,\rho_3^y,\rho_4^y)=\underline{0}$ equals $|H|$, which is nonzero, and thus this point is a non-degenerate saddle point of $F$. In particular, $\{F=0\}$ has an ordinary double point singularity at $\underline{0}$, by the complex Morse lemma. Therefore, the manifold $S^*$ has an ordinary double point singularity at $\rho=(0,0,\infty,\infty)$, when $\kappa_0^{+1}\kappa_\infty t_0\in q^\mathbb{Z}$ or $\kappa_0^{-1}\kappa_\infty t_0\in q^\mathbb{Z}$.

More generally, if some of the non-splitting conditions \eqref{eq:intro_irreducibleparameter} are violated, then the intersection $S_{sing}^*$  of $\Theta$ and
$S^*$ is non-empty, and at each point in $S_{sing}^*$, $S^*$ has an ordinary double point singularity and $S^*$ is smooth elsewhere. Otherwise, $S_{sing}^*$ is empty and in that case we have already shown that $S^*$ has no singularities.
This completes the proof of the proposition.
\end{proof}

We now proceed to prove Theorem \ref{thm:main_smooth} by using Proposition \ref{prop:smooth}.

\begin{proof}[Proof of Theorem \ref{thm:main_smooth}]
The first part of the proof is to show that the smoothness properties of the $3$-manifold $S^*(\kappa,t_0)$, established in Proposition \ref{prop:smooth}, are preserved by the quotient map to $\mathcal{S}^*(\kappa,t_0)$. The second step will be to translate these results to the monodromy manifold $\mathcal{M}(\kappa,t_0)$.

Recall that $\mathcal{S}^*(\kappa,t_0)$ is the zero set of the polynomial $T(\rho;\kappa,t_0)$, given in  Definition \ref{def:modulispace}. Due to Proposition \ref{prop:smooth}, it can be singular only at points in the finite set $\Theta$, given in equation \eqref{eq:theta_set}. Recall also that $S^*_{sing}$ refers to the subset of singular points lying on the 3-manifold $S^*(\kappa,t_0)$. Consider the smooth complex $3$-manifold
\begin{equation*}
    \widetilde{S}^*(\kappa,t_0)=S^*(\kappa,t_0)\setminus S_{sing}.
\end{equation*}
We denote the image of $\Theta$ under the quotient map $\iota_{S^*}$ by $\widehat{\Theta}$, so that the image of $\widetilde{S}^*(\kappa,t_0)$ under $\iota_{S^*}$ is given by
\begin{equation*}
    \widetilde{\mathcal{S}}^*(\kappa,t_0)=\mathcal{S}^*(\kappa,t_0)\setminus \mathcal{S}_{sing}^*,\quad  \mathcal{S}_{sing}^*:=\mathcal{S}^*(\kappa,t_0)\cap \widehat{\Theta}.
\end{equation*}
As (non-zero) scalar multiplication acts smoothly on $\widetilde{S}^*(\kappa,t_0)$, and no element of $\widetilde{S}^*(\kappa,t_0)$ is invariant under this operation, it follows that $\widetilde{\mathcal{S}}^*(\kappa,t_0)$ is a smooth complex surface.

Now, consider a point $\rho_0\in S_{sing}$. Since this point is invariant under the smooth action $\rho \mapsto c\,\rho$, $c\in\mathbb C^*$, it is easy to see that the quotient space $\mathcal{S}^*$ is not Hausdorff near its image $[\rho_0]$. In fact, near points in $S_{sing}$, the space $\mathcal{S}^*$ even fails to locally be a $T_1$ space. In particular, the smooth structure on $\widetilde{\mathcal{S}}^*(\kappa,t_0)$ cannot be extended to include points in $S_{sing}$.

To complete the proof of the theorem, we translate the results on $\mathcal{S}^*(\kappa,t_0)$ to $\mathcal{M}(\kappa,t_0)$ via the mapping $\mathcal{P}$. To this end, recall that $\mathcal{P}$ maps the finite set $\mathcal{M}_{sing}$ onto $\mathcal{S}_{sing}$. 

We have shown that $\mathcal{S}^*(\kappa,t_0)\setminus \mathcal{S}_{sing}$ is a smooth complex surface. Hence $\mathcal{M}(\kappa,t_0)\setminus \mathcal{M}_{sing}$ is a smooth complex surface.  Furthermore, elements of $\mathcal{M}_{sing}$ form singularities on the monodromy manifold, as points in $\mathcal{S}_{sing}$ are singularities on $\mathcal{S}^*(\kappa,t_0)$.

Finally, we note that $\mathcal{M}_{sing}$ is non-empty if and only if $\mathcal{S}_{sing}$ is non-empty, and the latter holds true if and only if some of the non-splitting conditions are violated, by the equivalence in Proposition \ref{prop:smooth}. The theorem follows.
\end{proof}

\subsection{The monodromy manifold as an algebraic surface}\label{subsection:embeddings}
In this section, we prove Theorem \ref{thm:main_affine}, which allows us to identify the monodromy manifold with an affine algebraic surface embedded in $\mathbb{C}^6$. Furthermore, we describe how the monodromy manifold can also be embedded in 
 $(\mathbb{P}^1)^3$.

\begin{proof}[Proof of Theorem \ref{thm:main_affine}]
The mapping $\Phi_{\mathcal M}$ is composed of two parts: $\mathcal P: \mathcal M\to \mathcal S^*$ and $\Phi:\mathcal S^*\to \mathcal F$. The mapping $\mathcal{P}$ is a (topological) isomorphism   due to theorem \ref{thm:main_moduli}. Hence, it only remains to show that the mapping $\Phi$ is an isomorphism. 
To prove this, we construct a continuous inverse, which we denote by $\Psi$, of $\Phi$.

We start by recalling that $\mathcal S^*$, defined in equation \eqref{eq:defsstar}, is locally described by coordinates $[\rho]$ in the ambient space $(\mathbb{P}^1)^4/\mathbb{C}^*$. Similarly,  $\mathcal F$ is described by the coordinates $\eta_{ij}$, $1\le i<j \le 4$, in $\mathbb{C}^6$. 

The mapping $\Phi$ is a continuous mapping from $\mathcal S^*$ to $\mathcal F$, described by equation \eqref{eq:eta_defi} with respect to the above coordinates.
In particular, note that, due to equation \eqref{eq:eta_defi}, for any labeling $\{i,j,k,l\}=\{1,2,3,4\}$, we have
\begin{equation}\label{eq:crucial_chi}
\eta_{ij}=0 \iff \rho_i=0\text{ or }\rho_j=0\text{ or }\rho_k=\infty\text{ or }\rho_l=\infty.
\end{equation}

This means that $\Phi$ maps the open subdomain $\mathcal{S}_0\subseteq\mathcal S^*$, given by
\begin{equation*}
    \mathcal{S}_0:=\{[\rho]\in\mathcal{S}^*:\rho_k\neq 0,\infty\text{ for $1\leq k\leq 4$}\},
\end{equation*}
into the subspace
\begin{equation*}
    \mathcal{F}_0:=\mathcal{F}\cap (\mathbb{C}^*)^6,
\end{equation*}
of the co-domain.

We proceed by defining an inverse of $\Phi$ on this subdomain and co-domain, and subsequently extending this inverse to one on the full domain.

The relevant mapping on $\mathcal{F}_0$ is the following,
\begin{equation*}
    \Psi|_{\mathcal{F}_0}:\mathcal{F}_0\rightarrow (\mathbb{P}^1)^4/\mathbb{C}^*, \eta\rightarrow \left[\left(\frac{T_{34}\eta_{13}}{T_{13}\eta_{34}},\frac{T_{34}\eta_{23}}{T_{23}\eta_{34}},\frac{T_{24}\eta_{23}}{T_{23}\eta_{24}},1\right)\right].
\end{equation*}
which we now show to be an inverse of $\Phi|_{\mathcal{S}_0}$. By equations \eqref{eq:eta_equationsa}, \eqref{eq:eta_equationsc} and \eqref{eq:eta_equationsd}, the image of $\Psi|_{\mathcal{F}_0}$ is contained in $\mathcal{S}$. Furthermore, due to \eqref{eq:eta_equationsb}, any point in the image cannot lie in $\mathcal{X}$. It thus follows that the image of $\Psi|_{\mathcal{F}_0}$ is contained in $\mathcal{S}^*$. Furthermore, as $\mathcal{F}_0$ by definition excludes any of the $\eta$-coordinates to equal zero, $\Psi|_{\mathcal{F}_0}$ maps $\mathcal{F}_0$ into  $\mathcal{S}_0$. Finally, note that, for any point $\rho\in \mathcal{S}_0$,
\begin{align*}
    \Psi|_{\mathcal{F}_0}\circ \Phi|_{\mathcal{S}_0}([\rho])&=\left[\left(\frac{T_{34}\eta_{13}}{T_{13}\eta_{34}},\frac{T_{34}\eta_{23}}{T_{23}\eta_{34}},\frac{T_{24}\eta_{23}}{T_{23}\eta_{24}},1\right)\right]\\
    &=[(\rho_1/\rho_4,\rho_2/\rho_4,\rho_3/\rho_4,1)]\\
    &=[(\rho_1,\rho_2,\rho_3,\rho_4)],
\end{align*}
where, in the second equality, we used equation \eqref{eq:eta_defi}.

Similarly, it can be seen that $\Phi|_{\mathcal{S}_0}\circ \Psi|_{\mathcal{F}_0}$ is the identity map on $\mathcal{F}_0$. It follows that $\Psi|_{\mathcal{F}_0}$ is a (continuous) inverse of $\Phi|_{\mathcal{S}_0}$.

The set $\mathcal{S}_0$ is an open dense subset of the domain $\mathcal{S}$ and, similarly, $\mathcal{F}_0$ is an open dense subset of the co-domain. It remains to deal with the special cases where one or more of the $\rho_k$, $1\leq k\leq 4$, is zero or infinite, and equivalently one or more of the $\eta_{ij}$, $1\leq i<j\leq 4$ is zero.

We handle each of these cases separately. The cases are described by 
\begin{equation}\label{eq:boundary_components}
\begin{cases}
    &\mathcal{S}_i^0:=\{[\rho]\in\mathcal{S}^*:\rho_i=0,\rho_k\notin \{0,\infty\}\text{ for }k\neq i\},\\
    &\mathcal{S}_j^\infty:=\{[\rho]\in\mathcal{S}^*:\rho_j=\infty,\rho_k\notin \{0,\infty\}\text{ for }k\neq j\},\\
    &\mathcal{S}_{i,j}^{0,\infty}:=\{[\rho]\in\mathcal{S}^*:\rho_i=0,\rho_j=\infty,\rho_k\notin \{0,\infty\}\text{ for }k\neq i,j\},
\end{cases}
\end{equation}
for $1\leq i,j\leq 4$ with $i\neq j$. Note that $\mathcal{S}_{i,j}^{0,\infty}$ provides the boundaries of $\mathcal{S}_i^0$ and $\mathcal{S}_j^\infty$. Since no point on $\mathcal S_0$ can have two or more components all zero or all infinite, the sets defined in equation \eqref{eq:boundary_components} glue together to provide all the boundaries or limit sets of $\mathcal S_0$ within $\mathcal{S}^*$.

We now express the surface $\mathcal{S}^*$ as a disjoint union of all of these cases with $\mathcal S_0$, that is,
\begin{align*}
    \mathcal{S}^*=&\mathcal{S}_0\sqcup \mathcal{S}_1^0\sqcup \mathcal{S}_2^0\sqcup \mathcal{S}_3^0\sqcup \mathcal{S}_4^0\\
    &\sqcup \mathcal{S}_1^\infty\sqcup \mathcal{S}_2^\infty\sqcup \mathcal{S}_3^\infty\sqcup \mathcal{S}_4^\infty\\
    &\sqcup \mathcal{S}_{1,2}^{0,\infty}\sqcup \mathcal{S}_{1,3}^{0,\infty}\sqcup \mathcal{S}_{1,4}^{0,\infty}\sqcup\mathcal{S}_{2,1}^{0,\infty}\sqcup\ldots \sqcup \mathcal{S}_{4,2}^{0,\infty}\sqcup \mathcal{S}_{4,3}^{0,\infty},
\end{align*}
where the last line indicates disjoint union of all $\mathcal{S}_{i,j}^{0,\infty}$, $1\leq i,j\leq 4$, with $i\neq j$.

We correspondingly decompose the codomain $\mathcal{F}$ into disjoint components. 
Motivated by equation \eqref{eq:crucial_chi}, we define these components by
\begin{equation}\label{eq:boundary_componentsF}
\begin{cases}
    &\mathcal{F}_i^0 :=\{\eta\in\mathcal{F}: i\in\{k,l\}\iff \eta_{kl}=0 ,\text{ for }1\leq k<l\leq 4\},\\
    &\mathcal{F}_j^\infty:=\{\eta\in\mathcal{F}: j\notin\{k,l\}\iff \eta_{kl}=0,\text{ for }1\leq k<l\leq 4\},\\
    &\mathcal{F}_{i,j}^{0,\infty}:=\{\eta\in\mathcal{F}: i\in\{k,l\}\text{ and }j\notin\{k,l\} \iff \eta_{kl}=0,\\
    &\qquad \qquad\quad  \text{ for }1\leq k<l\leq 4\}.
\end{cases}
\end{equation}
Equations \eqref{eq:eta_equations} imply that any element $\eta$ of $\mathcal{F}$ has either zero, three or four components equal to zero and the components in \eqref{eq:boundary_componentsF} indeed cover all of $\mathcal{F}\setminus \mathcal{F}_0$.

Then, inspired by \eqref{eq:crucial_chi}, we correspondingly decompose
$\mathcal{F}$ as a disjoint union,
\begin{align*}
    \mathcal{F}=&\mathcal{F}_0\sqcup \mathcal{F}_1^0\sqcup \mathcal{F}_2^0\sqcup \mathcal{F}_3^0\sqcup \mathcal{F}_4^0\\
    &\sqcup \mathcal{F}_1^\infty\sqcup \mathcal{F}_2^\infty\sqcup \mathcal{F}_3^\infty\sqcup \mathcal{F}_4^\infty\\
    &\sqcup \mathcal{F}_{1,2}^{0,\infty}\sqcup \mathcal{F}_{1,3}^{0,\infty}\sqcup \mathcal{F}_{1,4}^{0,\infty}\sqcup\mathcal{F}_{2,1}^{0,\infty}\sqcup\ldots \sqcup \mathcal{F}_{4,2}^{0,\infty}\sqcup \mathcal{F}_{4,3}^{0,\infty},
\end{align*}
Due to \eqref{eq:crucial_chi}, $\Phi$ maps each component in the decomposition of $\mathcal{S}^*$ into the corresponding component in the decomposition of $\mathcal{F}$.
We extend $\Psi$ to a global inverse of $\Phi$ on $\mathcal{F}$, by locally defining it on each of the components in the decomposition of $\mathcal{F}$. The arguments for each of the three types of components are similar, and so we give the details for one of each type below to illustrate the details.

For example, for
\begin{equation*}
    \mathcal{F}_1^0=\{\eta\in \mathcal{F}: \eta_{12}=\eta_{13}=\eta_{14}=0\text{ and }\eta_{23},\eta_{24},\eta_{34}\neq 0\},
\end{equation*}
we set
\begin{equation*}
\Psi|_{\mathcal{F}_1^0}:\mathcal{F}_1^0\rightarrow \mathcal{S}_1^0,\eta\mapsto \left[\left(0,\frac{T_{34}\eta_{23}}{T_{23}\eta_{34}},\frac{T_{24}\eta_{23}}{T_{23}\eta_{24}},1\right)\right],
\end{equation*}
which defines an inverse of $\Phi|_{\mathcal{S}_1^0}$. 
Similarly, for
\begin{equation*}
    \mathcal{F}_1^\infty=\{\eta\in \mathcal{F}:\eta_{23}=\eta_{24}=\eta_{34}=0\text{ and } \eta_{12},\eta_{13},\eta_{14}\neq 0\},
\end{equation*}
we define
\begin{equation*}
\Psi|_{\mathcal{F}_1^\infty}:\mathcal{F}_1^\infty\rightarrow \mathcal{S}_1^\infty,\eta\mapsto \left[\left(\infty,\frac{\eta_{12}}{T_{12}},\frac{\eta_{13}}{T_{13}},\frac{\eta_{14}}{T_{14}}\right)\right],
\end{equation*}
which is an inverse of $\Phi|_{\mathcal{S}_1^\infty}$.
For the third and final example
\begin{equation*}
    \mathcal{F}_{1,2}^{0,\infty}=\{\eta\in \mathcal{F}:\eta_{12}=\eta_{13}=\eta_{14}=\eta_{34}=0\text{ and } \eta_{23},\eta_{24}\neq 0\},
\end{equation*}
we take
\begin{equation*}
\Psi|_{\mathcal{F}_{1,2}^{0,\infty}}:\mathcal{F}_{1,2}^{0,\infty}\rightarrow \mathcal{S}_{1,2}^{0,\infty},\eta\mapsto \left[\left(0,\infty,\frac{\eta_{23}}{T_{23}},\frac{\eta_{24}}{T_{24}}\right)\right],
\end{equation*}
which is an inverse of $\Phi|_{\mathcal{S}_{1,2}^{0,\infty}}$. 

This extends $\Psi$ to a global inverse of $\Phi$ on $\mathcal{F}$. $\Psi$ is continuous on each of the separate components and it is straightforward to check that its continuations to common boundary points of different components agree with each other.
\end{proof}

We finish this section by describing an embedding of the monodromy manifold into $(\mathbb{P}^1)^3$. We assume that the non-splitting conditions \eqref{eq:intro_irreducibleparameter} hold true. In particular, all the coefficients of the polynomial $T(p:\kappa,t_0)$ are nonzero and, therefore, there are no points $\rho\in S^*(\kappa,t_0)$ with two or more components all zero or all infinite. Thus,
\begin{equation*}
    \rho_{ij}=\frac{\rho_i}{\rho_j}\in\mathbb{P}^1,\quad (1\leq i<j\leq 4),
\end{equation*}
form six well-defined coordinates on the surface $\mathcal{S}^*(\kappa,t_0)$ and thus also on the monodromy manifold $\mathcal{M}(\kappa,t_0)$.

 Ohyama et al. \cite{ohyamaramissualoy} study the $q\Psix$ monodromy manifold using these coordinates\footnote{To be precise, in \cite{ohyamaramissualoy}[\textsection 5.1.1] the `dual' coordinates $\Pi_{ij}:=\frac{\rho_i'}{\rho_j'}\in\mathbb{P}^1$ are used, where $\rho_k'=\pi(C(x_k)^T)$ for $1\leq k\leq 4$. These coordinates are bi-rationally equivalent to the $\rho_{ij}$ coordinates and we note that $(\rho_1',\rho_2',\rho_3',\rho_4')$ lies on the threefold $S^*(\kappa_\infty^{-1},\kappa_t,\kappa_1,\kappa_0^{-1},t_0)$.}.
 Theorem \ref{thm:main_moduli} yields explicit algebraic relations among them. For example, $\rho_{12},\rho_{23}$ and $\rho_{34}$ are related by
 \begin{equation}\label{eq:coordinatesrhoratios}
T_{12}\rho_{12}\rho_{23}^2
+T_{13}\rho_{12}\rho_{23}
+T_{14}\rho_{12}\rho_{23}\rho_{34}^{-1}
+T_{23}\rho_{23}
+T_{24}\rho_{23}\rho_{34}^{-1}
+T_{34}\rho_{34}^{-1}=0.
 \end{equation}
Analogously to the proof of Theorem \ref{thm:main_affine}, we can show that these three coordinates yield an embedding of the monodromy manifold into $(\mathbb{P}^1)^3$,
\begin{equation*}
\mathcal{M}(\kappa,t_0)\rightarrow (\mathbb{P}^1)^3,[C(z)]\mapsto (\rho_{12},\rho_{23},\rho_{34}),
\end{equation*}
 with range given by the surface \eqref{eq:coordinatesrhoratios} minus a curve. This curve is defined by the intersection of \eqref{eq:coordinatesrhoratios} as $\kappa_0$ varies over $\mathbb{C}^*$.
 
 \begin{remark}\label{remark:compatible}
Assuming the non-splitting conditions \eqref{eq:intro_irreducibleparameter}, the six coordinates $\rho_{ij}$, $1\leq i,j\leq 4$, are analytic rational functions from $\mathcal{F}(\kappa,t_0)$ to $\mathbb{CP}^1$, which together embed the surface into $(\mathbb{CP}^1)^6$. The same statements holds true for these coordinates, as functions on the monodromy manifold $\mathcal{M}(\kappa,t_0)$, with respect to the structure of a complex algebraic variety defined in Ohyama et al. \cite{ohyamaramissualoy}. It follows that this structure is compatible with the one induced by Theorem \ref{thm:main_affine}.
 \end{remark}

\section{Conclusion}\label{sec:conc}
In this paper, we studied the $q\Psix$ equation through its associated linear problem. Assuming non-resonant parameter conditions, we defined the corresponding Riemann-Hilbert problem, which captures the general solution of $q\Psix$. This problem was shown to be solvable for irreducible monodromy, leading to a one-to-one correspondence between solutions of $q\Psix$ and points on the corresponding monodromy manifold, when the non-splitting conditions are satisfied.

In turn, we constructed an explicit embedding of the monodromy manifold into $(\mathbb{CP}^1)^4/\mathbb{C}^*$, whose image is described by the zero locus of a single quadratic polynomial, minus a curve. This allowed us to show that the monodromy manifold is a smooth complex surface, when the non-splitting conditions hold true. We further proved that it can be identified with an affine algebraic surface, under the same assumptions. This surface can be described as the intersection of two quadrics in $\mathbb{C}^4$ and its projective completion is thus a Segre surface.

The results of this paper suggests a possible framework for tackling several open questions. These include, for example, the classification of algebraic or symmetric solutions of $q\Psix$, the construction of (classes of) special transcendental solutions via the geometry of the monodromy manifold, and the derivation of solutions with distinctive (e.g. bounded) global asymptotic behaviours.


\end{document}